\documentclass[a4paper,12pt]{amsart}
\usepackage{amssymb,amsthm,amsmath,color,url}
\usepackage[margin=3cm]{geometry}

\usepackage[utf8]{inputenc}
\usepackage[T1]{fontenc}
\usepackage{amsmath, amssymb, amsthm}
\usepackage{hyperref}
\usepackage{hhline}
\usepackage{array}
\usepackage{tikz}
\hypersetup{
  colorlinks   = true,
  linkcolor = red!70!black,
     citecolor    = green!50!black
   }

   \usepackage[capitalise, compress, nameinlink, noabbrev]{cleveref}
\crefname{lab}{}{}
\creflabelformat{lab}{#2(R#1)#3}
\crefname{sat}{}{}
\creflabelformat{sat}{#2(S#1)#3}
\crefname{pr}{}{}
\creflabelformat{pr}{#2(P#1)#3}

\makeatletter
\def\namedlabel#1#2{\begingroup
   \def\@currentlabel{#2}%
   \label{#1}\endgroup
}
\makeatother

\usepackage{thmtools}
\usepackage{thm-restate}
\usepackage{paralist}

\usepackage{tikz}
\usetikzlibrary{shadows}
\usetikzlibrary{backgrounds}
\usetikzlibrary{arrows}
\usetikzlibrary{patterns}
\usetikzlibrary{calc}
\usetikzlibrary{shapes}
\usetikzlibrary{positioning}
\usetikzlibrary{math}
\usetikzlibrary{external}
\usetikzlibrary{decorations.pathreplacing}

\declaretheorem[name=Theorem, numberwithin=section]{theorem}
\declaretheorem[name=Lemma, sibling=theorem]{lemma}
\declaretheorem[name=Proposition, sibling=theorem]{proposition}

\declaretheorem[name=Corollary, sibling=theorem]{corollary}

\declaretheorem[name=Problem, sibling=theorem]{problem}

\declaretheorem[name=Claim, sibling=theorem]{claim}
\declaretheorem[name=Claim, numbered=no]{claim*}

\declaretheorem[name=Remark, style=remark, sibling=theorem]{remark}

\def\cqedsymbol{\ifmmode$\lrcorner$\else{\unskip\nobreak\hfil
\penalty50\hskip1em\null\nobreak\hfil$\lrcorner$
\parfillskip=0pt\finalhyphendemerits=0\endgraf}\fi} 

\newcommand{\tab}{\mathsf{Table}}

\newcommand{\p}{\mathcal{P}}
\newcommand{\np}{\overline{\p}}

\newcommand{\C}{\mathcal{C}}
\newcommand{\F}{\mathcal{F}}

\newcommand{\ham}{\textsc{HamiltonianCycle}}

\let\le\leqslant
\let\ge\geqslant

\begin{document}

\title{Reductions in local certification}

\author[L.~Esperet]{Louis Esperet}
\address[L.~Esperet]{Univ.\ Grenoble Alpes, CNRS, Laboratoire G-SCOP,
  Grenoble, France}
\email{louis.esperet@grenoble-inp.fr}

\author[S.~Zeitoun]{Sébastien Zeitoun}
\address[S.~Zeitoun]{CNRS, INSA Lyon, UCBL, LIRIS, UMR5205, F-69622 Villeurbanne, France}
\email{sebastien.zeitoun@univ-lyon1.fr}

\thanks{The authors are partially supported by the French ANR Project
  TWIN-WIDTH
  (ANR-21-CE48-0014-01), and ENEDISC (ANR-24-CE48-7768-01) by LabEx
  PERSYVAL-lab (ANR-11-LABX-0025). A preliminary version of this paper
  appeared in the proceedings of the 51st International Workshop on
  Graph-Theoretic Concepts in Computer Science (WG 2025).}

\date{}

\begin{abstract}
	Local certification is a topic originating from distributed
        computing, where a prover tries to convince the vertices of a
        graph $G$ that $G$ satisfies some property~$\p$. To convince
        the vertices, the prover gives a small piece of information,
        called certificate, to each vertex, and the vertices then
        decide whether the property $\p$ is satisfied by just looking at
        their certificate and the certificates of their neighbors.
        When studying a property~$\p$ in the perspective of local
        certification, the aim is to find the optimal size of the
        certificates needed to certify~$\p$, which can be viewed as a
        measure of the local complexity of~$\p$.
	
	A certification scheme is considered to be efficient
        if the size of the certificates is polylogarithmic in the
        number of vertices. While there have been a number of
        meta-theorems providing efficient certification schemes for
        general graph classes, the proofs of the lower
        bounds on the size of the certificates are usually very
        problem-dependent.
        
	In this work, we introduce a notion of hardness reduction
        in local certification, and show that we can transfer a lower
        bound on the certificates for a property~$\p$ to a lower bound
        for another property~$\p'$, via a (local)
        hardness reduction from~$\p$ to $\p'$. We then give a number
        of applications in which we obtain polynomial lower bounds for
        many classical properties using such reductions.
\end{abstract}

\maketitle

\section{Introduction}\label{sec:intro}

The main focus of the paper is on \emph{local certification}, which is
a field at the intersection of distributed computing, graph theory,
and complexity studying how the vertices of a graph $G$ can be convinced
that $G$ satisfies some given property $\mathcal{P}$, even if each
vertex of $G$ only has a very local view of the graph (for
instance, if it can only interact with its neighbors). More precisely, a \emph{proof labeling scheme} for a property
$\mathcal{P}$ in a graph $G$ works as follows:
\begin{enumerate}[(1)]
 \item a prover assigns (short)
   certificates to the vertices of the graph $G$;
   \item each vertex $v$  checks its \emph{local view} (its certificate and the certificates of its neighbors),
  \item based only
  on its local view, each vertex must decide whether $G$ satisfies the given
  property $\mathcal{P}$.
\end{enumerate}

Note that there is no requirement on the
  computational power of the prover and the vertices. If the graph indeed satisfies $\mathcal{P}$, all vertices
  must accept the instance. If the graph does not satisfy
  $\mathcal{P}$, then for any possible assignment of
  certificates, at least one vertex must
  reject the instance. For $n\in \mathbb{N}$, the \emph{local
    complexity} of  $\mathcal{P}$ is the minimum $f(n)$ such that for
  every $n$-vertex graph $G$, the property $\mathcal{P}$ can be
  locally certified in $G$ with certificates of at most $f(n)$ bits per
  vertex. We can also study more generally the local complexity of a
  property  $\mathcal{P}$ within a graph class $\mathcal{C}$, by which
  we mean that the minimum above is taken over all $n$-vertex graphs
  from~$\mathcal{C}$.  
  
Note that  every property has local complexity $O(n^2)$: it suffices to write the adjacency matrix of the graph in the certificate of each vertex. By doing so, each vertex knows the whole graph and is able to decide whether the property is satisfied or not~\cite{GoosS16}.
On the other hand, a counting argument shows that some properties involving symmetries have local complexity $\Omega(n^2)$. The original work of Göös and Suomela \cite{GoosS16} identified three natural ranges of local complexity for graph classes:

\begin{itemize}
\item $\Theta(1)$: this includes $k$-colorability for fixed $k$, and in
  particular bipartiteness;
\item $\Theta(\log n)$: this includes non-bipartiteness, 
  acyclicity, planarity \cite{planar} and more generally bounded genus \cite{genus,EL};
  \item $\Theta(\text{poly}(n))$: this includes non-3-colorability
    \cite{GoosS16}, unit-distance graphs, unit-disk graphs, and 1-planar
    graphs \cite{DELMR24}, as well as problems involving
    symmetries~\cite{GoosS16}.
  \end{itemize}

  Many of the results mentioned above (either lower bounds or upper bounds) are proved using rather ad-hoc
  arguments that do not seem to be applicable to a large number of
  graph properties at once. On the other hand, a number of positive
  results (upper bounds on the local complexity) have been obtained
  via so-called ``meta-theorems'', that is theorems applicable to a
  wide range of problems and graph classes. 
  It was proved that any graph class of bounded treewidth
  which is expressible in monadic second order
  logic has local complexity $O(\log^2 n)$ \cite{tw}. This implies in particular that
  for any fixed $k$, the of treewidth at most $k$ has
  local complexity $O(\log^2 n)$. Similar meta-theorems involving
  graph properties expressible in first-order logic were proved
  for graphs of bounded treedepth \cite{BFT2} and graphs of bounded
  cliquewidth 
  \cite{FM0RT23}. Another type of meta-theorem was obtained in
  \cite{NPY20}: any property that can be decided in linear time has an
  \emph{interactive} proof labeling scheme with compact certificates
  (in this setting several rounds of interactions are allowed between the prover
  and the vertices).

  \medskip

Proof labeling schemes in graphs are considered to be
efficient if their complexity is polylogarithmic in the   size of the
graph. The class of all properties whose local complexity is polylogarithmic in the   size of the
graph can be seen as a distributed analogue of the class
\textsf{NP}. However, a major feature of \textsf{NP} has so far been
missing in the toolbox of local certification: the possibility of doing
hardness reductions. The purpose of this paper is to introduce a
framework allowing to use classical \textsf{NP}-hardness reductions to prove
automatic lower bounds on the local complexity of graph
properties. Our techniques apply whenever the reduction is
sufficiently local. In particular, using the
result of Göös and Suomela \cite{GoosS16} stating that
non-3-colorability has local complexity $\Omega(n^2/\log n)$, we prove
that a number of classical \textsf{coNP}-hard problems have polynomial local
complexity. 
	Table~\ref{fig:sum} sums up our lower bounds.
	
	\begin{table}[h!]
		\centering
		\begin{tabular}{ | >{\centering\arraybackslash}m{11em} | >{\centering\arraybackslash}m{11.5em} | >{\centering\arraybackslash}m{6em} | >{\centering\arraybackslash}m{4.5em} | }
			\hline
			Property $\p$ & Class of graphs $\mathcal{C}$ & Lower bound for $\p$ in $\mathcal{C}$ & Theorem \\
			\hhline{====}
			
			Non-3-colorability & Planar graphs
                                             with~$\Delta\le 4$ & $\Omega\left(\frac{\sqrt{n}}{\log^2 n}\right)$ & \ref{thm:non3coldegmax4_lineaireplanaire} \\ \hline
			Non-$k$-colorability ($k \geqslant 4$) &
                                                                 General graphs & $\Omega\left(\frac{n^2}{\log n}\right)$ &  \ref{thm:nonkcol_quadratique} \\ \hline
			Non-$k$-colorability ($k \geqslant 3$) &
                                                                 Graphs with $\Delta\le k + \lceil\sqrt{k}\rceil-1$ & $\Omega\left(\frac{n}{\log n}\right)$ & \ref{thm:nonkcol_lineaire} \\ \hline
			
			Domatic number at~most~$k$ ($k \geqslant 2$) & Graphs
                                                     with~$\Delta\le 5k^2 - 4k$ & $\Omega\left(\frac{n}{\log n}\right)$ & \ref{thm:domatic} \\ \hline
			Non-existence of a cubic subgraph & Graphs
                                                            with~$\Delta\le
                                                            7$ & $\Omega\left(\frac{n}{\log n}\right)$ & \ref{thm:cubic_subgraph} \\ \hline
			Non-existence of a partition in $k$ acyclic subgraphs ($k \geqslant 3$) & General graphs & $\Omega\left(\frac{n}{\log n}\right)$ &\ref{thm:acyclic_partition} \\ \hline
			Existence of a monochromatic triangle in every $2$-edge-coloring & Graphs with $\Delta \leqslant 40$ & $\Omega\left(\frac{n}{\log n}\right)$ & \ref{thm:monochromatic_triangle} \\ \hline
			
			Non-existence of a Hamiltonian cycle & Graphs with $\Delta \leqslant 4$ & $\Omega\left(\frac{\sqrt{n}}{\log n}\right)$ & \ref{thm:hamiltonian} \\ \hline
			Chromatic index equal~to~$\Delta+1$ & Cubic graphs & $\Omega\left(\frac{n}{\log n}\right)$ & \ref{thm:coloration aretes} \\ \hline
		\end{tabular}
	\caption{Summary of our lower bounds. In this table,~$\Delta$ denotes the maximum degree.}
	\label{fig:sum}
	\end{table}

        Note that most of the lower bounds in Table~\ref{fig:sum} hold in a bounded-degree graph class (so in general graphs too).
        However, in a bounded-degree graph class, every property has
        local complexity $O(n\log n)$. Indeed, $O(n \log n)$ bits are sufficient to encode all the edges of a bounded-degree graph (using adjacency lists), and we can just write the description of the whole graph in the certificate of every vertex.
        Thus, most of the results above in bounded-degree classes are optimal (within a polylogarithmic factor). The two exceptions are
        the local complexity of non-3-colorability in planar graphs of maximum degree~4, and the local complexity of non-hamiltonicity in graphs of maximum degree~4. We believe
        that in both cases the local complexity should be linear.
        For the non-existence of a partition into $k$ acyclic subgraphs,
        we only have a quasi-linear lower bound, which holds in
        general graphs. This is probably also not optimal (because in general graphs the universal upper bound is $O(n^2)$).
        Finally, note also that, except for $k$-colorability, our
        quasi-linear lower bounds which hold in bounded-degree graph classes remain quasi-linear in general graphs. Again, this is probably not optimal: an interesting open problem would be to determine whether each of them can be improved to quadratic in general graphs.

        \medskip

        As alluded to above, in the distributed world ``efficient''
        usually means logarithmic or polylogarithmic. From this point
        of view, our results show that all these problems have very
        high local complexity. On the other hand, it turns out that the complement
        problems all have very low local complexity (either constant
        or logarithmic).

        \subsection*{Organization of the paper} We start with some
        preliminary definitions on graphs and local certification in Section
        \ref{sec:prel}. In Section \ref{sec:3col}, we describe the
        lower bound on the local complexity of non-3-colorability
        by G\"o\"os and Suomela \cite{GoosS16}, and explain how to
        deduce lowers bounds in the more specific cases  of graphs with
        maximum degree~at most 4, and planar graphs. These lower
        bounds will be crucial in the remainder of the paper, as all
        problems we consider will be reduced to one of these problems
        (directly or indirectly).

        In Section \ref{sec:example}, we then present as a motivating example  a local
        reduction between non-3-colorability and
        non-$k$-colorability, and discuss the main features we expect
        of a local
        reduction. This motivates the general definition of a local
        reduction between problems that we present in Section \ref{sec:reduc}, together with a number of applications that give polynomial lower bounds for several \textsf{coNP}-hard problems.
        
        Then, in Section~\ref{sec:red3SAT}, we present a variant of our general reduction theorem which allows us to do local reductions from 3-SAT. This is motivated by the fact that it is  sometimes much more convenient to do an \textsf{NP}-hardness reduction from 3-SAT rather than from 3-colorability.
        To illustrate this, we give two applications, for problems having well-known \textsf{NP}-completeness reductions from 3-SAT.
        
        In Section~\ref{sec:variants}, we present a small variant of
        our general reduction theorem (namely, we give a relaxation of
        one of the conditions needed to apply it, that holds for
        properties on bounded-degree graph classes). We present again
        a direct application. Finally, in Section~\ref{sec:remark}, we
        show that there is no hope for a much more general theorem
        stating that any \textsf{coNP}-hard problem has polynomial
        local complexity: more precisely, we give an example of a        \textsf{coNP}-hard problem that has only logarithmic local complexity.
        We also show that our reduction framework applies not only to
        \textsf{coNP}-hard problems, by giving examples of reductions
        showing polynomial lower bounds on the local complexity of
        some properties that are in \textsf{P}.
        
\section{Preliminaries}\label{sec:prel}

\subsection{Graphs}

In this paper, graphs are assumed to be simple, loopless,
undirected, and connected.
The \emph{distance} between two vertices $u$ and $v$ in a graph $G$,
denoted by $d_G(u,v)$  is the minimum
length of a path between $u$ and $v$ (that is, its number of edges). The \emph{neighborhood} of a vertex $v$
in a graph $G$, denoted by $N_G(v)$ (or $N(v)$ if $G$ is clear from the
context),  is the set of vertices at distance exactly 1 from $v$. The
\emph{closed neighborhood} of $v$, denoted by  $N_G[v]:=\{v\}\cup N_G(v)$,
is the set of vertices at distance at most 1 from $v$. For a set $S$
of vertices of $G$, we define $N_G[S]:=\bigcup_{v\in S}N_G[v]$. Finally, for
every $k \geqslant 1$, we define $N_G^k[v]$ as the set of vertices at
distance at most~$k$ from~$v$ (in particular, we have $N_G^1[v] = N_G[v]$).

\subsection{Local certification}\label{sec:defpls}

In all the  $n$-vertex
graphs $G$ that we consider, the vertices of $G$ are assumed to be assigned distinct (but otherwise arbitrary)
identifiers $(\text{id}(v))_{v\in V(G)}$ from a set 
$\{1,\ldots,\text{poly}(n)\}$.
This is a standard assumption in local certification, which is
necessary if we want to be able to distinguish a graph from its
universal cover (and thus prove any meaningful result).
When we refer to a subgraph $H$ of a graph $G$, we
implicitly refer to the corresponding labeled subgraph of~$G$.
Note that the identifiers of each of the vertices of $G$ can be stored
using $O(\log n)$ bits.

\medskip

We follow the
terminology introduced by G\"o\"os and Suomela~\cite{GoosS16}.

\subsection*{Proofs and provers} A \emph{proof} for a graph $G$ is a function
$P:V(G)\to \{0,1\}^*$ (as $G$ is a  labeled graph,
the proof $P$ is allowed to depend on the identifiers of the
vertices of $G$). The binary words $P(v)$ are called \emph{certificates}. The \emph{size} of $P$ is the maximum size of a
certificate $P(v)$, for $v\in V(G)$.
A \emph{prover} for a graph class $\mathcal{G}$ is a function that
maps every $G\in \mathcal{G}$ to a proof for~$G$.

\subsection*{Local verifiers} A \emph{verifier} $\mathcal{A}$ is a function that takes a
graph $G$, a proof $P$
for $G$, and a vertex $v\in V(G)$ as inputs, and outputs an element of
$\{0,1\}$. We say that $v$ \emph{accepts} the instance if
$\mathcal{A}(G,P,v)=1$ and that $v$ \emph{rejects} the instance if
$\mathcal{A}(G,P,v)=0$. In the definition we do not require that $\mathcal{A}$ is
efficiently computable or even decidable (but in practice it will
almost always be the case).

Consider a graph $G$, a proof $P$ for $G$, and a
vertex $v\in V(G)$.
We denote by $G[v]$ the
subgraph of
$G$ induced by $N[v]$, the closed neighborhood of $v$, and similarly we denote by $P[v]$ the restriction of $P$ to
$N[v]$.

\medskip

A verifier $\mathcal{A}$ is \emph{local} if for any $v\in G$, the
output of $v$ only depends on its identifier and $P[v]$.

\medskip

We note that there is a small difference here with the original
setting of G\"o\"os and Suomela~\cite{GoosS16}, where the verifier
was allowed to depend on $G[v]$ as well (so vertices have the knowledge
of the subgraph induced by their neighborhood). Lower bounds in the
setting of  G\"o\"os and Suomela imply lower bounds in our setting,
but the converse does not hold in general. However, in the case of 
bounded degree graphs (which is the setting of most of our results) the two models are very close, as encoding the
subgraphs induced by each neighborhood can be done at an extra cost of
$O(\log n)$ bits per vertex. See also  Section~\ref{sec:variants}, for
the use of a similar idea, in a related context.

\medskip

Since the lower bounds of G\"o\"os and Suomela~\cite{GoosS16} on
non-3-colorability hold in their stronger setting, all the lower bounds we
deduce from it in this paper also hold in their setting (this is the case of all the
results of the paper, except those in the final section on the
exclusion of subgraphs). Since our goal was to design a general
framework for reductions in local certification, we have preferred to
work with the weaker setting, which applies more broadly (see 
\cite{Feuilloley21} for an overview of the different models of local
certification).

\medskip

A second difference between the setting of this paper and that of
G\"o\"os and Suomela~\cite{GoosS16} is that vertices only look at
distance at their neighbors 1 before accepting or rejecting the instance, while in~\cite{GoosS16} they can look at their
neighborhood at distance $r$, for some constant $r$ called the
\emph{local horizon}. The lower bounds of G\"o\"os and Suomela hold in
this more general setting, and so do the result in our paper (most of
the time, without any modification to the proofs). However, there are
quite a few places where having an arbitrary local horizon $r$ adds
a layer of technicality on top of some already quite technical statements
and proofs. So, for the sake of readability we have chosen to only state
and prove our results in this paper in the case where the local
horizon $r=1$, which already contains most of  the difficulty of the
problems we consider.

\subsection*{Proof labeling schemes}

A \emph{proof labeling scheme}
for a graph class $\mathcal{G}$ is a prover-verifier pair
$(\mathcal{P},\mathcal{A})$ where $\mathcal{A}$ is local, with the following properties.

\medskip

\noindent {\bf Completeness:} If $G\in \mathcal{G}$, then
$P:=\mathcal{P}(G)$ is a proof for $G$ such that for any vertex $v\in
V(G)$, $\mathcal{A}(G,P,v)=1$.

\medskip

\noindent {\bf Soundness:}  If $G\not\in \mathcal{G}$, then for every proof
$P'$ for $G$, there exists a vertex $v\in
V(G)$ such
that  $\mathcal{A}(G,P',v)=0$.

\medskip

In other words, upon looking at its closed neighborhood (labeled by
the identifiers and certificates), the local verifier of each vertex
of a graph $G\in \mathcal{G}$ accepts the instance, while if $G\not\in
\mathcal{G}$, for every possible choice of certificates, the local verifier of at least one vertex rejects the instance.

\medskip

The \emph{complexity} of the proof labeling scheme is the maximum size of a
proof $P=\mathcal{P}(G)$ for an $n$-vertex graph $G\in\mathcal{G}$,
and the \emph{local complexity} of $\mathcal{G}$ is the minimum
complexity of a proof labeling scheme for $\mathcal{G}$. If we say that the complexity is $O(f(n))$, for some
function $f$, the $O(\cdot)$ notation refers to $n\to \infty$. See
\cite{Feuilloley21,GoosS16} for more details on proof labeling schemes and local
certification in general.

\medskip

We implicitly identify every graph property $\mathcal{P}$, for
instance 3-colorability, or non-3-colorability, with the class of all
graphs that satisfy $\mathcal{P}$.  

\section{Non-3-colorability}\label{sec:3col}
	
	In what follows, we will use the following
result, proved by G\"o\"os and Suomela~\cite{GoosS16}:
	
	\begin{theorem}[\cite{GoosS16}]
		\label{thm:non3col_quadratique}
		Non-3-colorability has local complexity~$\Omega(n^2/\log
			n)$.
                  \end{theorem}

                  The full
                  description of the construction in \cite{GoosS16}
                  takes over 6 pages (including several very helpful
                  figures), so we do not describe it completely
                  here. But since we will need several variants of this result, we give a
                  short description of the main properties of the
                  construction (the interested reader is referred to
                  \cite{GoosS16} for more details). 
                   Let $k$ be an integer, and let $I=\{1,\ldots,2^k\}
                   $. For any two subsets
                  $A\subseteq I\times I $ and $B\subseteq I\times I$, G\"o\"os and Suomela construct a
                  graph $G_{A,B}$ on $n=\Theta(2^k)$ vertices such
                  that:

                  \begin{enumerate}
\item $G_{A,B}$ can be partitioned into two parts $V_A$ (such that $G_{A,B}[V_A]$ depends
  only on $A$) and $V_B$ (such that $G_{A,B}[V_B]$ depends
  only on   $B$).
  \item $V_A$ contains a set $S_A=a_1,\ldots,a_{s}$ of $\Theta(\log n)$ special vertices, and
    $V_B$ contains a set $S_B=b_1,\ldots,b_{s}$ of
    $\Theta(\log n)$ special vertices.

    \item Each edge between $V_A$ and $V_B$ in $G_{A,B}$ connects a
    vertex of 
    $S_A$ to a vertex of $S_B$, there are at most $O(k)$ such edges,
    and the subgraph of $G_{A,B}$ induced by the special vertices
    $S_A\cup S_B$ is independent of the
    choice of $A$ and $B$ (in particular $s$ does not depend on $A$
    and $B$). 
    \item $G_{A,B}$ is 3-colorable if and only if $A\cap B\ne \emptyset$.
    \end{enumerate}

    \smallskip

Given assignments of certificates to the vertices of two graphs $G_{A,B}$ and $G_{C,D}$, we say that the
\emph{certificates agree on their special vertices}, if for every $1\le
i \le s$ the certificate of $a_i$ in $G_{A,B}$ coincides with the certificate of
$a_i$ in $G_{C,D}$, and the certificate of $b_i$ in $G_{A,B}$ coincides with the certificate of
$b_i$ in $G_{C,D}$.

\smallskip
                  
For completeness we now  sketch the argument of G\"o\"os and Suomela \cite{GoosS16} 
showing that non-3-colorability cannot be certified locally with
                  certificates of size $o(n^2/\log n)$. Consider $G_{A,\bar{A}}$ for some set $A\subseteq I\times
I$, where $\bar{A}$ denotes the
    complement of $A$.
Note that (4) implies that 
    $G_{A,\bar{A}}$ is not 3-colorable.

    Since there are only $O(\log n)$ special vertices, 
                  if non-3-colorability can be certified locally with
                  certificates of size $o(n^2/\log n)$, the total
                  number of bits of certificates assigned to the
                  special vertices of $G_{A,\bar{A}}$ is $o(n^2)$, and
                  thus the number of possible certificate assignments to
                  the special vertices of $G_{A,\bar{A}}$ is $2^{o(n^2)}$.  Since there are $2^{(2^{k})^2}=2^{\Omega(n^2)}$ choices for the set $A$,
                  it follows from the pigeonhole principle that there exist
                  two distinct sets $A,B$ such that the certificates of $G_{A,\bar{A}}$ and
                  $G_{B,\bar{B}}$ agree on their special vertices. As $G_{A,\bar{A}}$ and
                  $G_{B,\bar{B}}$ are both non-3-colorable, in any
                  local certification scheme for non-3-colorability in
                  these two graphs, each
                  vertex accepts the instance upon reading its local
                  view  (i.e., the certificates
                  of its closed neighborhood).

                  Note that at
                  least one of $A\cap \bar{B}$ and $\bar{A}\cap B$ is
                  non-empty, say $A\cap \bar{B}\ne \emptyset$ by symmetry.
Consider the graph $G_{A,\bar{B}}$, and assign to the vertices of
$V_A$ their certificates in $G_{A,\bar{A}}$, and to the vertices of
$V_{\bar{B}}$ their certificate in $G_{B,\bar{B}}$. Observe that the
local view of each vertex of $G_{A,\bar{B}}$ is the same as in
$G_{A,\bar{A}}$ or $G_{B,\bar{B}}$, and thus each vertex accepts the
instance. By definition of a proof labeling scheme, this implies that
$G_{A,\bar{B}}$ is non-3-colorable. On the other hand,  by (4) above, $G_{A,\bar{B}}$ is 3-colorable (because $A\cap
\bar{B}\ne \emptyset$). This  contradiction shows
that non-3-colorability cannot be certified locally with
                  certificates of size $o(n^2/\log n)$, as desired.

                  \medskip
	
	We now deduce the following
        variant of the previous result, which will be
        crucial in most of our arguments.
	
	\begin{theorem}
		\label{thm:non3coldegmax4_lineaire}
		Non-3-colorability of graphs of maximum degree~4 has
                local complexity~$\Omega(n/\log n)$.
	\end{theorem}
	
	\begin{proof}
	For a graph $G$, we construct a new graph $f(G)$ of maximum
        degree~4 such that $G$ is 3-colorable if and only if $f(G)$ is
        3-colorable (this is a classical construction~\cite{GareyJS76}).
The graph $f(G)$ is constructed as follows: for every vertex $u$ of
degree~$d$ in $G$, we consider the graph $G_u$ depicted in
Figure~\ref{fig:3col4}, where the $d$ white vertices are indexed by
the neighbors of $u$ in $G$. Observe that in any 3-coloring of $G_u$,
all the white vertices receive the same color. The graph $f(G)$ is
obtained by taking the
disjoint union of all graphs $G_u$, for $u\in V(G)$, and adding, for every
edge $uv\in E(G)$, an edge between the vertex of $G_u$
indexed by $v$ and the vertex of $G_v$ indexed by $u$. Note that
$f(G)$ has maximum degree~$4$ and is 3-colorable if and only if $G$ is
3-colorable. As each graph $G_u$ contains at most $8d_G(u)$ vertices,
$f(G)$ contains at most $16|E(G)|\le 8|V(G)|^2$ vertices. 

\smallskip

\begin{figure}[htb]
  \centering
  \includegraphics[scale=1.2]{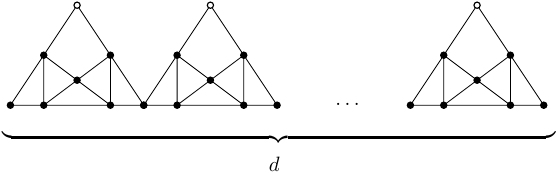}
  \caption{The graph $G_u$.}
  \label{fig:3col4} 
\end{figure}

For any integer $k$, and any two subsets
                  $A\subseteq I\times I $ and $B\subseteq I \times I$
                  (where $I=\{1,\ldots,2^k\}$), we
                  consider the graph $G_{A,B}$ of G\"o\"os and Suomela
                  \cite{GoosS16} introduced above, and we let
                  $H_{A,B}=f(G_{A,B})$. Note that $H_{A,B}$ contains
                  $n=\Theta(2^{2k})$ vertices, and moreover:

                  \begin{enumerate}
\item $H_{A,B}$ can be partitioned into two parts $V_A$ (such that $H_{A,B}[V_A]$ depends
  only on $A$) and $V_B$ (such that $H_{A,B}[V_B]$ depends
  only on   $B$).
  \item $V_A$ contains a set $S_A$ of $\Theta(\log n)$ special vertices, and
    $V_B$ contains a set $S_B$ of $\Theta(\log n)$ special vertices.
      \item Each edge between $V_A$ and $V_B$ in $G_{A,B}$, connects a
    vertex of 
    $S_A$ to a vertex of $S_B$, 
    and the subgraph of $G_{A,B}$ induced by the special vertices
    $S_A\cup S_B$ is independent of the
    choice of $A$ and $B$ (in particular the size of $S_A$ and $ S_B$
    is independent of $A$ and $B$). 
    \item $G_{A,B}$ is 3-colorable if and only if $A\cap B\ne \emptyset$.
                  \end{enumerate}
The fact that the sets $S_A$ and $S_B$ still have size $O(\log n)$
in $H_{A,B}$ follows from item (3) in the properties of $G_{A,B}$ (the
fact that there are only $O(\log n)$ edges with one endpoint in $S_A$
and one endpoint in $S_B$).

As above, if non-3-colorability can be certified locally with
                  certificates of size $o(n/\log n)$ in graphs of
                  maximum degree~4, the total
                  number of bits of certificates assigned to the
                  special vertices of $H_{A,\bar{A}}$ is $o(n)$.  Since there are $2^{(2^{k})^2}=2^{\Omega(n)}$ choices for the sets $A$,
                  it follows from the pigeonhole principle that there exist
                  two distinct sets $A,B$ such that the certificates of $H_{A,\bar{A}}$ and
                  $H_{B,\bar{B}}$ agree on their special vertices. We
                  then obtain a contradiction using the same argument as
                  above.
                \end{proof}

                Using a classical uncrossing technique we obtain the
                following variant for planar graphs (of bounded
                degree). 
	
	\begin{theorem}
		\label{thm:non3coldegmax4_lineaireplanaire}
		Non-3-colorability of planar graphs of maximum
                degree~4 has local complexity~$\Omega(\sqrt{n}/\log^2 n)$.
	\end{theorem}
	
	\begin{proof}
We start with the graph $H_{A,B}$ constructed in the proof of Theorem
\ref{thm:non3coldegmax4_lineaire}, and draw it in the plane with $V_A$
inside a region $R_A$, $V_B$ inside a region
$R_B$ disjoint from $R_A$, with the vertices of $S_A$ lying on the boundary of $R_A$
and the vertices $S_B$ lying on the boundary of $R_B$. Moreover we make sure that
the bipartite graph induced by the edge $E(S_A,S_B)$ with one endpoint in $S_A$ and
the other in $S_B$ is drawn outside of the interior of $R_A$ and
$R_B$ (so that the edges between $S_A$ and $S_B$ do not cross any
other edge of $H_{A,B}[V_A]$ and $H_{A,B}[V_B]$), see Figure
\ref{fig:ABplanar}. We call crossings between edges of $E(S_A,S_B)$
\emph{special crossings}.

\begin{figure}[htb]
  \centering
  \includegraphics[scale=0.9]{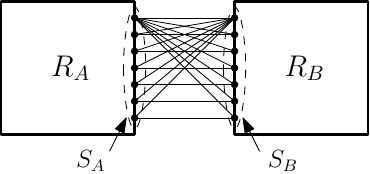}
  \caption{Drawing $H_{A,B}$ in the plane before the addition of uncrossing gadgets.}
  \label{fig:ABplanar} 
\end{figure}

We then use uncrossing gadgets
introduced in \cite{GareyJS76} (see Figure \ref{fig:cross} for an
illustration). Let $P_{A,B}$ be the resulting planar graph (which is
3-colorable if and only if $A\cap B\ne \emptyset$). Note that
the number of vertices of $P_{A,B}$ is equal to the number of vertices
of $H_{A,B}$ plus  a
constant number of vertices per crossing. As $H_{A,B}$
has $n=\Theta(2^{2k})$ vertices and maximum degree~4, it has
$\Theta(2^{2k})$ edges and thus at most $O(2^{4k})$ edge
crossings. It follows that the planar graph $P_{A,B}$ has
$O(2^{4k})$ vertices. Note that $P_{A,B}$ has maximum degree~$4\cdot
3=12$, as  $H_{A,B}$ has maximum degree~4.
\smallskip

\begin{figure}[htb]
  \centering
  \includegraphics[scale=0.8]{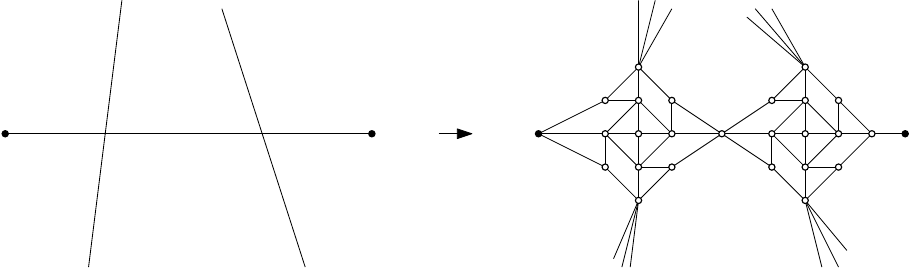}
  \caption{Uncrossing gadgets.}
  \label{fig:cross} 
\end{figure}

As $S_A$ and $S_B$ have size $O(\log n)$ in $H_{A,B}$ and this graph
has maximum degree~4, $E(S_A,S_B)$ contains $O(\log n)$ edges, and
thus there are $O(\log^2 n)$ special crossings. It follows
that there are $O(\log^2 n)$ special vertices in $P_{A,B}$ (as each
new crossing of $H_{A,B}$ creates a constant number of new vertices in
$P_{A,B}$).

We now consider $P'_{A,B}=f(P_{A,B})$ (the graph obtained from
$P_{A,B}$ by applying the maximum degree~4 reduction of the proof of
Theorem \ref{thm:non3coldegmax4_lineaire}). This graph is a planar graph
with maximum degree~4, is 3-colorable if and only if $A\cap B\ne
\emptyset$, has $n=O(2^{4k})$ vertices, but only contains $O(\log^2
n)$ special vertices. There are
$2^{(2^k)^2}=2^{2^{2k}}=2^{\Omega(\sqrt{n})}$ choices for the set $A$,
so if we can locally certify non-3-colorability with certificates of
size $o(n^{1/2}/\log^2 n)$ then the total number of bits of
certificates assigned to the special vertices of $P'_{A,\bar{A}}$ is $o(n^{1/2})$ and thus
there are two distinct sets $A,B$ such that the certificates of $P'_{A,\bar{A}}$ and
$P'_{B,\bar{B}}$ agree on their special vertices. We then obtain a contradiction using the same argument as above.
\end{proof}

                \begin{remark}
A closer inspection of the construction of G\"o\"os and
Suomela \cite{GoosS16} yields that the subgraph between the special
vertices $S_A$ and $S_B$ can be drawn in the plane with just $O(\log
n)$ crossings (instead of $O(\log^2 n)$ in the proof above). This
implies that the lower bound on the size of the certificates  in
Theorem \ref{thm:non3coldegmax4_lineaireplanaire}  can be increased to
$\Omega(n^{1/2}/\log n)$.  As we do
not believe that the exponent $1/2$ is best
possible we omit the details.
\end{remark}

\section{Non-$k$-colorability}\label{sec:example}

In this section we explain how to obtain a $\Omega(n^2/\log
    n)$ lower bound on the local complexity of
non-$k$-colorability, for every $k > 3$, using the results of the
previous section. This will also be a direct consequence of our main
reduction result, Theorem~\ref{thm:main theorem local reduction}, but
since the statement of this theorem is quite technical we prefer to
illustrate the main ideas of the proof with this simple example before
actually stating the
theorem. 

\medskip
	
So let us prove that at least $\Omega(n^2/\log n)$
bits are necessary to locally certify non-$k$-colorability for every $k
\geqslant 4$. Assume that there is a proof labeling scheme of complexity
$f(n)$ for this problem, and consider a
graph $G$ for which we wish to certify that $G$ is not
3-colorable. The main idea will be to simulate locally the construction of
a graph $G'$, so that $G'$ is $k$-colorable if and only if $G$ is
3-colorable, use the local certificates of non-$k$-colorability of
$G'$ and assign them to the vertices of $G$, in such a way that $G$
can also simulate the verification of the certificates of~$G'$.

\subsection*{First attempt}

We start by defining the graph $G'$ as the graph
obtained from $G$ by adding $k-3$ universal vertices (such a graph
$G'$ is $k$-colorable if and only if $G$ is 3-colorable).

Given a
proof labeling scheme for non-$k$-colorability in $G'$, our goal is to assign to
each vertex $u\in V(G)$ the list of certificates of a subset $C_u$ of
vertices of $G'$, together with their identifiers, in such a way that $u$ can
simulate the verification procedure of a subset $V_u$ of vertices of $G'$. In this case, our choice
for $C_u$ and $V_u$ is the following: for any $u\in V(G)$, we set $C_u=V(G')$
(that is, we assign to each vertex of $G$ the list of the certificates of all the vertices of $G'$ together
with their identifier), and $V_u$ consists of $u$ together with the
$k-3$ universal vertices of $G'$. This allows each vertex $u\in V(G)$ to
simulate the
verification process of the corresponding vertex in $G'$, since its local view in $G'$ consists of its
neighbors in $G$, together with the universal vertices. Each vertex
$u\in V(G)$ also simulates the verification of the $k-3$ universal
vertices (whose local view is the whole graph). Now each vertex $u\in
V(G)$ accepts the instance if and only if it would have accepted the
instance in $G'$ and each of the universal vertices would also have
accepted the instance in $G'$. The size of the certificates is at most
$n$ times the size of the certificates for non-$k$-colorability, so by
Theorem \ref{thm:non3col_quadratique}
this shows that non-$k$-colorability requires certificates of size
$\tfrac1n\cdot \Omega(n^2/\log n)=\Omega(n/\log n)$.

\medskip

The resulting bound of $\Omega(n/\log n)$ is much weaker than the
bound of $\Omega(n^2/\log n)$ for non-3-colorability. The main reason is that adding universal vertices is very far from
being a local operation, and this should be avoided as much as
possible. Let us now consider a second (more local) way of defining the
graph $G'$.

\subsection*{Second attempt}

The new graph $G'$ is obtained from the disjoint union of two graphs $G_0$ and $G_1$ defined
as follows. The graph $G_0$ is an isomorphic copy of $G$. The graph
$G_1$ is obtained from the graph $G$ by replacing each vertex~$u$ by a clique
of size~$k$ with vertices numbered from~1 to~$k$. We denote this
clique by $K_u$, and for any $1\le i \le k$, the vertex numbered~$i$
in $K_u$ is denoted by $K_u[i]$. Finally,
each edge $uv$ of $G$ is replaced by an antimatching between $K_u$ and
$K_v$ in $G_1$ (that is, for any $i\ne j$ we add an edge between
$K_u[i]$ and $K_v[j]$ in $G_1$). The edges between $G_0$ and $G_1$ are the following: for each $u \in V(G_0)$, we put an edge between $u$ and $K_u[i]$ for every $i \in \{4, \ldots, k\}$.
		We can define an identifier for each vertex of $G'$
                describing whether it belongs to $G_0$, or to $G_1$
                (and its number from $\{1, \ldots, k\}$ in the latter
                the case), together with the identifier of the
                corresponding vertex in $G$. All this information can
                be encoded on $O(\log n)$ bits.

                	It is straightforward to show that $G$ is 3-colorable
                if and only if $G'$ is $k$-colorable: indeed, the
                antimatchings enforce that, in any $k$-coloring of
                $G'$, for each $i \in \{1, \ldots, k\}$ the vertices from
                $\{K_v[i]: v \in V(G)\}$ all have the same color, and
                thus there are only three possible remaining colors
                for the vertices of $G_0$, which is an isomorphic copy
                of $G$.

		For every $u \in V(G)$, we define $V_u:=C_u$ as being
                the union of the copy of $u$ in $G_0$ and the clique
                $K_u$ in $G_1$ (as above, $C_u$ is the list of
                vertices of $G'$ whose certificate will be given to
                $u$ and $V_u$ is the set of vertices of $G'$ whose
                verification process will be simulated by $u$). See
                Figure~\ref{fig:k-col} for an example. Note that these
                lists of certificates are enough for each vertex $u$ of
                $V(G)$ to simulate the verification process of $V_u$,
                as for any neighbor $w$ of a vertex $v\in V_u$ in
                $G'$, the certificate of $w$ in $G'$ will be contained
                in the certificate assigned to some neighbor of $u$ in $G$.

                Finally, let us compute the size of certificates in
                the resulting proof labeling scheme for
                non-3-colorability, assuming the local complexity of
                non-$k$-colorability is $f(n)$. If $G$ has $n$
                vertices, then $G'$ has $(k+1)n$ vertices. So the
                size of the certificate assigned to each vertex of $G$
                is $(k+1)\big(O(\log n)+f\big((k+1)n\big)\big)$. As $k$ is constant, it follows
                from Theorem \ref{thm:non3col_quadratique} that non-$k$-colorability requires certificates of size
$\Omega(n^2/\log n)$, as desired.
				Note also that, if $G$ has maximum degree~4, then $G'$ has maximum degree~$5k-4$. Thus, it follows from Theorem~\ref{thm:non3coldegmax4_lineaire} that non-$k$-colorability requires certificates of size $\Omega(n/\log n)$ in graphs of maximum degree~$5k-4$.
		
		\begin{figure}[h!]
		\label{fig:k-col}
		\centering
		
		\begin{tikzpicture}[x=0.55pt,y=0.55pt,yscale=-0.7,xscale=0.7]
			
			\draw    (324,145) .. controls (496.5,-52) and (690.5,-7) .. (573.5,159.5) ;
			\draw    (388.5,209.5) .. controls (394.5,57) and (761.5,-41) .. (640.5,226.5) ;
			\draw    (324,273) .. controls (349.5,324) and (465.5,380) .. (574.5,287.5) ;
			\draw    (259.5,209) .. controls (274.5,395) and (420.5,399) .. (508.5,224.5) ;
			\draw    (259.5,209) -- (324,273) ;
			\draw    (324,273) -- (388.5,209.5) ;
			\draw    (259.5,209) -- (324,145) ;
			\draw    (388.5,209.5) -- (324,145) ;
			\draw [line width=3]    (542,287) -- (476,224) ;
			\draw [line width=3]    (574.5,287) -- (640.5,226.5) ;
			\draw [line width=3]    (573.5,127.5) -- (640.5,194.5) ;
			\draw [line width=3]    (476,192) -- (541,127) ;
			\draw [line width=3]    (508.5,224.5) -- (542,255) ;
			\draw [line width=3]    (608,194) -- (573.5,159.5) ;
			\draw [line width=3]    (574.5,255) -- (608,226) ;
			\draw [line width=3]    (508.5,192) -- (541,159) ;
			\draw    (541,127) -- (573.5,159.5) ;
			\draw    (541,159) -- (573.5,127.5) ;
			\draw    (541,127) -- (541,159) ;
			\draw    (541,159) -- (573.5,159.5) ;
			\draw    (573.5,127.5) -- (573.5,159.5) ;
			\draw    (541,127) -- (573.5,127.5) ;
			\draw  [fill={rgb, 255:red, 255; green, 255; blue, 255 }  ,fill opacity=1 ] (535,127) .. controls (535,123.69) and (537.69,121) .. (541,121) .. controls (544.31,121) and (547,123.69) .. (547,127) .. controls (547,130.31) and (544.31,133) .. (541,133) .. controls (537.69,133) and (535,130.31) .. (535,127) -- cycle ;
			\draw  [fill={rgb, 255:red, 255; green, 255; blue, 255 }  ,fill opacity=1 ] (535,159) .. controls (535,155.69) and (537.69,153) .. (541,153) .. controls (544.31,153) and (547,155.69) .. (547,159) .. controls (547,162.31) and (544.31,165) .. (541,165) .. controls (537.69,165) and (535,162.31) .. (535,159) -- cycle ;
			\draw  [fill={rgb, 255:red, 255; green, 255; blue, 255 }  ,fill opacity=1 ] (567.5,127.5) .. controls (567.5,124.19) and (570.19,121.5) .. (573.5,121.5) .. controls (576.81,121.5) and (579.5,124.19) .. (579.5,127.5) .. controls (579.5,130.81) and (576.81,133.5) .. (573.5,133.5) .. controls (570.19,133.5) and (567.5,130.81) .. (567.5,127.5) -- cycle ;
			\draw  [fill={rgb, 255:red, 255; green, 255; blue, 255 }  ,fill opacity=1 ] (567.5,159.5) .. controls (567.5,156.19) and (570.19,153.5) .. (573.5,153.5) .. controls (576.81,153.5) and (579.5,156.19) .. (579.5,159.5) .. controls (579.5,162.81) and (576.81,165.5) .. (573.5,165.5) .. controls (570.19,165.5) and (567.5,162.81) .. (567.5,159.5) -- cycle ;
			\draw    (476,192) -- (508.5,224.5) ;
			\draw    (476,224) -- (508.5,192.5) ;
			\draw    (476,192) -- (476,224) ;
			\draw    (476,224) -- (508.5,224.5) ;
			\draw    (508.5,192) -- (508.5,224) ;
			\draw    (476,192) -- (508.5,192.5) ;
			\draw  [fill={rgb, 255:red, 255; green, 255; blue, 255 }  ,fill opacity=1 ] (470,192) .. controls (470,188.69) and (472.69,186) .. (476,186) .. controls (479.31,186) and (482,188.69) .. (482,192) .. controls (482,195.31) and (479.31,198) .. (476,198) .. controls (472.69,198) and (470,195.31) .. (470,192) -- cycle ;
			\draw  [fill={rgb, 255:red, 255; green, 255; blue, 255 }  ,fill opacity=1 ] (470,224) .. controls (470,220.69) and (472.69,218) .. (476,218) .. controls (479.31,218) and (482,220.69) .. (482,224) .. controls (482,227.31) and (479.31,230) .. (476,230) .. controls (472.69,230) and (470,227.31) .. (470,224) -- cycle ;
			\draw  [fill={rgb, 255:red, 255; green, 255; blue, 255 }  ,fill opacity=1 ] (502.5,192.5) .. controls (502.5,189.19) and (505.19,186.5) .. (508.5,186.5) .. controls (511.81,186.5) and (514.5,189.19) .. (514.5,192.5) .. controls (514.5,195.81) and (511.81,198.5) .. (508.5,198.5) .. controls (505.19,198.5) and (502.5,195.81) .. (502.5,192.5) -- cycle ;
			\draw  [fill={rgb, 255:red, 255; green, 255; blue, 255 }  ,fill opacity=1 ] (502.5,224.5) .. controls (502.5,221.19) and (505.19,218.5) .. (508.5,218.5) .. controls (511.81,218.5) and (514.5,221.19) .. (514.5,224.5) .. controls (514.5,227.81) and (511.81,230.5) .. (508.5,230.5) .. controls (505.19,230.5) and (502.5,227.81) .. (502.5,224.5) -- cycle ;
			\draw    (608,194) -- (640.5,226.5) ;
			\draw    (608,226) -- (640.5,194.5) ;
			\draw    (608,194) -- (608,226) ;
			\draw    (608,226) -- (640.5,226.5) ;
			\draw    (640.5,194.5) -- (640.5,226.5) ;
			\draw    (608,194) -- (640.5,194.5) ;
			\draw  [fill={rgb, 255:red, 255; green, 255; blue, 255 }  ,fill opacity=1 ] (602,194) .. controls (602,190.69) and (604.69,188) .. (608,188) .. controls (611.31,188) and (614,190.69) .. (614,194) .. controls (614,197.31) and (611.31,200) .. (608,200) .. controls (604.69,200) and (602,197.31) .. (602,194) -- cycle ;
			\draw  [fill={rgb, 255:red, 255; green, 255; blue, 255 }  ,fill opacity=1 ] (602,226) .. controls (602,222.69) and (604.69,220) .. (608,220) .. controls (611.31,220) and (614,222.69) .. (614,226) .. controls (614,229.31) and (611.31,232) .. (608,232) .. controls (604.69,232) and (602,229.31) .. (602,226) -- cycle ;
			\draw  [fill={rgb, 255:red, 255; green, 255; blue, 255 }  ,fill opacity=1 ] (634.5,194.5) .. controls (634.5,191.19) and (637.19,188.5) .. (640.5,188.5) .. controls (643.81,188.5) and (646.5,191.19) .. (646.5,194.5) .. controls (646.5,197.81) and (643.81,200.5) .. (640.5,200.5) .. controls (637.19,200.5) and (634.5,197.81) .. (634.5,194.5) -- cycle ;
			\draw  [fill={rgb, 255:red, 255; green, 255; blue, 255 }  ,fill opacity=1 ] (634.5,226.5) .. controls (634.5,223.19) and (637.19,220.5) .. (640.5,220.5) .. controls (643.81,220.5) and (646.5,223.19) .. (646.5,226.5) .. controls (646.5,229.81) and (643.81,232.5) .. (640.5,232.5) .. controls (637.19,232.5) and (634.5,229.81) .. (634.5,226.5) -- cycle ;
			\draw    (542,255) -- (574.5,287.5) ;
			\draw    (542,287) -- (574.5,255.5) ;
			\draw    (542,255) -- (542,287) ;
			\draw    (542,287) -- (574.5,287.5) ;
			\draw    (574.5,255) -- (574.5,287) ;
			\draw    (542,255) -- (574.5,255.5) ;
			\draw  [fill={rgb, 255:red, 255; green, 255; blue, 255 }  ,fill opacity=1 ] (536,255) .. controls (536,251.69) and (538.69,249) .. (542,249) .. controls (545.31,249) and (548,251.69) .. (548,255) .. controls (548,258.31) and (545.31,261) .. (542,261) .. controls (538.69,261) and (536,258.31) .. (536,255) -- cycle ;
			\draw  [fill={rgb, 255:red, 255; green, 255; blue, 255 }  ,fill opacity=1 ] (536,287) .. controls (536,283.69) and (538.69,281) .. (542,281) .. controls (545.31,281) and (548,283.69) .. (548,287) .. controls (548,290.31) and (545.31,293) .. (542,293) .. controls (538.69,293) and (536,290.31) .. (536,287) -- cycle ;
			\draw  [fill={rgb, 255:red, 255; green, 255; blue, 255 }  ,fill opacity=1 ] (568.5,255.5) .. controls (568.5,252.19) and (571.19,249.5) .. (574.5,249.5) .. controls (577.81,249.5) and (580.5,252.19) .. (580.5,255.5) .. controls (580.5,258.81) and (577.81,261.5) .. (574.5,261.5) .. controls (571.19,261.5) and (568.5,258.81) .. (568.5,255.5) -- cycle ;
			\draw  [fill={rgb, 255:red, 255; green, 255; blue, 255 }  ,fill opacity=1 ] (568.5,287.5) .. controls (568.5,284.19) and (571.19,281.5) .. (574.5,281.5) .. controls (577.81,281.5) and (580.5,284.19) .. (580.5,287.5) .. controls (580.5,290.81) and (577.81,293.5) .. (574.5,293.5) .. controls (571.19,293.5) and (568.5,290.81) .. (568.5,287.5) -- cycle ;
			\draw  [fill={rgb, 255:red, 255; green, 255; blue, 255 }  ,fill opacity=1 ] (318,145) .. controls (318,141.69) and (320.69,139) .. (324,139) .. controls (327.31,139) and (330,141.69) .. (330,145) .. controls (330,148.31) and (327.31,151) .. (324,151) .. controls (320.69,151) and (318,148.31) .. (318,145) -- cycle ;
			\draw  [fill={rgb, 255:red, 255; green, 255; blue, 255 }  ,fill opacity=1 ] (382.5,209.5) .. controls (382.5,206.19) and (385.19,203.5) .. (388.5,203.5) .. controls (391.81,203.5) and (394.5,206.19) .. (394.5,209.5) .. controls (394.5,212.81) and (391.81,215.5) .. (388.5,215.5) .. controls (385.19,215.5) and (382.5,212.81) .. (382.5,209.5) -- cycle ;
			\draw  [fill={rgb, 255:red, 255; green, 255; blue, 255 }  ,fill opacity=1 ] (253.5,209) .. controls (253.5,205.69) and (256.19,203) .. (259.5,203) .. controls (262.81,203) and (265.5,205.69) .. (265.5,209) .. controls (265.5,212.31) and (262.81,215) .. (259.5,215) .. controls (256.19,215) and (253.5,212.31) .. (253.5,209) -- cycle ;
			\draw  [fill={rgb, 255:red, 255; green, 255; blue, 255 }  ,fill opacity=1 ] (318,273) .. controls (318,269.69) and (320.69,267) .. (324,267) .. controls (327.31,267) and (330,269.69) .. (330,273) .. controls (330,276.31) and (327.31,279) .. (324,279) .. controls (320.69,279) and (318,276.31) .. (318,273) -- cycle ;
			\draw  [color={rgb, 255:red, 208; green, 2; blue, 27 }  ,draw opacity=1 ][dash pattern={on 5.63pt off 4.5pt}][line width=1.5]  (443.5,89) .. controls (507.5,85) and (598.5,96) .. (597.5,140) .. controls (596.5,184) and (539.5,185) .. (484.5,152) .. controls (429.5,119) and (344.5,170) .. (323.5,162) .. controls (302.5,154) and (299.5,138) .. (314.5,125) .. controls (329.5,112) and (379.5,93) .. (443.5,89) -- cycle ;
			\draw    (7,209) -- (71.5,273) ;
			\draw    (71.5,273) -- (136,209.5) ;
			\draw    (7,209) -- (71.5,145) ;
			\draw    (136,209.5) -- (71.5,145) ;
			\draw  [fill={rgb, 255:red, 255; green, 255; blue, 255 }  ,fill opacity=1 ] (65.5,145) .. controls (65.5,141.69) and (68.19,139) .. (71.5,139) .. controls (74.81,139) and (77.5,141.69) .. (77.5,145) .. controls (77.5,148.31) and (74.81,151) .. (71.5,151) .. controls (68.19,151) and (65.5,148.31) .. (65.5,145) -- cycle ;
			\draw  [fill={rgb, 255:red, 255; green, 255; blue, 255 }  ,fill opacity=1 ] (130,209.5) .. controls (130,206.19) and (132.69,203.5) .. (136,203.5) .. controls (139.31,203.5) and (142,206.19) .. (142,209.5) .. controls (142,212.81) and (139.31,215.5) .. (136,215.5) .. controls (132.69,215.5) and (130,212.81) .. (130,209.5) -- cycle ;
			\draw  [fill={rgb, 255:red, 255; green, 255; blue, 255 }  ,fill opacity=1 ] (1,209) .. controls (1,205.69) and (3.69,203) .. (7,203) .. controls (10.31,203) and (13,205.69) .. (13,209) .. controls (13,212.31) and (10.31,215) .. (7,215) .. controls (3.69,215) and (1,212.31) .. (1,209) -- cycle ;
			\draw  [fill={rgb, 255:red, 255; green, 255; blue, 255 }  ,fill opacity=1 ] (65.5,273) .. controls (65.5,269.69) and (68.19,267) .. (71.5,267) .. controls (74.81,267) and (77.5,269.69) .. (77.5,273) .. controls (77.5,276.31) and (74.81,279) .. (71.5,279) .. controls (68.19,279) and (65.5,276.31) .. (65.5,273) -- cycle ;
			
			\draw (83,132.4) node [anchor=north west][inner sep=0.75pt]    {$u$};
			\draw (280,80) node [anchor=north west][inner sep=0.75pt]  [color={rgb, 255:red, 208; green, 2; blue, 27 }  ,opacity=1 ]  {$V_{u}$};
			\draw (60,302.4) node [anchor=north west][inner sep=0.75pt]  [font=\large]  {$G$};
			\draw (384,375.4) node [anchor=north west][inner sep=0.75pt]  [font=\large]  {$G'$};

		\end{tikzpicture}
		
		\caption{The graph $G'$ when $k=4$ and $G$ is a cycle of length 4. The thick double edges between the cliques represent the antimatchings.}
			
              \end{figure}
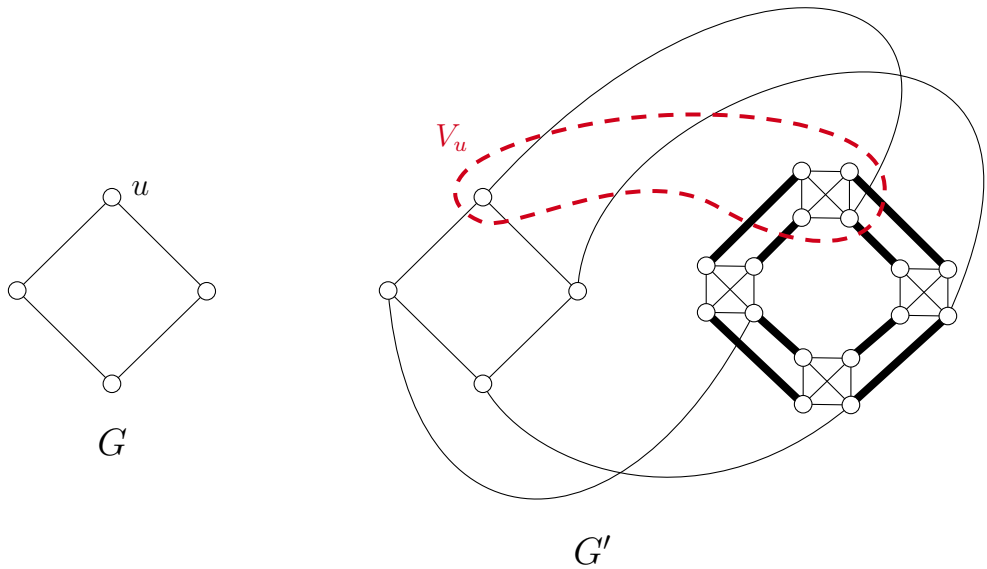

              We have thus proved the following extension of Theorems~\ref{thm:non3col_quadratique} and~\ref{thm:non3coldegmax4_lineaire}.

              	\begin{theorem}
		\label{thm:nonkcol_quadratique}
		For any fixed $k\ge 3$, non-$k$-colorability has local
                complexity at least $\Omega(n^2/\log
			n)$ in general graphs, and at least $\Omega(n/\log
			n)$ in graphs of maximum degree~$5k-4$.
                  \end{theorem}
              
              \begin{remark}
              	\label{rem:k+sqrtk-1}
              	In Section~\ref{sec:variants}, we will prove that
                non-$k$-colorability has local complexity at least
                $\Omega(n/\log n)$ even in graphs of maximum degree~$k
                + \lceil\sqrt{k}\rceil - 1$ (see Theorem~\ref{thm:nonkcol_lineaire}).
              \end{remark}

	In the next section, we identify the key components of the
       proof above and explain how it can be used to define a general
       class of local reductions between graph problems. In particular
       we formalize the necessary and sufficient conditions on the
       sets $V_u$ and $C_u$ allowing the local simulation of the
       verification process, and show
       how the relative sizes of these two sets influence the relation
       between the local complexities of the two problems.
	


	\section{Reductions from $\p$ to $\p'$}\label{sec:reduc}

		Let $\C$, $\C'$ be two graph classes, and $\p$, $\p'$ be two graph properties
		on $\C$, $\C'$ respectively. Let $\alpha, \beta :
                \mathbb{N} \to \mathbb{N}$, with $\beta(n) \geqslant n$ for every integer~$n$.

                \smallskip
                
		We say that there exists a \emph{local reduction} from~$\p$ to~$\p'$ with
		local expansion $\alpha$ and global expansion $\beta$ if there exists a function $f_{\p,\p'}$ which takes as input
		a graph $G \in \C$ with $n$ vertices and unique identifiers in $\{1, \ldots,
		n\}$ \footnote{We emphasize  that the range of identifiers is
                  more restricted than in the definition of local
                  certification, where they were assumed to be in  $\{1, \ldots,
		n^c\}$, for some $c\ge 1$. This is just to make the
                applications of our framework easier (but we will show that we do not
                lose any level of generality).}, and outputs a graph (with unique identifiers) $G'=f_{\p,\p'}(G) \in \C'$ having at most $\beta(n)$ vertices, such that the following properties are satisfied for all~$G$:
		\begin{enumerate}[(R1)]
			\item the identifiers of $G'$ are written on $O(\log n)$ bits, \label[lab]{l1}
			\item $G$ satisfies $\p$ if and only if $G'$ satisfies $\p'$, \label[lab]{l2}
			\item for every vertex $u$ of $G$, there exist two sets $C_u, V_u \subseteq
			V(G')$, such that, for every $u \in V(G)$, the following properties
			are satisfied: \label[lab]{l3}
			\begin{enumerate}
				\item $V_u \subseteq C_u$,\label[lab]{l3a}
				\item $\bigcup_{v \in V(G)} V_v = V(G')$, \label[lab]{l3b}
				\item $|C_u| \leqslant \alpha(n)$, \label[lab]{l3c}
				\item for every $t \in V(G')$, the subgraph of $G$ induced by $\{v
				\in V(G) \; | \; t \in C_v\}$ is connected, \label[lab]{l3d}
				\item for every $t \in V_u$, $N_{G'}[t]$ is included in $\bigcup_{v \in N_G[u]}
				C_v$, \label[lab]{l3e}
				\item $V_u$ and its neighborhood (resp.\ $C_u$) only depend on the identifiers of the vertices in $N_G[u]$. In
				other words: if $H$ is another
                                $n$-vertex graph from $\C$ with unique identifiers in
				$\{1, \ldots, n\}$, and if the subgraph with identifiers formed by $u$ and its adjacent edges is the same in $G$ and
				$H$, then the sets $V_u$ and the subgraphs with identifiers formed by the vertices in $V_u$ and their adjacent edges in
				$G'$ and $H'=f_{\p,\p'}(H)$ are the same (resp.\ the sets $C_u$ are the same in $G'$
				and $H'$). \label[lab]{l3f}
			\end{enumerate}
                      \end{enumerate}

	Let us elaborate on the different parts of the definition of a
        local reduction, in view of the results from the previous
        section. Condition \cref{l1}  says that there is no significant blow-up
        in the size of $G'$ compared to that of $G$ ($G'$ has size
        polynomial in that of $G$). It also states implicitly that
        there is a way to describe the identifiers of $G'$ as a function
        of the identifiers of $G$. Condition \cref{l2} is the underlying principle of
        a reduction. In condition \cref{l3}, the set $V_u$  again stands for the
        subsets of vertices of $G'$ whose verification in $G'$ will be simulated by
        the vertex $u$ of $G$, and $C_u$ corresponds to the set of
        vertices of $G'$ whose certificates will be given to
        $u$. It is reasonable to require that $u$ is given the certificates of
        the vertices it is supposed to simulate (that is $V_u\subseteq
        C_u$, according to \cref{l3a}). It is not strictly necessary, as $u$
        could also find the certificates of some vertices of $V_u$ in the certificates
        assigned to its neighbors in~$G$, but assuming $V_u\subseteq
        C_u$ makes some of our
        arguments simpler. Condition \cref{l3b} simply says that every
        vertex of $G'$ has to be simulated by at least one vertex
        of~$G$. Condition \cref{l3d} is necessary to make sure that if the
        certificate of $t$ in $G'$ is given to several vertices of
        $G$, then all these vertices are given the same
        certificate. Condition \cref{l3e} is simply saying that in order to
        simulate the verification process of every vertex of $V_u$ in $G'$,
        $u$ has all the required certificates in its neighborhood in
        $G$. Finally, condition \cref{l3f} makes sure $u$ can reconstruct
        the local views of all vertices of $V_u$ in $G'$, only based on
        its local view in $G$.

        \smallskip

        It is perhaps useful to comment on the local and global
        expansion functions $\alpha$ and $\beta$. In the first attempt
        in Section \ref{sec:example}, the global expansion was
        reasonable (the size of $G'$ was linear in the size of $G$),
        but the local expansion was maximal, as $C_u$ was the set of
        all vertices. In our second attempt, the global expansion was
        still reasonable (the size of $G'$ was again linear in the size of
        $G$), but the local expansion was bounded by a constant. 
	
	\begin{remark}
		\label{rem:def local reduction}
		From the definition of a local reduction, we can make the following observations:
		\begin{itemize}
			\item Every local reduction from $\p$ to $\p'$
                          with local expansion $\alpha$ has global
                          expansion at most~$n \alpha(n)$. This
                          follows from properties \cref{l3a}, \cref{l3b} and \cref{l3c}.
			In particular, for every $\delta > 0$, every local reduction with local expansion $O(n^\delta)$ has global expansion $O(n^{\delta+1})$.
			\item The definition of local reduction is symmetric between properties and their complement. In other words: there exists a local reduction
			from $\p$ to $\p'$ if and only if there exists a local
			reduction from $\overline{\p}$ to $\overline{\p'}$, the local and global expansions are the same, and we have $f_{\p,\p'} = f_{\np,\overline{\p'}}$.
		\end{itemize}
	\end{remark}
	
	We now prove that, as expected, whenever we have a local reduction from a
        problem $\p$ to a problem $\p'$, any local certification
        scheme for $\p'$ can be turned into a local certification
        scheme for $\p$. This implies that lower bounds on the local
        complexity of  $\p$ can
        be used to prove lower bounds on the local complexity of
        $\p'$, assuming the local and global expansion factors in the
        local reduction are not
        too large.
	
	\begin{theorem}
		\label{thm:main theorem local reduction}
		Assume that:
		\begin{itemize}
			\item there exists a certification scheme for $\p'$ with certificates of size
			$s(n)$, where $s : \mathbb{N} \to \mathbb{N}$ is an increasing function satisfying $s(n) = \Omega(\log n)$, and
			\item there exists a local reduction from $\p$ to $\p'$ with local expansion
			$\alpha$ and global expansion~$\beta$.
		\end{itemize}
		Then there exists a certification scheme for $\p$ with certificates of size
		$O\big(\alpha(n) \cdot s(\beta(n))\big)$.
	\end{theorem}

	Before proving Theorem~\ref{thm:main theorem local reduction},
        let us give the following immediate corollary (which is the result that we will use in practice to prove lower bounds).

	\begin{corollary}
		\label{cor:reduction3col}
		
		Let $\p$ be a graph property (on some graph class $\C$) and let $k \geqslant 3$.
		\begin{enumerate}
			\item If there exists a local reduction from $k$-colorability (in general
			graphs) to $\p$ with local expansion
                        $O(n^\delta)$ for some $0 \leqslant \delta <
                        2$, and global expansion $O(n^\gamma)$, then
                        $\np$ has local complexity
			$\Omega\big(n^{(2-\delta)/\gamma}/\log
                          n\big)$. In particular, since $\gamma
                        \leqslant \delta+1$, $\np$ has local complexity $\Omega\big(n^{\frac{3}{\delta+1}-1}/\log n\big)$.
			\item If there exists a local reduction from
                          3-colorability in graphs of maximum degree
                          4 (or from $k$-colorability in graphs of
			maximum degree~$5k-4$) to $\p$ with local expansion
                        $O(n^\delta)$ for some $0 \leqslant \delta <
                        1$, and global expansion $O(n^\gamma)$, then
                        $\np$ has local complexity 
			$\Omega\big(n^{(1-\delta)/\gamma}/\log
                          n\big)$. In particular, since $\gamma
                        \leqslant \delta+1$, $\np$ has local complexity $\Omega\big(n^{\frac{2}{\delta+1}-1}/\log n\big)$.
		\end{enumerate}
		
	\end{corollary}

	\begin{proof}
		\begin{enumerate}
			\item By contradiction, assume that there exists a certification scheme for
			$\np$ with certificates of size $o\big(n^{(2-\delta)/\gamma}/\log n\big)$.
			By Theorem~\ref{thm:main theorem local reduction}, there exists a
			certification scheme for non-$k$-colorability with certificates of size
			\[o\left(n^\delta \cdot \frac{(n^{\gamma})^{(2-\delta)/\gamma}}{\log n}\right) =
			o\left(\frac{n^2}{\log n}\right),\] which contradicts
			Theorem~\ref{thm:nonkcol_quadratique}
			\item By contradiction, assume that there exists a certification scheme for
			$\np$ with certificates of size $o\big(n^{(1-\delta)/\gamma}/\log n\big)$.
			By Theorem~\ref{thm:main theorem local reduction}, there exists a
			certification scheme for non-3-colorability in
                        graphs of maximum degree 4, or for non-$k$-colorability in graphs of maximum degree~$k + \lceil\sqrt{k}\rceil-1$, with
			certificates of size $o\big(n^\delta \cdot \frac{(n^{\gamma})^{(1-\delta)/\gamma}}{\log n}\big) =
			o\big(\frac{n}{\log n}\big)$. This
                        contradicts Theorem
                        \ref{thm:non3coldegmax4_lineaire} and
                        Theorem~\ref{thm:nonkcol_quadratique}, respectively.
		\end{enumerate}
	\end{proof}

	\begin{proof}[Proof of Theorem~\ref{thm:main theorem local reduction}.]
	Let $G \in \C$ be an $n$-vertex graph in which the vertices are given
		unique identifiers on $O(\log n)$ bits. First, using a
                renaming technique from~\cite{BousquetEFZ24}, we can
                assume that the unique identifiers
                $(\mathrm{id}(u))_{u\in V(G)}$ of~$G$ are in $\{1,
                \ldots, n\}$, at a cost of $O(\log n)$ additional bits
                per certificates. Since $s(n) = \Omega(\log n)$ and $\beta(n) \geqslant n$, we
                have $\log n=O\big(\alpha(n) \cdot s(\beta(n))\big)$, so the renaming
                does not affect the final bound. We also assume that
                every vertex knows~$n$ (because it can also be
                certified at a cost of $O(\log n)$ bits in the
                certificates using a spanning tree,
                see~\cite{Feuilloley21} for more details).

                \medskip
                
\noindent {\bf Certification.}	 The certificate given by the prover
to each vertex of~$G$ consists in a table. For every $u \in V(G)$,
this table is denoted by $\tab(u)$. To construct this table, the prover first computes the graph $G'=f_{\p, \p'}(G)$, and for every vertex $u \in V(G)$, it determines the sets~$C_u$ and~$V_u$. The entries of $\tab(u)$ are indexed by the identifiers
		in~$G'$ of the vertices in~$C_u$. For every $t \in C_u$, we will
		denote by $\tab(u)[t]$ the entry of $\tab(u)$ indexed by the identifier of~$t$.
		The prover writes in $\tab(u)[t]$ the certificate that $t$ would get in the
		certification scheme for $\p'$ in~$G'$.

		Let us show that these certificates have size $O\big(\alpha(n) \cdot  s(\beta(n))\big)$.
		By definition of global expansion, the graph~$G'$ has at most $\beta(n)$ vertices. Thus, each
		certificate (corresponding to the proof labeling scheme for~$\p'$
		in~$G'$) written in $\tab(u)$ has size at most $s(\beta(n))$. Since
		this table has at most $\alpha(n)$ entries (one for each vertex in $C_u$), its
		size is $O\big(\alpha(n) \cdot s(\beta(n))\big)$.
		
		\medskip
		
		\noindent {\bf Verification.}
		The verification procedure of each vertex $u\in V(G)$ consists in the following steps.
		At each step, if the verification fails, the procedure
                stops and $u$ rejects
                the instance.
		
		\begin{enumerate}[(1)]
			\item First, by condition \cref{l3f}, $u$ can compute the identifiers of all the
			vertices in~$C_u$. It checks that every vertex in $C_u$ has exactly one entry in
			$\tab(u)$, and conversely that every entry in $\tab(u)$ is indexed with the
			identifier of a vertex in $C_u$.
			
			\item Then, $u$ checks that, for every $t \in C_u$ and every $v \in N_G[u]$ having an entry indexed by $t$ in $\tab(v)$, we have $\tab(u)[t] = \tab(v)[t]$.

			\item Finally, $u$ does the following. By condition \cref{l3f}, $u$ can compute
			the set $V_u$ and the subgraph (with identifiers) of $G'$
			formed by $V_u$ and the adjacent edges. In particular, for every $t_0 \in V_u$, $u$ can compute the
			subgraph  of $G'$ formed by $t_0$ and its incident edges.
			
			Let $t_0 \in V_u$. Since $V_u \subseteq C_u$, $u$ can determine the certificate
			that $t_0$ would have in the certification scheme for $\p'$ in $G'$,
			because it is written in $\tab(u)[t_0]$.
			Moreover, by condition \cref{l3e}, every neighbor $t$ of $t_0$ in $G'$
			is in $\bigcup_{v \in N_G[u]} C_v$, so $u$ can also determine the certificate of
			$t$ because it is written in $\tab(v)[t]$ for some $v \in N_G[u]$.
			
			Thus, for all $t_0 \in V_u$, $u$ can recover the certificates of all the
			vertices in $N_{G'}[t_0]$, so $u$ can compute
                        the local view of $t_0$ in $G'$,
			and can apply the verification procedure of
                        the proof labeling scheme for $\p'$ in $G'$. If this fails for some $t_0 \in
			V_{u}$, then $u$ rejects.
			
			\item If $u$ did not reject at any previous
                          step, $u$ accepts the instance.
		\end{enumerate}

		Let us prove that this proof labeling scheme for $\p$ is correct, by proving its
		completeness and its soundness.

                \medskip
		
		\noindent{\bf Completeness.}
		Assume that $G$ satisfies $\p$ (and thus $G'$
                satisfies $\p'$), and let us prove that with the assignment of
		the certificates described above, every vertex accepts
                the instance.
		First, no vertex can reject at steps~(i) or~(ii), because the prover correctly wrote in $\tab(u)$ the list of the certificates that all the vertices in $C_u$ would
		receive in the proof labeling scheme for $\p'$ in $G'$.
		Let us look at step~(iii). Since 
		$G'$ satisfies $\p'$, all the vertices of $G'$ accept
                the instance in the labeling scheme for $\p'$ in $G'$.
		Hence, all the vertices in $V_u$ accept, and so $u$
                cannot reject the instance at
		step~(iii).
		Finally, since it did not reject before, $u$ accepts
                the instance at step~(iv).
		
		\medskip
		
		\noindent{\bf Soundness.}
		Assume now that all the vertices of $G$ accept the
                instance given some assignment of 
		certificates, and let us show that in this case $G$
                satisfies~$\p$. By property~2 of the definition of a
                local reduction, it is enough to prove that $G'$ satisfies~$\p'$.
		By property \cref{l3f}, every vertex $u$ can indeed compute the sets $C_u$,
		$V_u$, and the subgraph of $G'$ formed by the vertices in $V_u$ and their incident edges.
		Since $u$ did not reject the instance at step~(i), $\tab(u)$ has exactly one entry for each
		vertex $t \in C_u$.
		
		To prove that $G'$ satisfies $\p'$, it is sufficient to show that
		there exists an assignment of certificates to the vertices of
		$G'$ such that the verification procedure for $\p'$ accepts at every
		vertex. Let $t \in V(G')$. By property \cref{l3a} and \cref{l3b}, there exists
		$u \in V(G)$ such that $t \in C_u$. Let us show that, for every other vertex $w
		\in V(G)$ such that $t \in C_w$, we have $\tab(u)[t] = \tab(w)[t]$.
		By property \cref{l3d}, the subgraph of $G$ induced by the set $\{v \in V(G) \; | \;
		t \in C_v\}$ is connected. So there exists a path $u_1, \ldots, u_k$ in $G$ such
		that $u_1=u$, $u_k=w$, and $t \in C_{u_i}$ for every $i \in \{1, \ldots, k\}$.
		For every $i \in \{1, \ldots, k-1\}$, since $u_i$ did not reject at step~(ii)
		of the verification procedure, we have $\tab(u_i)[t]=\tab(u_{i+1})[t]$. This
		implies $\tab(u)[t]=\tab(w)[t]$.
		
		Hence, for every $t \in V(G')$, we can define $c(t)$ as the
		certificate of $t$ written in $\tab(u)[t]$ for all the vertices $u \in V(G)$ such
		that $t \in C_u$. By the previous discussion, $c(t)$
                exists and is well defined.
		Let us show that with the certificate assignment $c$, all the vertices in
		$G'$ accept the instance.
		Let $t \in V(G')$, and let $u \in V(G)$ such that $t \in V_u$ (such a
		vertex $u$ exists by property \cref{l3b}). At step~(iii), $u$ computes the view of $t$
		in $G'$, with the certificate assignment $c$.
		Then, $u$ simulates the verification procedure of~$t$ (in the certification
		scheme for $\p'$ in $G'$). Since $u$ accepts, then $t$ must accept with the certificates
		given by $c$. Since this is the case for every $t \in V(G')$, this shows
		that $G'$ satisfies $\p'$, and thus $G$ satisfies $\p$, which concludes the proof.
		\qedhere
	\end{proof}

        We now explain how our local reduction framework, namely
        Theorem~\ref{thm:main theorem local reduction} and Corollary~\ref{cor:reduction3col}, can be
        combined with the lower bounds for non-3-colorability from Section~\ref{sec:3col} to provide optimal (or close to optimal) lower
        bounds on the local complexity of a number of
        \textsf{coNP}-hard problems.

        All the problems that we study later in	this section are
        mentioned in the book by Garey and Johnson~\cite{GareyJS76},
        but the reference for each of them is either an unpublished
        result or a personal communication. Thus, we present here some
        reductions that we found, which are maybe the same as the ones
        cited by Garey and Johnson~\cite{GareyJS76}, but to our knowledge, these reductions do not appear in any prior publication. Note that, for the problem of determining whether there exists a cubic subgraph, some reductions were already published (see e.g.~\cite{Stewart94,Stewart97}), but to our knowledge they are all different from the one  we present here (we do not use these previous reductions here, because they are not sufficiently local).

        \subsection{Domatic number}

        A \emph{dominating set} in a graph $G$ is a subset $S$ of
        vertices such that every vertex of $G$ is in $S$ or is
        adjacent to a vertex of $S$.
        The \emph{domatic number} of a graph $G$ is the maximum
        integer $k$
        such that the vertex set of $G$ can be partitioned into $k$
        dominating sets.

	\begin{theorem}
		\label{thm:domatic}
		Let $k \geqslant 2$. The property of having domatic number at most~$k$ has
                local complexity $\Omega(n/\log n)$, even in graphs of maximum degree
		at most~$5k^2-4k$.
	\end{theorem}

	\begin{proof}
		Let $\p$ be the property of having domatic number at least~$k+1$ in graphs of maximum degree $5k^2-4k$. We will show that
		there exists a local reduction from $(k+1)$-colorability in graphs of maximum degree~$5k+1$
		to $\p$ with local expansion~$O(1)$, and the result will follow from
		Corollary~\ref{cor:reduction3col}. Let $G$ be a graph of maximum degree~at
		most~$5k+1$, with $n$ vertices having unique identifiers in $\{1, \ldots, n\}$.
		
		The construction of the graph $G'=f_{\text{$(k+1)$-col}, \p}(G)$ is the following. For each edge $uv\in
		E(G)$, we add a clique~$K_{uv}$ of size $k-1$, complete to both $u$
                and $v$ (see
		Figure~\ref{fig:domatic_edge} for an example).
		
		\begin{figure}[h]
			\centering
			\begin{tikzpicture}[x=0.75pt,y=0.75pt,yscale=-1,xscale=1]
				
				\draw    (350.5,122.5) -- (310.5,87) ;
				\draw    (350.5,122.5) -- (297.5,94) ;
				\draw    (350.5,122.5) -- (326.5,89) ;
				\draw    (335.5,96) -- (283.5,122.5) ;
				\draw    (326.5,83) -- (283.5,122.5) ;
				\draw    (302.5,89) -- (283.5,122.5) ;
				\draw    (68.5,122.5) -- (135.5,122.5) ;
				\draw  [fill={rgb, 255:red, 255; green, 255; blue, 255 }  ,fill opacity=1 ] (64,122.5) .. controls (64,120.01) and (66.01,118) .. (68.5,118) .. controls (70.99,118) and (73,120.01) .. (73,122.5) .. controls (73,124.99) and (70.99,127) .. (68.5,127) .. controls (66.01,127) and (64,124.99) .. (64,122.5) -- cycle ;
				\draw  [fill={rgb, 255:red, 255; green, 255; blue, 255 }  ,fill opacity=1 ] (131,122.5) .. controls (131,120.01) and (133.01,118) .. (135.5,118) .. controls (137.99,118) and (140,120.01) .. (140,122.5) .. controls (140,124.99) and (137.99,127) .. (135.5,127) .. controls (133.01,127) and (131,124.99) .. (131,122.5) -- cycle ;
				\draw    (283.5,122.5) -- (350.5,122.5) ;
				\draw  [fill={rgb, 255:red, 255; green, 255; blue, 255 }  ,fill opacity=1 ] (279,122.5) .. controls (279,120.01) and (281.01,118) .. (283.5,118) .. controls (285.99,118) and (288,120.01) .. (288,122.5) .. controls (288,124.99) and (285.99,127) .. (283.5,127) .. controls (281.01,127) and (279,124.99) .. (279,122.5) -- cycle ;
				\draw  [fill={rgb, 255:red, 255; green, 255; blue, 255 }  ,fill opacity=1 ] (346,122.5) .. controls (346,120.01) and (348.01,118) .. (350.5,118) .. controls (352.99,118) and (355,120.01) .. (355,122.5) .. controls (355,124.99) and (352.99,127) .. (350.5,127) .. controls (348.01,127) and (346,124.99) .. (346,122.5) -- cycle ;
				\draw [line width=1.5]    (178.5,104) -- (235.5,104) ;
				\draw [shift={(239.5,104)}, rotate = 180] [fill={rgb, 255:red, 0; green, 0; blue, 0 }  ][line width=0.08]  [draw opacity=0] (11.61,-5.58) -- (0,0) -- (11.61,5.58) -- cycle    ;
				\draw  [fill={rgb, 255:red, 255; green, 255; blue, 255 }  ,fill opacity=1 ] (292.5,88) .. controls (292.5,80.27) and (303.47,74) .. (317,74) .. controls (330.53,74) and (341.5,80.27) .. (341.5,88) .. controls (341.5,95.73) and (330.53,102) .. (317,102) .. controls (303.47,102) and (292.5,95.73) .. (292.5,88) -- cycle ;
				
				\draw (72.5,127.9) node [anchor=north west][inner sep=0.75pt]    {$u$};
				\draw (139.5,127.9) node [anchor=north west][inner sep=0.75pt]    {$v$};
				\draw (287.5,127.9) node [anchor=north west][inner sep=0.75pt]    {$u$};
				\draw (354.5,127.9) node [anchor=north west][inner sep=0.75pt]    {$v$};
				\draw (303,80.4) node [anchor=north west][inner sep=0.75pt]    {$K_{uv}$};

			\end{tikzpicture}
			
			\caption{The reduction from $(k+1)$-colorability to domatic number at least~$k+1$. The clique $K_{uv}$ contains $k-1$ vertices.}
			\label{fig:domatic_edge}
		\end{figure}
		
		The identifier of each vertex in $G'$ indicates if it is a vertex
		of $G$ (and which one), or a vertex in a clique $K_{uv}$ added for an edge $uv$ of $G$ (together with the identifiers of $u$ and~$v$, and the numbering of this vertex in $K_{uv}$).
		This information can be encoded on $O(\log n)$ bits.
		
		For each vertex $u$ of $G$, we define $C_u:=V_u$ as
                the set of vertices of $G'$ containing $u$ and
		all the vertices of $K_{uv}$ corresponding to an edge $uv$ in $G$ (see Figure~\ref{fig:domatic_graph} for an example with $k=2$).
		
		\begin{figure}[h]
			\centering
			\begin{tikzpicture}[x=0.65pt,y=0.65pt,yscale=-1,xscale=1]
				
				\draw    (332,58.5) -- (393,79.5) ;
				\draw    (393,79.5) -- (414,140.5) ;
				\draw    (373,24.5) -- (414,58.5) ;
				\draw    (373,24.5) -- (332,58.5) ;
				\draw    (297,99.5) -- (332,58.5) ;
				\draw    (297,99.5) -- (332,140.5) ;
				\draw    (414,58.5) -- (449,98.5) ;
				\draw    (414,140.5) -- (449,98.5) ;
				\draw    (332,140.5) -- (373,173.5) ;
				\draw    (373,173.5) -- (414,140.5) ;
				\draw    (58.5,57.5) -- (140.5,57.5) ;
				\draw    (58.5,57.5) -- (58.5,139.5) ;
				\draw    (140.5,57.5) -- (140.5,139.5) ;
				\draw    (58.5,139.5) -- (140.5,139.5) ;
				\draw    (58.5,57.5) -- (140.5,139.5) ;
				\draw  [fill={rgb, 255:red, 255; green, 255; blue, 255 }  ,fill opacity=1 ]
				(53,57.5) .. controls (53,54.46) and (55.46,52) .. (58.5,52) .. controls
				(61.54,52) and (64,54.46) .. (64,57.5) .. controls (64,60.54) and (61.54,63) ..
				(58.5,63) .. controls (55.46,63) and (53,60.54) .. (53,57.5) -- cycle ;
				\draw  [fill={rgb, 255:red, 255; green, 255; blue, 255 }  ,fill opacity=1 ]
				(135,57.5) .. controls (135,54.46) and (137.46,52) .. (140.5,52) .. controls
				(143.54,52) and (146,54.46) .. (146,57.5) .. controls (146,60.54) and
				(143.54,63) .. (140.5,63) .. controls (137.46,63) and (135,60.54) .. (135,57.5)
				-- cycle ;
				\draw  [fill={rgb, 255:red, 255; green, 255; blue, 255 }  ,fill opacity=1 ]
				(53,139.5) .. controls (53,136.46) and (55.46,134) .. (58.5,134) .. controls
				(61.54,134) and (64,136.46) .. (64,139.5) .. controls (64,142.54) and
				(61.54,145) .. (58.5,145) .. controls (55.46,145) and (53,142.54) .. (53,139.5)
				-- cycle ;
				\draw  [fill={rgb, 255:red, 255; green, 255; blue, 255 }  ,fill opacity=1 ,
				line width=1.5] (135,139.5) .. controls (135,136.46) and (137.46,134) ..
				(140.5,134) .. controls (143.54,134) and (146,136.46) .. (146,139.5) .. controls
				(146,142.54) and (143.54,145) .. (140.5,145) .. controls (137.46,145) and
				(135,142.54) .. (135,139.5) -- cycle ;
				\draw    (332,58.5) -- (414,58.5) ;
				\draw    (332,58.5) -- (332,140.5) ;
				\draw    (414,58.5) -- (414,140.5) ;
				\draw    (332,140.5) -- (414,140.5) ;
				\draw    (332,58.5) -- (414,140.5) ;
				\draw  [fill={rgb, 255:red, 255; green, 255; blue, 255 }  ,fill opacity=1 ]
				(326.5,58.5) .. controls (326.5,55.46) and (328.96,53) .. (332,53) .. controls
				(335.04,53) and (337.5,55.46) .. (337.5,58.5) .. controls (337.5,61.54) and
				(335.04,64) .. (332,64) .. controls (328.96,64) and (326.5,61.54) ..
				(326.5,58.5) -- cycle ;
				\draw  [fill={rgb, 255:red, 255; green, 255; blue, 255 }  ,fill opacity=1 ]
				(408.5,58.5) .. controls (408.5,55.46) and (410.96,53) .. (414,53) .. controls
				(417.04,53) and (419.5,55.46) .. (419.5,58.5) .. controls (419.5,61.54) and
				(417.04,64) .. (414,64) .. controls (410.96,64) and (408.5,61.54) ..
				(408.5,58.5) -- cycle ;
				\draw  [fill={rgb, 255:red, 255; green, 255; blue, 255 }  ,fill opacity=1 ]
				(326.5,140.5) .. controls (326.5,137.46) and (328.96,135) .. (332,135) ..
				controls (335.04,135) and (337.5,137.46) .. (337.5,140.5) .. controls
				(337.5,143.54) and (335.04,146) .. (332,146) .. controls (328.96,146) and
				(326.5,143.54) .. (326.5,140.5) -- cycle ;
				\draw  [fill={rgb, 255:red, 255; green, 255; blue, 255 }  ,fill opacity=1 ]
				(408.5,140.5) .. controls (408.5,137.46) and (410.96,135) .. (414,135) ..
				controls (417.04,135) and (419.5,137.46) .. (419.5,140.5) .. controls
				(419.5,143.54) and (417.04,146) .. (414,146) .. controls (410.96,146) and
				(408.5,143.54) .. (408.5,140.5) -- cycle ;
				\draw  [fill={rgb, 255:red, 255; green, 255; blue, 255 }  ,fill opacity=1 ]
				(367.5,24.5) .. controls (367.5,21.46) and (369.96,19) .. (373,19) .. controls
				(376.04,19) and (378.5,21.46) .. (378.5,24.5) .. controls (378.5,27.54) and
				(376.04,30) .. (373,30) .. controls (369.96,30) and (367.5,27.54) ..
				(367.5,24.5) -- cycle ;
				\draw  [fill={rgb, 255:red, 255; green, 255; blue, 255 }  ,fill opacity=1 ]
				(443.5,98.5) .. controls (443.5,95.46) and (445.96,93) .. (449,93) .. controls
				(452.04,93) and (454.5,95.46) .. (454.5,98.5) .. controls (454.5,101.54) and
				(452.04,104) .. (449,104) .. controls (445.96,104) and (443.5,101.54) ..
				(443.5,98.5) -- cycle ;
				\draw  [fill={rgb, 255:red, 255; green, 255; blue, 255 }  ,fill opacity=1 ]
				(291.5,99.5) .. controls (291.5,96.46) and (293.96,94) .. (297,94) .. controls
				(300.04,94) and (302.5,96.46) .. (302.5,99.5) .. controls (302.5,102.54) and
				(300.04,105) .. (297,105) .. controls (293.96,105) and (291.5,102.54) ..
				(291.5,99.5) -- cycle ;
				\draw  [fill={rgb, 255:red, 255; green, 255; blue, 255 }  ,fill opacity=1 ]
				(367.5,173.5) .. controls (367.5,170.46) and (369.96,168) .. (373,168) ..
				controls (376.04,168) and (378.5,170.46) .. (378.5,173.5) .. controls
				(378.5,176.54) and (376.04,179) .. (373,179) .. controls (369.96,179) and
				(367.5,176.54) .. (367.5,173.5) -- cycle ;
				\draw  [fill={rgb, 255:red, 255; green, 255; blue, 255 }  ,fill opacity=1 ]
				(387.5,79.5) .. controls (387.5,76.46) and (389.96,74) .. (393,74) .. controls
				(396.04,74) and (398.5,76.46) .. (398.5,79.5) .. controls (398.5,82.54) and
				(396.04,85) .. (393,85) .. controls (389.96,85) and (387.5,82.54) ..
				(387.5,79.5) -- cycle ;
				\draw  [dash pattern={on 5.63pt off 4.5pt}][line width=1.5]  (394.5,135) .. controls (400.5,124) and (363.5,72) .. (389.5,63) .. controls (415.5,54) and (398.5,102) .. (413.5,113) .. controls (428.5,124) and (435.5,72) .. (456.5,83) .. controls (477.5,94) and (442.5,137.25) .. (425.5,150) .. controls (408.5,162.75) and (367.5,202) .. (357.5,185) .. controls (347.5,168) and (388.5,146) .. (394.5,135) -- cycle ;

				\draw (142.5,148.4) node [anchor=north west][inner sep=0.75pt]    {$u$};
				\draw (450,140) node [anchor=north west][inner sep=0.75pt]    {$V_u$};
				\draw (92,162) node [anchor=north west][inner sep=0.75pt]  [font=\large] 
				{$G$};
				\draw (332,195.4) node [anchor=north west][inner sep=0.75pt]  [font=\large] 
				{$G'$};
				
			\end{tikzpicture}
			\caption{Some graphs $G$ and $G'=f_{\text{$(k+1)$-col}, \p}(G)$, with $k=2$. 
			}
			\label{fig:domatic_graph}
		\end{figure}
		
		Let us show that $G$ admits a proper $(k+1)$-coloring if and only if the domatic
		number of $G'$ is at least~$k+1$. Assume that $G$ is $(k+1)$-colorable. Then,
		we define a ($k+1)$-coloring of the vertices of $G'$ in the following
		way: each vertex of $G$ retains its color in $G'$, and
                for each edge $uv$ in $G$, we color arbitrarily the
                clique $K_{uv}$ with the $k-1$ colors distinct from
                that of $u$ and $v$ (each vertex of the clique is
                assigned a distinct color).
		See Figure~\ref{fig:domatic_partition} for an
		example with $k=2$. Note that by construction, each color class is a dominating
                set in $G'$.
                
		Conversely, assume that the vertices of $G'$ can be
                partitioned into $k+1$ dominating sets. Then, we
                obtain a proper $(k+1)$-coloring of~$G$ by considering
                the restriction of
		this partition to the original vertices of~$G$: indeed, let $uv$ be an edge of~$G$,
		and $w$ be a vertex of $G'$ in the $(k-1)$-clique $K_{uv}$. Since
		there are only $k+1$ vertices in $N[w]$ (which are $u$, $v$ and all the vertices of~$K_{uv}$), these
		$k+1$ vertices must be in different parts of the
                partition, since otherwise $w$ would not be dominated
                by one of the $k+1$ dominating sets. In particular, $u$ and $v$ are in different parts, and the corresponding
		$(k+1)$-coloring of $G$ is indeed proper.
		
		\begin{figure}[h]
			\centering
			\begin{tikzpicture}[x=0.65pt,y=0.65pt,yscale=-1,xscale=1]
				
				\draw    (332,58.5) -- (393,79.5) ;
				\draw    (393,79.5) -- (414,140.5) ;
				\draw    (373,24.5) -- (414,58.5) ;
				\draw    (373,24.5) -- (332,58.5) ;
				\draw    (297,99.5) -- (332,58.5) ;
				\draw    (297,99.5) -- (332,140.5) ;
				\draw    (414,58.5) -- (449,98.5) ;
				\draw    (414,140.5) -- (449,98.5) ;
				\draw    (332,140.5) -- (373,173.5) ;
				\draw    (373,173.5) -- (414,140.5) ;
				\draw    (58.5,57.5) -- (140.5,57.5) ;
				\draw    (58.5,57.5) -- (58.5,139.5) ;
				\draw    (140.5,57.5) -- (140.5,139.5) ;
				\draw    (58.5,139.5) -- (140.5,139.5) ;
				\draw    (58.5,57.5) -- (140.5,139.5) ;
				\draw  [fill={rgb, 255:red, 80; green, 227; blue, 194 }  ,fill opacity=1 ]
				(53,57.5) .. controls (53,54.46) and (55.46,52) .. (58.5,52) .. controls
				(61.54,52) and (64,54.46) .. (64,57.5) .. controls (64,60.54) and (61.54,63) ..
				(58.5,63) .. controls (55.46,63) and (53,60.54) .. (53,57.5) -- cycle ;
				\draw  [fill={rgb, 255:red, 245; green, 166; blue, 35 }  ,fill opacity=1 ]
				(135,57.5) .. controls (135,54.46) and (137.46,52) .. (140.5,52) .. controls
				(143.54,52) and (146,54.46) .. (146,57.5) .. controls (146,60.54) and
				(143.54,63) .. (140.5,63) .. controls (137.46,63) and (135,60.54) .. (135,57.5)
				-- cycle ;
				\draw  [fill={rgb, 255:red, 245; green, 166; blue, 35 }  ,fill opacity=1 ]
				(53,139.5) .. controls (53,136.46) and (55.46,134) .. (58.5,134) .. controls
				(61.54,134) and (64,136.46) .. (64,139.5) .. controls (64,142.54) and
				(61.54,145) .. (58.5,145) .. controls (55.46,145) and (53,142.54) .. (53,139.5)
				-- cycle ;
				\draw  [fill={rgb, 255:red, 144; green, 19; blue, 254 }  ,fill opacity=1 ]
				(135,139.5) .. controls (135,136.46) and (137.46,134) .. (140.5,134) .. controls
				(143.54,134) and (146,136.46) .. (146,139.5) .. controls (146,142.54) and
				(143.54,145) .. (140.5,145) .. controls (137.46,145) and (135,142.54) ..
				(135,139.5) -- cycle ;
				\draw    (332,58.5) -- (414,58.5) ;
				\draw    (332,58.5) -- (332,140.5) ;
				\draw    (414,58.5) -- (414,140.5) ;
				\draw    (332,140.5) -- (414,140.5) ;
				\draw    (332,58.5) -- (414,140.5) ;
				\draw  [fill={rgb, 255:red, 80; green, 227; blue, 194 }  ,fill opacity=1 ]
				(326.5,58.5) .. controls (326.5,55.46) and (328.96,53) .. (332,53) .. controls
				(335.04,53) and (337.5,55.46) .. (337.5,58.5) .. controls (337.5,61.54) and
				(335.04,64) .. (332,64) .. controls (328.96,64) and (326.5,61.54) ..
				(326.5,58.5) -- cycle ;
				\draw  [fill={rgb, 255:red, 245; green, 166; blue, 35 }  ,fill opacity=1 ]
				(408.5,58.5) .. controls (408.5,55.46) and (410.96,53) .. (414,53) .. controls
				(417.04,53) and (419.5,55.46) .. (419.5,58.5) .. controls (419.5,61.54) and
				(417.04,64) .. (414,64) .. controls (410.96,64) and (408.5,61.54) ..
				(408.5,58.5) -- cycle ;
				\draw  [fill={rgb, 255:red, 245; green, 166; blue, 35 }  ,fill opacity=1 ]
				(326.5,140.5) .. controls (326.5,137.46) and (328.96,135) .. (332,135) ..
				controls (335.04,135) and (337.5,137.46) .. (337.5,140.5) .. controls
				(337.5,143.54) and (335.04,146) .. (332,146) .. controls (328.96,146) and
				(326.5,143.54) .. (326.5,140.5) -- cycle ;
				\draw  [fill={rgb, 255:red, 144; green, 19; blue, 254 }  ,fill opacity=1 ]
				(408.5,140.5) .. controls (408.5,137.46) and (410.96,135) .. (414,135) ..
				controls (417.04,135) and (419.5,137.46) .. (419.5,140.5) .. controls
				(419.5,143.54) and (417.04,146) .. (414,146) .. controls (410.96,146) and
				(408.5,143.54) .. (408.5,140.5) -- cycle ;
				\draw  [fill={rgb, 255:red, 144; green, 19; blue, 254 }  ,fill opacity=1 ]
				(367.5,24.5) .. controls (367.5,21.46) and (369.96,19) .. (373,19) .. controls
				(376.04,19) and (378.5,21.46) .. (378.5,24.5) .. controls (378.5,27.54) and
				(376.04,30) .. (373,30) .. controls (369.96,30) and (367.5,27.54) ..
				(367.5,24.5) -- cycle ;
				\draw  [fill={rgb, 255:red, 80; green, 227; blue, 194 }  ,fill opacity=1 ]
				(443.5,98.5) .. controls (443.5,95.46) and (445.96,93) .. (449,93) .. controls
				(452.04,93) and (454.5,95.46) .. (454.5,98.5) .. controls (454.5,101.54) and
				(452.04,104) .. (449,104) .. controls (445.96,104) and (443.5,101.54) ..
				(443.5,98.5) -- cycle ;
				\draw  [fill={rgb, 255:red, 144; green, 19; blue, 254 }  ,fill opacity=1 ]
				(291.5,99.5) .. controls (291.5,96.46) and (293.96,94) .. (297,94) .. controls
				(300.04,94) and (302.5,96.46) .. (302.5,99.5) .. controls (302.5,102.54) and
				(300.04,105) .. (297,105) .. controls (293.96,105) and (291.5,102.54) ..
				(291.5,99.5) -- cycle ;
				\draw  [fill={rgb, 255:red, 80; green, 227; blue, 194 }  ,fill opacity=1 ]
				(367.5,173.5) .. controls (367.5,170.46) and (369.96,168) .. (373,168) ..
				controls (376.04,168) and (378.5,170.46) .. (378.5,173.5) .. controls
				(378.5,176.54) and (376.04,179) .. (373,179) .. controls (369.96,179) and
				(367.5,176.54) .. (367.5,173.5) -- cycle ;
				\draw  [fill={rgb, 255:red, 245; green, 166; blue, 35 }  ,fill opacity=1 ]
				(387.5,79.5) .. controls (387.5,76.46) and (389.96,74) .. (393,74) .. controls
				(396.04,74) and (398.5,76.46) .. (398.5,79.5) .. controls (398.5,82.54) and
				(396.04,85) .. (393,85) .. controls (389.96,85) and (387.5,82.54) ..
				(387.5,79.5) -- cycle ;
				
				\draw (92,199.4) node [anchor=north west][inner sep=0.75pt]  [font=\large] 
				{$G$};
				\draw (332,195.4) node [anchor=north west][inner sep=0.75pt]  [font=\large] 
				{$G'$};
				
			\end{tikzpicture}
			\caption{The transformation from a proper $(k+1)$-coloring of~$G$ to a partition into
				$k+1$ dominating sets of~$G'=f_{(k+1)\text{-col},\p}(G)$, with $k=2$. In the drawing of $G'$, the
				colors correspond to the parts of the partition.}
			\label{fig:domatic_partition}
		\end{figure}
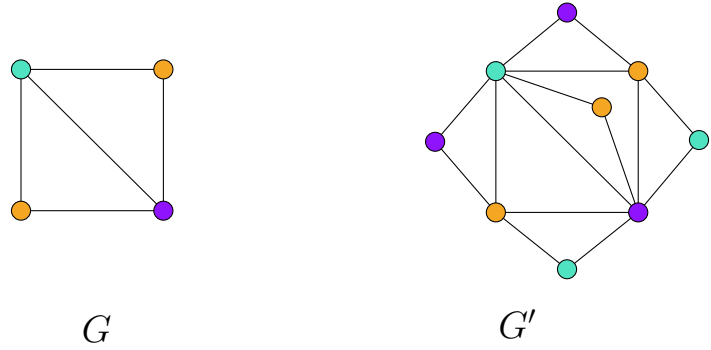
		
		Let us show that the conditions of the
                definition of local reduction are satisfied. We have already seen that \cref{l1} and \cref{l2} both hold. Properties
		\cref{l3a}, \cref{l3b}, \cref{l3e} and \cref{l3f} simply follow from the definition of $V_u$ and $C_u$.
		For \cref{l3c}, since $G$ has maximum degree~at most~$5k+1$,  for all $u \in V(G)$ we
		have $|C_u| \leqslant 1 + (k-1)(5k +1)=5k^2-4k$ and thus $|C_u| = O(1)$. For \cref{l3d}, if $C_u \cap C_v \neq
		\emptyset$, then $u$ and $v$ are neighbors. So all
                conditions of a local reduction are verified, and
		Corollary~\ref{cor:reduction3col} gives us a lower
                bound of order $\Omega(n/\log
			n)$ on the local complexity of the problem. Finally, since $G$ has maximum degree~at most~$5k+1$,  $G'=f_{\text{3-col},
			\p}(G)$ has maximum degree~at most~$5k^2-4k$.
	\end{proof}

        \begin{remark}
		In Section~\ref{sec:variants}, we will prove with Corollary~\ref{cor:reduction3coldeg} that the maximum degree in the statement of Theorem~\ref{thm:domatic} can be decreased to $k^2 + k\lceil\sqrt{k+1}\rceil$.
	\end{remark}
	
	\subsection{Cubic subgraph}
	
	Let $G$ and $H$ be two graphs. We say that $H$ is a \emph{subgraph} (resp.\ \emph{induced subgraph}) of~$G$ if $H$ can be obtained from~$G$ by deleting vertices and edges (resp.\ vertices only). A graph is \emph{cubic} if all its vertices have degree~3.
	
	\begin{theorem}
		\label{thm:cubic_subgraph}
		The local complexity of the property of not containing
                a cubic subgraph is $\Omega(n/\log n)$, even in graphs of maximum degree
		at most~7.
	\end{theorem}
	
	\begin{proof}
		Let $\p$ be the property of having a cubic subgraph. We will show that there is
		a local reduction from 3-colorability in graphs of maximum degree~at most~4 to
		$\p$ with local expansion $O(1)$. The result will then follow from
		Corollary~\ref{cor:reduction3col}. Let $G$ be a graph of maximum degree~at
		most~4, with $n$ vertices having unique identifiers in $\{1, \ldots, n\}$.
		
		Let us describe the graph $G'=f_{\text{3-col}, \p}(G)$. First, we define a \emph{linking
			component with $d$ outputs}. It consists in a cycle of length $2d+1$, and for
		$d$ pairs of consecutive vertices, we add a new vertex linked to both vertices
		of this pair, to form a triangle. The vertex of the cycle which does not belong
		to these pairs is called the \emph{root}. The vertices which are not in the
		cycle are called the \emph{terminals}. The remaining vertices are called the
		\emph{internal vertices}. Note that all internal vertices have degree~3. See
		Figure~\ref{fig:linking} for an example.
		
		\begin{figure}[h]
			\centering
			\begin{tikzpicture}[x=0.6pt,y=0.6pt,yscale=-1,xscale=1]
				
				\draw   (441,202.8) .. controls (441,199.04) and (444.04,196) .. (447.8,196)
				-- (597.7,196) .. controls (601.46,196) and (604.5,199.04) .. (604.5,202.8) --
				(604.5,223.2) .. controls (604.5,226.96) and (601.46,230) .. (597.7,230) --
				(447.8,230) .. controls (444.04,230) and (441,226.96) .. (441,223.2) -- cycle ;
				\draw    (447.8,230) -- (522.75,263) ;
				\draw    (522.75,263) -- (597.7,230) ;
				\draw    (12,184.5) -- (188.75,291.25) ;
				\draw    (188.75,291.25) -- (365.5,184.5) ;
				\draw    (12,184.5) -- (37.25,153) ;
				\draw    (113,184.5) -- (138.25,153) ;
				\draw    (214,184.5) -- (239.25,153) ;
				\draw    (315,184.5) -- (340.25,153) ;
				\draw    (62.5,184.5) -- (37.25,153) ;
				\draw    (163.5,184.5) -- (138.25,153) ;
				\draw    (264.5,184.5) -- (239.25,153) ;
				\draw    (365.5,184.5) -- (340.25,153) ;
				\draw    (12,184.5) -- (365.5,184.5) ;
				\draw  [fill={rgb, 255:red, 208; green, 2; blue, 27 }  ,fill opacity=1 ]
				(181.5,291.25) .. controls (181.5,287.25) and (184.75,284) .. (188.75,284) ..
				controls (192.75,284) and (196,287.25) .. (196,291.25) .. controls (196,295.25)
				and (192.75,298.5) .. (188.75,298.5) .. controls (184.75,298.5) and
				(181.5,295.25) .. (181.5,291.25) -- cycle ;
				\draw  [fill={rgb, 255:red, 248; green, 231; blue, 28 }  ,fill opacity=1 ]
				(206.75,184.5) .. controls (206.75,180.5) and (210,177.25) .. (214,177.25) ..
				controls (218,177.25) and (221.25,180.5) .. (221.25,184.5) .. controls
				(221.25,188.5) and (218,191.75) .. (214,191.75) .. controls (210,191.75) and
				(206.75,188.5) .. (206.75,184.5) -- cycle ;
				\draw  [fill={rgb, 255:red, 248; green, 231; blue, 28 }  ,fill opacity=1 ]
				(156.25,184.5) .. controls (156.25,180.5) and (159.5,177.25) .. (163.5,177.25)
				.. controls (167.5,177.25) and (170.75,180.5) .. (170.75,184.5) .. controls
				(170.75,188.5) and (167.5,191.75) .. (163.5,191.75) .. controls (159.5,191.75)
				and (156.25,188.5) .. (156.25,184.5) -- cycle ;
				\draw  [fill={rgb, 255:red, 248; green, 231; blue, 28 }  ,fill opacity=1 ]
				(307.75,184.5) .. controls (307.75,180.5) and (311,177.25) .. (315,177.25) ..
				controls (319,177.25) and (322.25,180.5) .. (322.25,184.5) .. controls
				(322.25,188.5) and (319,191.75) .. (315,191.75) .. controls (311,191.75) and
				(307.75,188.5) .. (307.75,184.5) -- cycle ;
				\draw  [fill={rgb, 255:red, 248; green, 231; blue, 28 }  ,fill opacity=1 ]
				(257.25,184.5) .. controls (257.25,180.5) and (260.5,177.25) .. (264.5,177.25)
				.. controls (268.5,177.25) and (271.75,180.5) .. (271.75,184.5) .. controls
				(271.75,188.5) and (268.5,191.75) .. (264.5,191.75) .. controls (260.5,191.75)
				and (257.25,188.5) .. (257.25,184.5) -- cycle ;
				\draw  [fill={rgb, 255:red, 248; green, 231; blue, 28 }  ,fill opacity=1 ]
				(105.75,184.5) .. controls (105.75,180.5) and (109,177.25) .. (113,177.25) ..
				controls (117,177.25) and (120.25,180.5) .. (120.25,184.5) .. controls
				(120.25,188.5) and (117,191.75) .. (113,191.75) .. controls (109,191.75) and
				(105.75,188.5) .. (105.75,184.5) -- cycle ;
				\draw  [fill={rgb, 255:red, 248; green, 231; blue, 28 }  ,fill opacity=1 ]
				(55.25,184.5) .. controls (55.25,180.5) and (58.5,177.25) .. (62.5,177.25) ..
				controls (66.5,177.25) and (69.75,180.5) .. (69.75,184.5) .. controls
				(69.75,188.5) and (66.5,191.75) .. (62.5,191.75) .. controls (58.5,191.75) and
				(55.25,188.5) .. (55.25,184.5) -- cycle ;
				\draw  [fill={rgb, 255:red, 248; green, 231; blue, 28 }  ,fill opacity=1 ]
				(4.75,184.5) .. controls (4.75,180.5) and (8,177.25) .. (12,177.25) .. controls
				(16,177.25) and (19.25,180.5) .. (19.25,184.5) .. controls (19.25,188.5) and
				(16,191.75) .. (12,191.75) .. controls (8,191.75) and (4.75,188.5) ..
				(4.75,184.5) -- cycle ;
				\draw  [fill={rgb, 255:red, 248; green, 231; blue, 28 }  ,fill opacity=1 ]
				(358.25,184.5) .. controls (358.25,180.5) and (361.5,177.25) .. (365.5,177.25)
				.. controls (369.5,177.25) and (372.75,180.5) .. (372.75,184.5) .. controls
				(372.75,188.5) and (369.5,191.75) .. (365.5,191.75) .. controls (361.5,191.75)
				and (358.25,188.5) .. (358.25,184.5) -- cycle ;
				\draw  [fill={rgb, 255:red, 255; green, 255; blue, 255 }  ,fill opacity=1 ]
				(30,153) .. controls (30,149) and (33.25,145.75) .. (37.25,145.75) .. controls
				(41.25,145.75) and (44.5,149) .. (44.5,153) .. controls (44.5,157) and
				(41.25,160.25) .. (37.25,160.25) .. controls (33.25,160.25) and (30,157) ..
				(30,153) -- cycle ;
				\draw  [fill={rgb, 255:red, 255; green, 255; blue, 255 }  ,fill opacity=1 ]
				(131,153) .. controls (131,149) and (134.25,145.75) .. (138.25,145.75) ..
				controls (142.25,145.75) and (145.5,149) .. (145.5,153) .. controls (145.5,157)
				and (142.25,160.25) .. (138.25,160.25) .. controls (134.25,160.25) and (131,157)
				.. (131,153) -- cycle ;
				\draw  [fill={rgb, 255:red, 255; green, 255; blue, 255 }  ,fill opacity=1 ]
				(232,153) .. controls (232,149) and (235.25,145.75) .. (239.25,145.75) ..
				controls (243.25,145.75) and (246.5,149) .. (246.5,153) .. controls (246.5,157)
				and (243.25,160.25) .. (239.25,160.25) .. controls (235.25,160.25) and (232,157)
				.. (232,153) -- cycle ;
				\draw  [fill={rgb, 255:red, 255; green, 255; blue, 255 }  ,fill opacity=1 ]
				(333,153) .. controls (333,149) and (336.25,145.75) .. (340.25,145.75) ..
				controls (344.25,145.75) and (347.5,149) .. (347.5,153) .. controls (347.5,157)
				and (344.25,160.25) .. (340.25,160.25) .. controls (336.25,160.25) and (333,157)
				.. (333,153) -- cycle ;
				\draw  [fill={rgb, 255:red, 255; green, 255; blue, 255 }  ,fill opacity=1 ]
				(455,213) .. controls (455,209) and (458.25,205.75) .. (462.25,205.75) ..
				controls (466.25,205.75) and (469.5,209) .. (469.5,213) .. controls (469.5,217)
				and (466.25,220.25) .. (462.25,220.25) .. controls (458.25,220.25) and (455,217)
				.. (455,213) -- cycle ;
				\draw  [fill={rgb, 255:red, 255; green, 255; blue, 255 }  ,fill opacity=1 ]
				(494.25,213) .. controls (494.25,209) and (497.5,205.75) .. (501.5,205.75) ..
				controls (505.5,205.75) and (508.75,209) .. (508.75,213) .. controls
				(508.75,217) and (505.5,220.25) .. (501.5,220.25) .. controls (497.5,220.25) and
				(494.25,217) .. (494.25,213) -- cycle ;
				\draw  [fill={rgb, 255:red, 255; green, 255; blue, 255 }  ,fill opacity=1 ]
				(533.5,213) .. controls (533.5,209) and (536.75,205.75) .. (540.75,205.75) ..
				controls (544.75,205.75) and (548,209) .. (548,213) .. controls (548,217) and
				(544.75,220.25) .. (540.75,220.25) .. controls (536.75,220.25) and (533.5,217)
				.. (533.5,213) -- cycle ;
				\draw  [fill={rgb, 255:red, 255; green, 255; blue, 255 }  ,fill opacity=1 ]
				(572.75,213) .. controls (572.75,209) and (576,205.75) .. (580,205.75) ..
				controls (584,205.75) and (587.25,209) .. (587.25,213) .. controls (587.25,217)
				and (584,220.25) .. (580,220.25) .. controls (576,220.25) and (572.75,217) ..
				(572.75,213) -- cycle ;
				\draw  [fill={rgb, 255:red, 208; green, 2; blue, 27 }  ,fill opacity=1 ]
				(515.5,263) .. controls (515.5,259) and (518.75,255.75) .. (522.75,255.75) ..
				controls (526.75,255.75) and (530,259) .. (530,263) .. controls (530,267) and
				(526.75,270.25) .. (522.75,270.25) .. controls (518.75,270.25) and (515.5,267)
				.. (515.5,263) -- cycle ;
				
			\end{tikzpicture}
			\caption{A linking-component with four outputs, and its symbolic
				representation. The red vertex is the root. The white vertices are the
				terminals. The yellow vertices are the internal vertices (and do not appear in
				the symbolic representation).}
			\label{fig:linking}
		\end{figure}
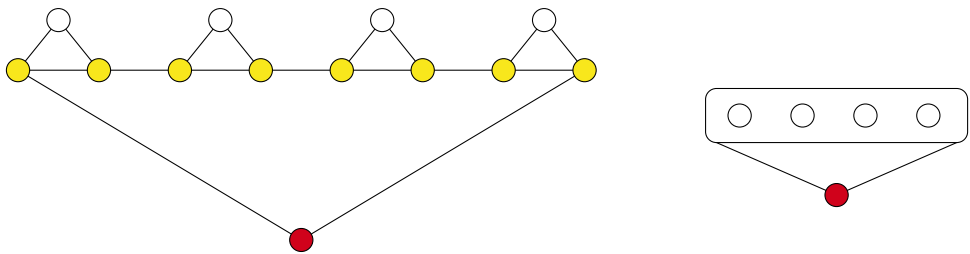
		
		Then, we define an \emph{tag}, represented on Figure~\ref{fig:tag} (this gadget
		has already been used by Stewart~\cite{Stewart94} in a
                different reduction, which was not sufficiently local
                for our purposes). A tag has
		only one vertex of degree~1, that we call the \emph{root}. All the other
		vertices have degree~3, and we call them the \emph{internal vertices}.
		
		\begin{figure}[h]
			\centering
			\begin{tikzpicture}[x=0.6pt,y=0.6pt,yscale=-1,xscale=1]
				
				\draw    (379.5,135) -- (379.5,202.5) ;
				\draw    (264,97.5) -- (311.5,97.5) ;
				\draw    (216.5,97.5) -- (264,97.5) ;
				\draw    (264,45) -- (311.5,97.5) ;
				\draw    (264,45) -- (264,97.5) ;
				\draw    (216.5,97.5) -- (264,45) ;
				\draw    (216.5,97.5) -- (264,150) ;
				\draw    (264,150) -- (311.5,97.5) ;
				\draw    (264,150) -- (264,202.5) ;
				\draw  [fill={rgb, 255:red, 208; green, 2; blue, 27 }  ,fill opacity=1 ]
				(256.75,202.5) .. controls (256.75,198.5) and (260,195.25) .. (264,195.25) ..
				controls (268,195.25) and (271.25,198.5) .. (271.25,202.5) .. controls
				(271.25,206.5) and (268,209.75) .. (264,209.75) .. controls (260,209.75) and
				(256.75,206.5) .. (256.75,202.5) -- cycle ;
				\draw  [fill={rgb, 255:red, 255; green, 255; blue, 255 }  ,fill opacity=1 ]
				(256.75,150) .. controls (256.75,146) and (260,142.75) .. (264,142.75) ..
				controls (268,142.75) and (271.25,146) .. (271.25,150) .. controls (271.25,154)
				and (268,157.25) .. (264,157.25) .. controls (260,157.25) and (256.75,154) ..
				(256.75,150) -- cycle ;
				\draw  [fill={rgb, 255:red, 255; green, 255; blue, 255 }  ,fill opacity=1 ]
				(256.75,97.5) .. controls (256.75,93.5) and (260,90.25) .. (264,90.25) ..
				controls (268,90.25) and (271.25,93.5) .. (271.25,97.5) .. controls
				(271.25,101.5) and (268,104.75) .. (264,104.75) .. controls (260,104.75) and
				(256.75,101.5) .. (256.75,97.5) -- cycle ;
				\draw  [fill={rgb, 255:red, 255; green, 255; blue, 255 }  ,fill opacity=1 ]
				(256.75,45) .. controls (256.75,41) and (260,37.75) .. (264,37.75) .. controls
				(268,37.75) and (271.25,41) .. (271.25,45) .. controls (271.25,49) and
				(268,52.25) .. (264,52.25) .. controls (260,52.25) and (256.75,49) ..
				(256.75,45) -- cycle ;
				\draw  [fill={rgb, 255:red, 255; green, 255; blue, 255 }  ,fill opacity=1 ]
				(304.25,97.5) .. controls (304.25,93.5) and (307.5,90.25) .. (311.5,90.25) ..
				controls (315.5,90.25) and (318.75,93.5) .. (318.75,97.5) .. controls
				(318.75,101.5) and (315.5,104.75) .. (311.5,104.75) .. controls (307.5,104.75)
				and (304.25,101.5) .. (304.25,97.5) -- cycle ;
				\draw  [fill={rgb, 255:red, 255; green, 255; blue, 255 }  ,fill opacity=1 ]
				(209.25,97.5) .. controls (209.25,93.5) and (212.5,90.25) .. (216.5,90.25) ..
				controls (220.5,90.25) and (223.75,93.5) .. (223.75,97.5) .. controls
				(223.75,101.5) and (220.5,104.75) .. (216.5,104.75) .. controls (212.5,104.75)
				and (209.25,101.5) .. (209.25,97.5) -- cycle ;
				\draw  [fill={rgb, 255:red, 0; green, 0; blue, 0 }  ,fill opacity=1 ]
				(379.5,143.25) -- (387.75,135) -- (379.5,126.75) -- (371.25,135) -- cycle ;
				\draw  [fill={rgb, 255:red, 208; green, 2; blue, 27 }  ,fill opacity=1 ]
				(372.25,202.5) .. controls (372.25,198.5) and (375.5,195.25) .. (379.5,195.25)
				.. controls (383.5,195.25) and (386.75,198.5) .. (386.75,202.5) .. controls
				(386.75,206.5) and (383.5,209.75) .. (379.5,209.75) .. controls (375.5,209.75)
				and (372.25,206.5) .. (372.25,202.5) -- cycle ;
				
			\end{tikzpicture}
			\caption{A tag, and its symbolic representation. The red vertex is the root.}
			\label{fig:tag}
		\end{figure}
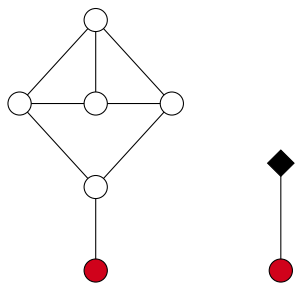
		
		The graph $G'=f_{\text{3-col},\p}(G)$ is constructed as follows. For every vertex $u \in V(G)$
		of degree~$d$ (with $d \in \{1, 2, 3, 4\}$), we create one tag and three
		linking-components with $d$ outputs associated to $u$, and we merge all the four
		roots into a single root vertex.
		We number from~1 to~3 the linking-components of~$u$. 
		In each of the three linking components, we associate one terminal to each
		neighbor of $u$. For each $v \in N(u)$, we denote by $L_u[i,v]$ the terminal of
		the $i$-th linking component of~$u$ associated to~$v$. For each $i,j \in
		\{1,2,3\}$ such that $i \neq j$, and for all $v \in
                N(u)$, we put an edge
		between $L_u[i,v]$ and $L_v[j,u]$. These edges are called the \emph{connecting edges}.
		Note that since $G$ has maximum degree~at most~4, $G'$ has maximum degree~at most~7.
		For each $u \in V(G)$, we define $C_u:=V_u$ as the set of vertices which are in the three linking components and in
		the tag corresponding to $u$.
		Finally, the identifier of every vertex of $G'$ indicates if it belongs to a tag or a linking component, to which vertex of $G$ is it associated, and which vertex it is precisely in this tag or linking component. All this information can be encoded on $O(\log n)$ bits.
		See Figure~\ref{fig:cubicsubgraph_reduction} for an example of this reduction.

		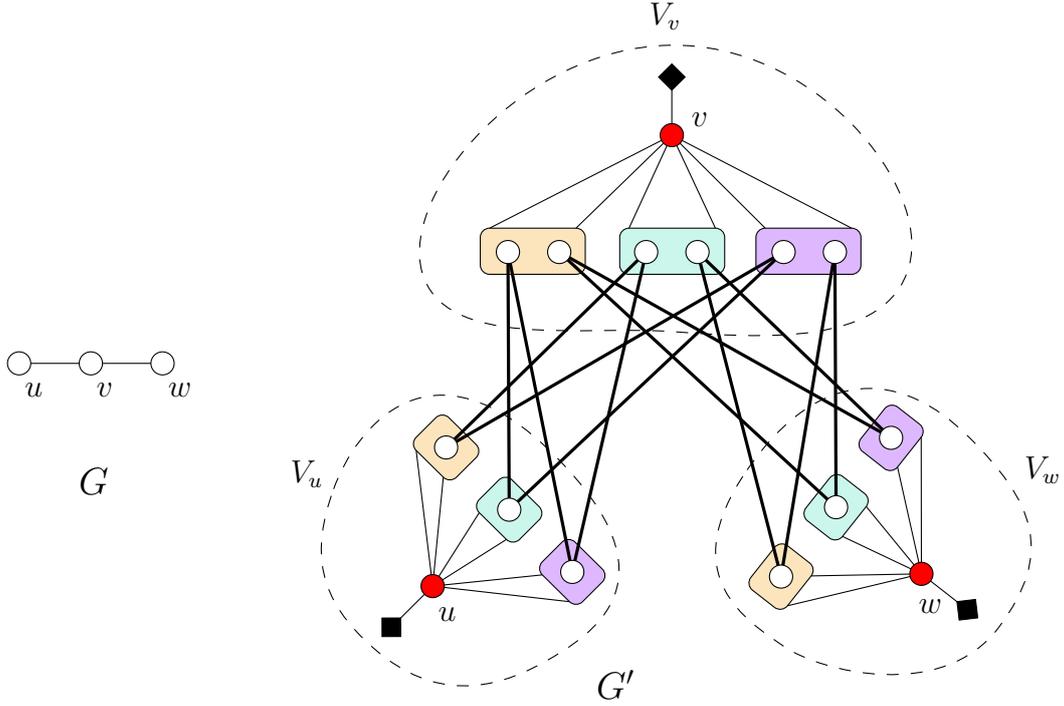
\begin{figure}[h]
			\centering
			\begin{tikzpicture}[x=0.6pt,y=0.6pt,yscale=-1,xscale=1]
				
				\draw  [fill={rgb, 255:red, 144; green, 19; blue, 254 }  ,fill opacity=0.3 ]
				(612.5,268.4) .. controls (609.98,266.42) and (609.55,262.77) .. (611.53,260.26)
				-- (626.31,241.48) .. controls (628.29,238.96) and (631.94,238.52) ..
				(634.45,240.5) -- (648.13,251.26) .. controls (650.65,253.25) and
				(651.08,256.89) .. (649.1,259.41) -- (634.32,278.19) .. controls (632.34,280.71)
				and (628.69,281.14) .. (626.17,279.16) -- cycle ;
				\draw  [fill={rgb, 255:red, 144; green, 19; blue, 254 }  ,fill opacity=0.3 ]
				(424.95,325.44) .. controls (427.21,323.17) and (430.88,323.16) ..
				(433.15,325.42) -- (450.1,342.27) .. controls (452.38,344.53) and (452.39,348.2)
				.. (450.13,350.47) -- (437.86,362.81) .. controls (435.6,365.08) and
				(431.93,365.09) .. (429.66,362.84) -- (412.71,345.99) .. controls
				(410.44,343.73) and (410.43,340.05) .. (412.68,337.78) -- cycle ;
				\draw  [fill={rgb, 255:red, 144; green, 19; blue, 254 }  ,fill opacity=0.3 ]
				(546,133.8) .. controls (546,130.6) and (548.6,128) .. (551.8,128) --
				(605.7,128) .. controls (608.9,128) and (611.5,130.6) .. (611.5,133.8) --
				(611.5,151.2) .. controls (611.5,154.4) and (608.9,157) .. (605.7,157) --
				(551.8,157) .. controls (548.6,157) and (546,154.4) .. (546,151.2) -- cycle ;
				\draw  [fill={rgb, 255:red, 80; green, 227; blue, 194 }  ,fill opacity=0.3 ]
				(578.18,312.02) .. controls (575.66,310.04) and (575.23,306.39) ..
				(577.21,303.87) -- (591.99,285.09) .. controls (593.97,282.57) and
				(597.62,282.14) .. (600.13,284.12) -- (613.81,294.88) .. controls
				(616.32,296.86) and (616.76,300.51) .. (614.78,303.02) -- (600,321.81) ..
				controls (598.02,324.32) and (594.37,324.76) .. (591.85,322.78) -- cycle ;
				\draw  [fill={rgb, 255:red, 80; green, 227; blue, 194 }  ,fill opacity=0.3 ]
				(385.59,286.32) .. controls (387.85,284.04) and (391.52,284.03) ..
				(393.79,286.29) -- (410.74,303.14) .. controls (413.01,305.4) and
				(413.03,309.07) .. (410.77,311.34) -- (398.5,323.68) .. controls (396.24,325.95)
				and (392.57,325.97) .. (390.3,323.71) -- (373.35,306.86) .. controls
				(371.08,304.6) and (371.07,300.93) .. (373.32,298.66) -- cycle ;
				\draw  [fill={rgb, 255:red, 80; green, 227; blue, 194 }  ,fill opacity=0.3 ]
				(461,133.8) .. controls (461,130.6) and (463.6,128) .. (466.8,128) --
				(520.7,128) .. controls (523.9,128) and (526.5,130.6) .. (526.5,133.8) --
				(526.5,151.2) .. controls (526.5,154.4) and (523.9,157) .. (520.7,157) --
				(466.8,157) .. controls (463.6,157) and (461,154.4) .. (461,151.2) -- cycle ;
				\draw  [fill={rgb, 255:red, 245; green, 166; blue, 35 }  ,fill opacity=0.3 ]
				(543.86,355.63) .. controls (541.34,353.65) and (540.91,350) .. (542.89,347.49)
				-- (557.67,328.71) .. controls (559.65,326.19) and (563.29,325.75) ..
				(565.81,327.73) -- (579.49,338.49) .. controls (582,340.48) and (582.44,344.12)
				.. (580.46,346.64) -- (565.68,365.42) .. controls (563.7,367.94) and
				(560.05,368.37) .. (557.53,366.39) -- cycle ;
				\draw  [fill={rgb, 255:red, 245; green, 166; blue, 35 }  ,fill opacity=0.3 ]
				(346.23,247.19) .. controls (348.49,244.92) and (352.16,244.9) ..
				(354.43,247.16) -- (371.38,264.01) .. controls (373.65,266.27) and
				(373.66,269.94) .. (371.41,272.22) -- (359.14,284.56) .. controls
				(356.88,286.83) and (353.21,286.84) .. (350.94,284.58) -- (333.99,267.73) ..
				controls (331.72,265.47) and (331.7,261.8) .. (333.96,259.53) -- cycle ;
				\draw  [fill={rgb, 255:red, 245; green, 166; blue, 35 }  ,fill opacity=0.3 ]
				(374,133.8) .. controls (374,130.6) and (376.6,128) .. (379.8,128) --
				(433.7,128) .. controls (436.9,128) and (439.5,130.6) .. (439.5,133.8) --
				(439.5,151.2) .. controls (439.5,154.4) and (436.9,157) .. (433.7,157) --
				(379.8,157) .. controls (376.6,157) and (374,154.4) .. (374,151.2) -- cycle ;
				\draw[line width=1.2pt]    (391.25,143) -- (392.05,305) ;
				\draw[line width=1.2pt]    (391.25,143) -- (431.41,344.13) ;
				\draw[line width=1.2pt]    (423.25,143) -- (595.99,303.45) ;
				\draw[line width=1.2pt]    (423.25,143) -- (630.31,259.83) ;
				\draw[line width=1.2pt]    (477.5,143) -- (352.68,265.87) ;
				\draw[line width=1.2pt]    (477.5,143) -- (431.41,344.13) ;
				\draw[line width=1.2pt]    (509.5,143) -- (561.67,347.06) ;
				\draw[line width=1.2pt]    (509.5,143) -- (630.31,259.83) ;
				\draw[line width=1.2pt]    (563.25,143) -- (352.68,265.87) ;
				\draw[line width=1.2pt]    (563.25,143) -- (392.05,305) ;
				\draw[line width=1.2pt]    (595.25,143) -- (561.67,347.06) ;
				\draw[line width=1.2pt]   (595.25,143) -- (595.99,303.45) ;
				\draw    (86.25,213) -- (131,213) ;
				\draw    (131,213) -- (175.75,213) ;
				\draw    (344.28,353.05) -- (318.55,378.93) ;
				\draw    (333.99,267.73) -- (344.28,353.05) ;
				\draw    (350.94,284.58) -- (344.28,353.05) ;
				\draw    (373.35,306.86) -- (344.28,353.05) ;
				\draw    (390.3,323.71) -- (344.28,353.05) ;
				\draw    (412.71,345.99) -- (344.28,353.05) ;
				\draw    (429.66,362.84) -- (344.28,353.05) ;
				\draw  [fill={rgb, 255:red, 255; green, 0; blue, 0 }  ,fill opacity=1 ]
				(339.14,347.94) .. controls (341.96,345.1) and (346.55,345.08) ..
				(349.39,347.91) .. controls (352.23,350.73) and (352.25,355.32) ..
				(349.42,358.16) .. controls (346.6,361) and (342.01,361.01) .. (339.17,358.19)
				.. controls (336.33,355.37) and (336.32,350.78) .. (339.14,347.94) -- cycle ;
				\draw  [fill={rgb, 255:red, 0; green, 0; blue, 0 }  ,fill opacity=1 ]
				(312.73,384.78) -- (324.4,384.75) -- (324.36,373.08) -- (312.7,373.12) -- cycle
				;
				\draw  [fill={rgb, 255:red, 255; green, 255; blue, 255 }  ,fill opacity=1 ]
				(386.9,299.89) .. controls (389.73,297.05) and (394.32,297.03) ..
				(397.16,299.86) .. controls (400,302.68) and (400.01,307.27) .. (397.19,310.11)
				.. controls (394.36,312.95) and (389.77,312.96) .. (386.93,310.14) .. controls
				(384.09,307.32) and (384.08,302.73) .. (386.9,299.89) -- cycle ;
				\draw  [fill={rgb, 255:red, 255; green, 255; blue, 255 }  ,fill opacity=1 ]
				(426.26,339.02) .. controls (429.09,336.18) and (433.68,336.16) ..
				(436.52,338.99) .. controls (439.36,341.81) and (439.37,346.4) ..
				(436.55,349.24) .. controls (433.72,352.08) and (429.13,352.09) ..
				(426.29,349.27) .. controls (423.45,346.45) and (423.44,341.86) ..
				(426.26,339.02) -- cycle ;
				\draw  [fill={rgb, 255:red, 255; green, 255; blue, 255 }  ,fill opacity=1 ]
				(347.54,260.76) .. controls (350.37,257.92) and (354.96,257.91) ..
				(357.8,260.73) .. controls (360.64,263.55) and (360.65,268.14) ..
				(357.83,270.98) .. controls (355,273.82) and (350.41,273.84) .. (347.57,271.01)
				.. controls (344.73,268.19) and (344.72,263.6) .. (347.54,260.76) -- cycle ;
				\draw    (493.5,33) -- (493.5,69.5) ;
				\draw    (493.5,69.5) -- (379.8,128) ;
				\draw    (493.5,69.5) -- (433.7,128) ;
				\draw    (493.5,69.5) -- (466.8,128) ;
				\draw    (493.5,69.5) -- (520.7,128) ;
				\draw    (493.5,69.5) -- (551.8,128) ;
				\draw    (493.5,69.5) -- (605.7,128) ;
				\draw  [fill={rgb, 255:red, 255; green, 0; blue, 0 }  ,fill opacity=1 ]
				(486.25,69.5) .. controls (486.25,65.5) and (489.5,62.25) .. (493.5,62.25) ..
				controls (497.5,62.25) and (500.75,65.5) .. (500.75,69.5) .. controls
				(500.75,73.5) and (497.5,76.75) .. (493.5,76.75) .. controls (489.5,76.75) and
				(486.25,73.5) .. (486.25,69.5) -- cycle ;
				\draw  [fill={rgb, 255:red, 0; green, 0; blue, 0 }  ,fill opacity=1 ]
				(493.5,41.25) -- (501.75,33) -- (493.5,24.75) -- (485.25,33) -- cycle ;
				\draw  [fill={rgb, 255:red, 255; green, 255; blue, 255 }  ,fill opacity=1 ]
				(470.25,143) .. controls (470.25,139) and (473.5,135.75) .. (477.5,135.75) ..
				controls (481.5,135.75) and (484.75,139) .. (484.75,143) .. controls
				(484.75,147) and (481.5,150.25) .. (477.5,150.25) .. controls (473.5,150.25) and
				(470.25,147) .. (470.25,143) -- cycle ;
				\draw  [fill={rgb, 255:red, 255; green, 255; blue, 255 }  ,fill opacity=1 ]
				(502.25,143) .. controls (502.25,139) and (505.5,135.75) .. (509.5,135.75) ..
				controls (513.5,135.75) and (516.75,139) .. (516.75,143) .. controls
				(516.75,147) and (513.5,150.25) .. (509.5,150.25) .. controls (505.5,150.25) and
				(502.25,147) .. (502.25,143) -- cycle ;
				\draw  [fill={rgb, 255:red, 255; green, 255; blue, 255 }  ,fill opacity=1 ]
				(556,143) .. controls (556,139) and (559.25,135.75) .. (563.25,135.75) ..
				controls (567.25,135.75) and (570.5,139) .. (570.5,143) .. controls (570.5,147)
				and (567.25,150.25) .. (563.25,150.25) .. controls (559.25,150.25) and (556,147)
				.. (556,143) -- cycle ;
				\draw  [fill={rgb, 255:red, 255; green, 255; blue, 255 }  ,fill opacity=1 ]
				(588,143) .. controls (588,139) and (591.25,135.75) .. (595.25,135.75) ..
				controls (599.25,135.75) and (602.5,139) .. (602.5,143) .. controls (602.5,147)
				and (599.25,150.25) .. (595.25,150.25) .. controls (591.25,150.25) and (588,147)
				.. (588,143) -- cycle ;
				\draw  [fill={rgb, 255:red, 255; green, 255; blue, 255 }  ,fill opacity=1 ]
				(416,143) .. controls (416,139) and (419.25,135.75) .. (423.25,135.75) ..
				controls (427.25,135.75) and (430.5,139) .. (430.5,143) .. controls (430.5,147)
				and (427.25,150.25) .. (423.25,150.25) .. controls (419.25,150.25) and (416,147)
				.. (416,143) -- cycle ;
				\draw  [fill={rgb, 255:red, 255; green, 255; blue, 255 }  ,fill opacity=1 ]
				(384,143) .. controls (384,139) and (387.25,135.75) .. (391.25,135.75) ..
				controls (395.25,135.75) and (398.5,139) .. (398.5,143) .. controls (398.5,147)
				and (395.25,150.25) .. (391.25,150.25) .. controls (387.25,150.25) and (384,147)
				.. (384,143) -- cycle ;
				\draw    (649.23,345.35) -- (677.92,367.92) ;
				\draw    (565.68,365.42) -- (649.23,345.35) ;
				\draw    (580.46,346.64) -- (649.23,345.35) ;
				\draw    (600,321.81) -- (649.23,345.35) ;
				\draw    (614.78,303.02) -- (649.23,345.35) ;
				\draw    (634.32,278.19) -- (649.23,345.35) ;
				\draw    (649.1,259.41) -- (649.23,345.35) ;
				\draw  [fill={rgb, 255:red, 255; green, 0; blue, 0 }  ,fill opacity=1 ]
				(644.75,351.04) .. controls (641.6,348.57) and (641.06,344.01) ..
				(643.54,340.86) .. controls (646.01,337.72) and (650.57,337.17) ..
				(653.72,339.65) .. controls (656.86,342.12) and (657.41,346.68) ..
				(654.93,349.83) .. controls (652.46,352.98) and (647.9,353.52) ..
				(644.75,351.04) -- cycle ;
				\draw  [fill={rgb, 255:red, 0; green, 0; blue, 0 }  ,fill opacity=1 ]
				(684.4,373.02) -- (683.02,361.43) -- (671.44,362.82) -- (672.82,374.4) -- cycle
				;
				\draw  [fill={rgb, 255:red, 255; green, 255; blue, 255 }  ,fill opacity=1 ]
				(591.51,309.15) .. controls (588.36,306.67) and (587.82,302.11) ..
				(590.3,298.96) .. controls (592.77,295.82) and (597.33,295.27) ..
				(600.48,297.75) .. controls (603.62,300.23) and (604.17,304.78) ..
				(601.69,307.93) .. controls (599.21,311.08) and (594.66,311.62) ..
				(591.51,309.15) -- cycle ;
				\draw  [fill={rgb, 255:red, 255; green, 255; blue, 255 }  ,fill opacity=1 ]
				(625.83,265.53) .. controls (622.68,263.05) and (622.14,258.5) ..
				(624.62,255.35) .. controls (627.09,252.2) and (631.65,251.66) .. (634.8,254.14)
				.. controls (637.94,256.61) and (638.49,261.17) .. (636.01,264.32) .. controls
				(633.54,267.46) and (628.98,268.01) .. (625.83,265.53) -- cycle ;
				\draw  [fill={rgb, 255:red, 255; green, 255; blue, 255 }  ,fill opacity=1 ]
				(557.19,352.76) .. controls (554.04,350.28) and (553.5,345.73) ..
				(555.97,342.58) .. controls (558.45,339.43) and (563.01,338.89) ..
				(566.15,341.37) .. controls (569.3,343.84) and (569.84,348.4) .. (567.37,351.55)
				.. controls (564.89,354.69) and (560.33,355.24) .. (557.19,352.76) -- cycle ;
				\draw  [fill={rgb, 255:red, 255; green, 255; blue, 255 }  ,fill opacity=1 ]
				(79,213) .. controls (79,209) and (82.25,205.75) .. (86.25,205.75) .. controls
				(90.25,205.75) and (93.5,209) .. (93.5,213) .. controls (93.5,217) and
				(90.25,220.25) .. (86.25,220.25) .. controls (82.25,220.25) and (79,217) ..
				(79,213) -- cycle ;
				\draw  [fill={rgb, 255:red, 255; green, 255; blue, 255 }  ,fill opacity=1 ]
				(123.75,213) .. controls (123.75,209) and (127,205.75) .. (131,205.75) ..
				controls (135,205.75) and (138.25,209) .. (138.25,213) .. controls (138.25,217)
				and (135,220.25) .. (131,220.25) .. controls (127,220.25) and (123.75,217) ..
				(123.75,213) -- cycle ;
				\draw  [fill={rgb, 255:red, 255; green, 255; blue, 255 }  ,fill opacity=1 ]
				(168.5,213) .. controls (168.5,209) and (171.75,205.75) .. (175.75,205.75) ..
				controls (179.75,205.75) and (183,209) .. (183,213) .. controls (183,217) and
				(179.75,220.25) .. (175.75,220.25) .. controls (171.75,220.25) and (168.5,217)
				.. (168.5,213) -- cycle ;
				
				\draw  [dash pattern={on 4.5pt off 4.5pt}] (282.5,292) .. controls (296.5,253) and (343.5,193) .. (420.5,271) .. controls (497.5,349) and (445.5,383) .. (413.5,401) .. controls (381.5,419) and (345.5,425) .. (308.5,393) .. controls (271.5,361) and (268.5,331) .. (282.5,292) -- cycle ;
				\draw  [dash pattern={on 4.5pt off 4.5pt}] (337.5,155) .. controls (324.5,105) and (401,15) .. (491.75,13) .. controls (582.5,11) and (655.5,101) .. (641.5,152) .. controls (627.5,203) and (548.5,197) .. (493.5,193) .. controls (438.5,189) and (350.5,205) .. (337.5,155) -- cycle ;
				\draw  [dash pattern={on 4.5pt off 4.5pt}] (540.5,280) .. controls (586.5,213) and (637.5,218) .. (678.5,259) .. controls (719.5,300) and (732.5,332) .. (697.5,379) .. controls (662.5,426) and (586.5,407) .. (563.5,391) .. controls (540.5,375) and (494.5,347) .. (540.5,280) -- cycle ;

				\draw (88.25,223.65) node [anchor=north west][inner sep=0.75pt]    {$u$};
				\draw (133,223.65) node [anchor=north west][inner sep=0.75pt]    {$v$};
				\draw (177.75,223.65) node [anchor=north west][inner sep=0.75pt]    {$w$};
				\draw (346.17,364.59) node [anchor=north west][inner sep=0.75pt]    {$u$};
				\draw (504.25,53.65) node [anchor=north west][inner sep=0.75pt]    {$v$};
				\draw (646,359.44) node [anchor=north west][inner sep=0.75pt]    {$w$};
				\draw (122,277.4) node [anchor=north west][inner sep=0.75pt]    {\large $G$};
				\draw (445,404.4) node [anchor=north west][inner sep=0.75pt]    {\large
					$G'$};
				\draw (254,272.4) node [anchor=north west][inner sep=0.75pt]    {$V_{u}$};
				\draw (478,-15) node [anchor=north west][inner sep=0.75pt]    {$V_{v}$};
				\draw (712,270) node [anchor=north west][inner sep=0.75pt]    {$V_{w}$};

			\end{tikzpicture}
			\caption{The graph $G'=f_{\text{3-col},\p}(G)$, where $G$ is a path on three vertices.
				The red vertices are the roots. The thick edges are the connecting edges. The colors of the linking components represent their number.}
			\label{fig:cubicsubgraph_reduction}
		\end{figure}

		Let us prove that $G$ is 3-colorable if and only if $G'$ has a
		cubic subgraph. First, assume that $G$ is 3-colorable, and let $c : V(G)
		\to \{1,2,3\}$ be a proper 3-coloring of~$G$. Let us describe a
		cubic subgraph of $G'$. In fact, it will be an induced subgraph:
		for this reason, we just describe the subset $S \subseteq V(G')$
		which defines it.
		For each vertex $u \in V(G)$, we put in $S$ all the
                vertices from its tag, and all
		the vertices from its $c(u)$-th linking component. Let us prove that the subgraph
		of $G'$ induced by $S$ is cubic. We already know that all the
		internal vertices of the tags and linking components have degree~3. In
		$G'[S]$, each root also has degree~3, because its tag and exactly
		one of its linking components is included in $S$. Finally, for every $u \in
		V(G)$, let us show that all the terminals of the $c(u)$-th linking component of
		$u$ also have degree~3. Let $v$ be a neighbor of $u$ in~$G$. The vertex
		$L_u[c(u),v]$ has already two neighbors inside the $c(u)$-th linking component
		of $u$. Moreover, since $c$ is a proper coloring of~$G$, we have $c(u) \neq
		c(v)$. By definition of $G'$, there is an edge between
		$L_u[c(u),v]$ and $L_v[c(v),u]$, so $L_u[c(u),v]$ has indeed degree~3 in
		$G'[S]$, which is thus a cubic (induced) subgraph
		of~$G'$.
		
		Conversely, assume that $G'$ has a cubic subgraph, and let us
		prove that $G$ is 3-colorable. Let $H$ be a cubic subgraph of $G'$. We will need
		the following lemma.
		
		\begin{lemma}
			\label{lem:cubic_subgraph main lemma}
			For every vertex $u \in V(G)$, $H$ contains all the vertices and edges of the
			tag and exactly one linking component associated to it, and no other vertices in the two other linking components sharing the same root.
		\end{lemma}
		
		Lemma~\ref{lem:cubic_subgraph main lemma} will itself be proved using the
		following claim.
		
		\begin{claim}
			\label{claim:cubic_subgraph} The following holds.
			\begin{enumerate}[(1)]
				\item If $V(H)$ contains an internal vertex of a linking component (resp.\ of
				a tag), then $H$ contains all the vertices and edges of this linking component
				(resp.\ of this tag).\label{5.7.1}
				
				\item If $V(H)$ contains a root vertex, then $H$ contains all the vertices
				and edges of the tag adjacent to it, and all the vertices and edges of exactly
				one linking component adjacent to it. \label{5.7.2}
				
				\item If $V(H)$ contains a terminal vertex of some linking component, then
				$H$ contains all the vertices and edges of this linking component. \label{5.7.3}
			\end{enumerate}
		\end{claim}
		
		\begin{proof}[Proof of Claim~\ref{claim:cubic_subgraph}.]
			\begin{enumerate}[(1)]
				\item Assume that $V(H)$ contains an internal vertex $t$ of some linking
				component. Since $t$ has degree~3, and $H$ is cubic,  $H$ contains all the
				three edges adjacent to $t$, and all its three neighbors.
				Since some of the neighbors of $t$ are also internal vertices, we can apply
				the same argument for them.
				Finally, we obtain that $H$ contains all the internal vertices of this
				linking component, all the edges adjacent to them, and all their neighbors.
				Thus, it contains all the vertices and edges of the linking component.
				The same proof also holds for the tags.
				
				\item Assume that $V(H)$ contains a root vertex $t$. This vertex $t$ has two neighbors in every adjacent linking component, and one neighbor in the tag. Since $t$ has degree~3 in $H$,  $H$ should contain exactly one linking component adjacent to it, and the tag (by \eqref{5.7.1}, all the linking component and the tag are in $H$).
				
				\item Assume that $V(H)$ contains a terminal vertex $t$ of some linking component. Since $t$ is adjacent to only two connecting edges, $H$ contains at least one edge adjacent to $t$ which is in the linking component. It contains also the other endpoint of this edge, which is an internal vertex, so by \eqref{5.7.1} it contains the whole linking component.\qedhere
			\end{enumerate}
		\end{proof}
		
		\begin{proof}[Proof of Lemma~\ref{lem:cubic_subgraph main lemma}]
			By \eqref{5.7.1}, \eqref{5.7.2}, and
                        \eqref{5.7.3} of
                        Claim~\ref{claim:cubic_subgraph}, for every $u
                        \in V(G)$, if \mbox{$H \cap V_u \neq
                          \emptyset$}, then $H$ contains a whole
                        linking component associated to $u$, it
                        contains also the whole tag sharing the same
                        root, and it has no other vertex from the two
                        others linking components sharing the same
                        root. Thus, to prove
                        Lemma~\ref{lem:cubic_subgraph main lemma}, we
                        just have to show that $H \cap V_u \neq
                        \emptyset$ for all $u \in V(G)$. Since $G$ is
                        connected, we just have to show that $H \cap
                        V_u \neq \emptyset$ implies $H \cap V_v \neq
                        \emptyset$ for every vertex neighbor $v$ of
                        $u$ in $G$. So let $u \in V(G)$ such that $H \cap V_u \neq \emptyset$. Let $v$ be a neighbor of $u$, and let $i \in \{1,2,3\}$ such that $H$ contains a vertex in the $i$-th linking component associated to $u$. Then, $L_u[i,v] \in V(H)$, and since it has degree~3 in~$H$ but only two neighbors in this linking component, then $H$ contains some connecting edge adjacent to $L_u[i,v]$, which has $L_v[j,u]$ as the other endpoint, for some $j \in \{1,2,3\} \setminus \{i\}$. So $H \cap V_j \neq \emptyset$, which concludes the proof.
		\end{proof}

		Using Lemma~\ref{lem:cubic_subgraph main lemma}, we
                will show that $G$ is 3-colorable. By
                Lemma~\ref{lem:cubic_subgraph main lemma}, for every
                $u \in V(G)$, $H$ contains all the vertices and edges
                of a unique linking component associated to it. Let
                $c(u)$ be the number of this linking component. We
                will show that $c$ defines a proper 3-coloring of
                $G$. Let $u,v$ be two neighbors in~$G$. Since
                $L_u[c(u),v]$ and $L_v[c(v),u]$ both have only two
                neighbors in their linking component, and have
                degree~3 in $H$, then $H$ contains a connecting edge
                between them. In particular, by definition of~$H$, we
                have $c(u) \neq c(v)$. So $c$ is indeed a proper
                coloring.

                \smallskip
		
		To finish the proof of
                Theorem~\ref{thm:cubic_subgraph}, we just have to show
                that the conditions of the definition of a local
                reduction are satisfied. By definition of the
                identifiers of $G'$, \cref{l1} is satisfied. For
                \cref{l2}, we just proved that it holds as
                well. Conditions \cref{l3a}, \cref{l3b}, \cref{l3e}
                and \cref{l3f} hold by definition. For \cref{l3c},
                it is also true by definition with $\alpha$ being a
                constant function taking the value 42 (because, by
                assumption $G$ has maximum degree~4, and just by
                counting the maximum size of $C_u$). Finally,
                \cref{l3c} holds as well because if $u \neq v$, we have $C_u \neq C_v$.
	\end{proof}
	
	\begin{remark}
		The proof of Theorem~\ref{thm:cubic_subgraph} shows that $G'$ has a cubic subgraph if and only if it has a cubic \emph{induced} subgraph. Thus, we obtain the same $\Omega(n/\log n)$ lower bound to certify that a graph does not have a cubic induced subgraph, even for graphs of maximum degree~7.
	\end{remark}
	
\subsection{Acyclic partition}
        
	\begin{theorem}
		\label{thm:acyclic_partition}
		For any fixed $k \geqslant 3$, the property that  the
                vertex set of a given graph $G$ cannot be partitioned into $k$
                induced forests has local complexity  $\Omega(n/\log n)$.
	\end{theorem}
	
	\begin{proof}
		Let $\p$ be the property of having a partition of the
                vertex set into $k$ acyclic induced subgraphs (i.e.,
                $k$ induced forests). We will show that there is a local reduction from $k$-colorability to $\p$ with local expansion~$O(n)$ and global expansion~$O(n)$, and the result will then follow from Corollary~\ref{cor:reduction3col}.
		Let $G$ be an $n$-vertex graph having unique identifiers in~$\{1, \ldots, n\}$. The graph $G' = f_{k-col,\p}(G)$ is obtained by adding a set $U$ of $k$ universal vertices to~$G$. For every $u \in V(G)$, we define $V_u:=U \cup \{u\}$ and $C_u:=V(G')$.
		
		We only prove that Property~\cref{l2} of the definition of a
                local reduction  is satisfied, since the other
                properties are straightforward to verify. Assume that
                $G$ is $k$-colorable. Then, we can partition the
                vertex set of $G'$ into $k$ acyclic induced subgraphs,
                each of which corresponds to one color class together
                with one of the universal vertices in~$U$. Indeed, by
                doing so, each resulting induced subgraph is a star, which is acyclic.
		Conversely, assume that $V(G')$ can be partitioned
                into $k$ acyclic induced subgraphs. Then, $G'$ is
                $2k$-colorable (since each of these acyclic subgraphs
                is $2$-colorable). Consider any proper $2k$-coloring
                of $G'$. Then, its restriction to $G$ is a proper $k$-coloring of $G$, since each vertex in~$U$ must be alone in its color class (because it is a universal vertex), so $G$ is $k$-colorable.
	\end{proof}

	\subsection{Monochromatic triangles}

	\begin{theorem}
		\label{thm:monochromatic_triangle}
		The property of having no 2-edge-coloring without monochromatic triangles has local
                complexity $\Omega(n/\log n)$, even in graphs of maximum degree~40.
	\end{theorem}
	
	\begin{proof}
		Let $\p$ be the property of the existence of a $2$-edge-coloring without any monochromatic triangle. We will prove that there exists a local reduction of local expansion~$O(1)$ from 3-colorability in graphs of maximum degree~$4$ to $\p$. The result will then follow from Corollary~\ref{cor:reduction3col}.
		
		Let $G$ be a graph of maximum degree~at most~$4$, with $n$ vertices having unique identifiers in $\{1, \ldots, n\}$.
		Let us describe the graph $G' = f_{\text{3-col},
                  \p}(G)$. This graph has two parts, denoted by $G_0$
                and $G_1$, which are the following. $G_0$ is obtained
                by taking a triangle~$T^0_u$ for each vertex~$u$
                of~$G$, and creating a new edge $e_{uv}$ that is
                complete to both $T^0_u$ and $T^0_v$ for every edge
                $uv$ of $G$. See Figure~\ref{fig:G0_triangles} for an example.

		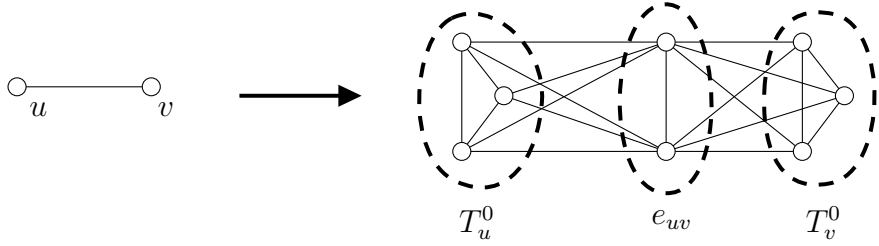
\begin{figure}[h!]
			\centering
			
			\begin{tikzpicture}[x=0.75pt,y=0.75pt,yscale=-1,xscale=1]
				
				\draw    (479.5,97) -- (411.5,152) ;
				\draw    (479.5,97) -- (411.5,97) ;
				\draw    (479.5,152) -- (411.5,152) ;
				\draw    (500.5,124) -- (411.5,97) ;
				\draw    (500.5,124) -- (411.5,152) ;
				\draw    (479.5,152) -- (411.5,97) ;
				\draw    (411.5,97) -- (309.5,97) ;
				\draw    (411.5,152) -- (309.5,97) ;
				\draw    (411.5,152) -- (330.5,124) ;
				\draw    (411.5,97) -- (330.5,124) ;
				\draw    (411.5,97) -- (309.5,152) ;
				\draw    (411.5,152) -- (309.5,152) ;
				\draw    (330.5,124) -- (309.5,152) ;
				\draw    (309.5,97) -- (330.5,124) ;
				\draw    (479.5,152) -- (500.5,124) ;
				\draw    (87.5,119.5) -- (154.5,119.5) ;
				\draw  [fill={rgb, 255:red, 255; green, 255; blue, 255 }  ,fill opacity=1 ] (83,119.5) .. controls (83,117.01) and (85.01,115) .. (87.5,115) .. controls (89.99,115) and (92,117.01) .. (92,119.5) .. controls (92,121.99) and (89.99,124) .. (87.5,124) .. controls (85.01,124) and (83,121.99) .. (83,119.5) -- cycle ;
				\draw  [fill={rgb, 255:red, 255; green, 255; blue, 255 }  ,fill opacity=1 ] (150,119.5) .. controls (150,117.01) and (152.01,115) .. (154.5,115) .. controls (156.99,115) and (159,117.01) .. (159,119.5) .. controls (159,121.99) and (156.99,124) .. (154.5,124) .. controls (152.01,124) and (150,121.99) .. (150,119.5) -- cycle ;
				\draw    (309.5,97) -- (309.5,152) ;
				\draw  [fill={rgb, 255:red, 255; green, 255; blue, 255 }  ,fill opacity=1 ] (305,152) .. controls (305,149.51) and (307.01,147.5) .. (309.5,147.5) .. controls (311.99,147.5) and (314,149.51) .. (314,152) .. controls (314,154.49) and (311.99,156.5) .. (309.5,156.5) .. controls (307.01,156.5) and (305,154.49) .. (305,152) -- cycle ;
				\draw  [fill={rgb, 255:red, 255; green, 255; blue, 255 }  ,fill opacity=1 ] (326,124) .. controls (326,121.51) and (328.01,119.5) .. (330.5,119.5) .. controls (332.99,119.5) and (335,121.51) .. (335,124) .. controls (335,126.49) and (332.99,128.5) .. (330.5,128.5) .. controls (328.01,128.5) and (326,126.49) .. (326,124) -- cycle ;
				\draw  [fill={rgb, 255:red, 255; green, 255; blue, 255 }  ,fill opacity=1 ] (305,97) .. controls (305,94.51) and (307.01,92.5) .. (309.5,92.5) .. controls (311.99,92.5) and (314,94.51) .. (314,97) .. controls (314,99.49) and (311.99,101.5) .. (309.5,101.5) .. controls (307.01,101.5) and (305,99.49) .. (305,97) -- cycle ;
				\draw [line width=1.5]    (198.5,124) -- (255.5,124) ;
				\draw [shift={(259.5,124)}, rotate = 180] [fill={rgb, 255:red, 0; green, 0; blue, 0 }  ][line width=0.08]  [draw opacity=0] (11.61,-5.58) -- (0,0) -- (11.61,5.58) -- cycle    ;
				\draw    (411.5,97) -- (411.5,152) ;
				\draw    (479.5,97) -- (479.5,152) ;
				\draw    (479.5,97) -- (500.5,124) ;
				\draw  [fill={rgb, 255:red, 255; green, 255; blue, 255 }  ,fill opacity=1 ] (407,97) .. controls (407,94.51) and (409.01,92.5) .. (411.5,92.5) .. controls (413.99,92.5) and (416,94.51) .. (416,97) .. controls (416,99.49) and (413.99,101.5) .. (411.5,101.5) .. controls (409.01,101.5) and (407,99.49) .. (407,97) -- cycle ;
				\draw  [fill={rgb, 255:red, 255; green, 255; blue, 255 }  ,fill opacity=1 ] (407,152) .. controls (407,149.51) and (409.01,147.5) .. (411.5,147.5) .. controls (413.99,147.5) and (416,149.51) .. (416,152) .. controls (416,154.49) and (413.99,156.5) .. (411.5,156.5) .. controls (409.01,156.5) and (407,154.49) .. (407,152) -- cycle ;
				\draw  [fill={rgb, 255:red, 255; green, 255; blue, 255 }  ,fill opacity=1 ] (475,97) .. controls (475,94.51) and (477.01,92.5) .. (479.5,92.5) .. controls (481.99,92.5) and (484,94.51) .. (484,97) .. controls (484,99.49) and (481.99,101.5) .. (479.5,101.5) .. controls (477.01,101.5) and (475,99.49) .. (475,97) -- cycle ;
				\draw  [fill={rgb, 255:red, 255; green, 255; blue, 255 }  ,fill opacity=1 ] (475,152) .. controls (475,149.51) and (477.01,147.5) .. (479.5,147.5) .. controls (481.99,147.5) and (484,149.51) .. (484,152) .. controls (484,154.49) and (481.99,156.5) .. (479.5,156.5) .. controls (477.01,156.5) and (475,154.49) .. (475,152) -- cycle ;
				\draw  [fill={rgb, 255:red, 255; green, 255; blue, 255 }  ,fill opacity=1 ] (496,124) .. controls (496,121.51) and (498.01,119.5) .. (500.5,119.5) .. controls (502.99,119.5) and (505,121.51) .. (505,124) .. controls (505,126.49) and (502.99,128.5) .. (500.5,128.5) .. controls (498.01,128.5) and (496,126.49) .. (496,124) -- cycle ;
				\draw  [dash pattern={on 5.63pt off 4.5pt}][line width=1.5]  (304.5,85) .. controls (324.5,75) and (351.5,99) .. (349.5,127) .. controls (347.5,155) and (332.5,179) .. (306.5,167) .. controls (280.5,155) and (284.5,95) .. (304.5,85) -- cycle ;
				\draw  [dash pattern={on 5.63pt off 4.5pt}][line width=1.5]  (411,78) .. controls (439.5,77) and (444,174) .. (411.5,173) .. controls (379,172) and (382.5,79) .. (411,78) -- cycle ;
				\draw  [dash pattern={on 5.63pt off 4.5pt}][line width=1.5]  (460.5,122) .. controls (461.5,100) and (468.5,77) .. (495.5,83) .. controls (522.5,89) and (522.5,162) .. (493.5,168) .. controls (464.5,174) and (459.5,144) .. (460.5,122) -- cycle ;
				
				\draw (92.5,125.9) node [anchor=north west][inner sep=0.75pt]    {$u$};
				\draw (156.5,125.9) node [anchor=north west][inner sep=0.75pt]    {$v$};
				\draw (307,178.4) node [anchor=north west][inner sep=0.75pt]    {$T^0_{u}$};
				\draw (480,178.4) node [anchor=north west][inner sep=0.75pt]    {$T^0_{v}$};
				\draw (403,180.4) node [anchor=north west][inner sep=0.75pt]    {$e_{uv}$};

			\end{tikzpicture}
			
			\caption{The construction of $G_0$.}
			\label{fig:G0_triangles}
		\end{figure}
		
		Then, $G_1$ is obtained by taking seven triangles
                $T^1_u$, $T^2_u$, $T^3_u$, $T^{12}_u$, $T^{23}_u$,
                $T^{13}_u$, $T^{123}_u$ for every vertex~$u$ of~$G$,
                and adding three triangles $T^1_{uv}$, $T^2_{uv}$
                $T^3_{uv}$ for every edge $uv$ of $G$. Finally, we add the following edges in $G_1$, and between $G_0$ and $G_1$:
		\begin{itemize}
			\item for every $u \in V(G)$, we make one edge of $T^{123}_u$ complete to $T^1_u$, one complete to $T^2_u$ and one complete to $T^3_u$ (see Figure~\ref{fig:G1_triangles} for an illustration);
			\item for every $u \in V(G)$ and $1 \leqslant i < j \leqslant 3$, we make one edge of $T_u^{ij}$ complete to $T^i_u$, one complete to $T^j_u$, and one complete to $T^0_u$;
			\item for every $\{u,v\} \in E(G)$ and $i \in \{1,2,3\}$, we make one edge of $T^i_{uv}$ complete to $T^i_u$, one complete to $T^i_v$, and one complete to both $T^0_u$ and $T^0_v$.
		\end{itemize}
		
		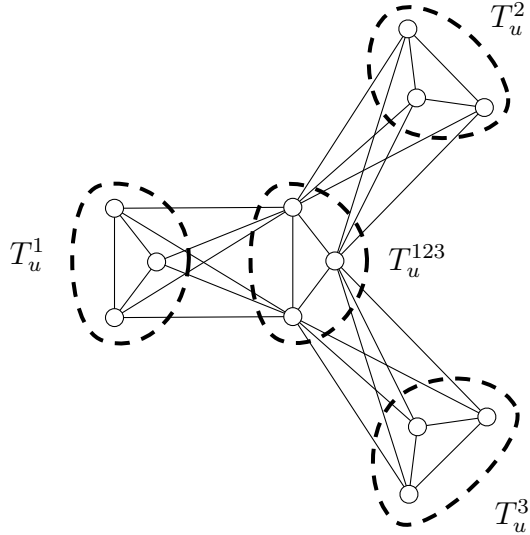
\begin{figure}[h!]
			\centering
			
			\begin{tikzpicture}[x=0.75pt,y=0.75pt,yscale=-1,xscale=1]
				
				\draw    (300.5,228) -- (358.48,317.37) ;
				\draw    (300.5,228) -- (362.72,283.43) ;
				\draw    (300.5,228) -- (397.37,278.48) ;
				\draw    (321.5,200) -- (358.48,317.37) ;
				\draw    (321.5,200) -- (362.72,283.43) ;
				\draw    (321.5,200) -- (397.37,278.48) ;
				\draw    (357.92,83.28) -- (300.5,173) ;
				\draw    (362.26,118.01) -- (300.5,173) ;
				\draw    (396.12,122.84) -- (300.5,173) ;
				\draw    (357.92,83.28) -- (321.5,200) ;
				\draw    (362.26,118.01) -- (321.5,200) ;
				\draw    (396.12,122.84) -- (321.5,200) ;
				\draw    (211.5,173.5) -- (300.5,173) ;
				\draw    (232.5,200.5) -- (300.5,173) ;
				\draw    (211.5,228.5) -- (300.5,173) ;
				\draw    (211.5,173.5) -- (300.5,228) ;
				\draw    (232.5,200.5) -- (300.5,228) ;
				\draw    (211.5,228.5) -- (300.5,228) ;
				\draw    (232.5,200.5) -- (211.5,228.5) ;
				\draw    (211.5,173.5) -- (232.5,200.5) ;
				\draw    (211.5,173.5) -- (211.5,228.5) ;
				\draw  [fill={rgb, 255:red, 255; green, 255; blue, 255 }  ,fill opacity=1 ] (207,228.5) .. controls (207,226.01) and (209.01,224) .. (211.5,224) .. controls (213.99,224) and (216,226.01) .. (216,228.5) .. controls (216,230.99) and (213.99,233) .. (211.5,233) .. controls (209.01,233) and (207,230.99) .. (207,228.5) -- cycle ;
				\draw  [fill={rgb, 255:red, 255; green, 255; blue, 255 }  ,fill opacity=1 ] (228,200.5) .. controls (228,198.01) and (230.01,196) .. (232.5,196) .. controls (234.99,196) and (237,198.01) .. (237,200.5) .. controls (237,202.99) and (234.99,205) .. (232.5,205) .. controls (230.01,205) and (228,202.99) .. (228,200.5) -- cycle ;
				\draw  [fill={rgb, 255:red, 255; green, 255; blue, 255 }  ,fill opacity=1 ] (207,173.5) .. controls (207,171.01) and (209.01,169) .. (211.5,169) .. controls (213.99,169) and (216,171.01) .. (216,173.5) .. controls (216,175.99) and (213.99,178) .. (211.5,178) .. controls (209.01,178) and (207,175.99) .. (207,173.5) -- cycle ;
				\draw    (362.26,118.01) -- (357.92,83.28) ;
				\draw    (396.12,122.84) -- (362.26,118.01) ;
				\draw    (396.12,122.84) -- (357.92,83.28) ;
				\draw  [fill={rgb, 255:red, 255; green, 255; blue, 255 }  ,fill opacity=1 ] (361.16,80.15) .. controls (362.88,81.94) and (362.83,84.79) .. (361.04,86.52) .. controls (359.26,88.24) and (356.41,88.19) .. (354.68,86.4) .. controls (352.95,84.62) and (353,81.77) .. (354.79,80.04) .. controls (356.58,78.31) and (359.43,78.36) .. (361.16,80.15) -- cycle ;
				\draw  [fill={rgb, 255:red, 255; green, 255; blue, 255 }  ,fill opacity=1 ] (365.5,114.88) .. controls (367.23,116.67) and (367.18,119.52) .. (365.39,121.24) .. controls (363.6,122.97) and (360.75,122.92) .. (359.03,121.13) .. controls (357.3,119.35) and (357.35,116.5) .. (359.14,114.77) .. controls (360.92,113.04) and (363.77,113.09) .. (365.5,114.88) -- cycle ;
				\draw  [fill={rgb, 255:red, 255; green, 255; blue, 255 }  ,fill opacity=1 ] (399.36,119.72) .. controls (401.09,121.5) and (401.04,124.35) .. (399.25,126.08) .. controls (397.46,127.81) and (394.61,127.76) .. (392.89,125.97) .. controls (391.16,124.18) and (391.21,121.33) .. (393,119.6) .. controls (394.79,117.88) and (397.63,117.93) .. (399.36,119.72) -- cycle ;
				
				\draw    (321.5,200) -- (300.5,228) ;
				\draw    (300.5,173) -- (321.5,200) ;
				\draw    (300.5,173) -- (300.5,228) ;
				\draw  [fill={rgb, 255:red, 255; green, 255; blue, 255 }  ,fill opacity=1 ] (296,228) .. controls (296,225.51) and (298.01,223.5) .. (300.5,223.5) .. controls (302.99,223.5) and (305,225.51) .. (305,228) .. controls (305,230.49) and (302.99,232.5) .. (300.5,232.5) .. controls (298.01,232.5) and (296,230.49) .. (296,228) -- cycle ;
				\draw  [fill={rgb, 255:red, 255; green, 255; blue, 255 }  ,fill opacity=1 ] (317,200) .. controls (317,197.51) and (319.01,195.5) .. (321.5,195.5) .. controls (323.99,195.5) and (326,197.51) .. (326,200) .. controls (326,202.49) and (323.99,204.5) .. (321.5,204.5) .. controls (319.01,204.5) and (317,202.49) .. (317,200) -- cycle ;
				\draw  [fill={rgb, 255:red, 255; green, 255; blue, 255 }  ,fill opacity=1 ] (296,173) .. controls (296,170.51) and (298.01,168.5) .. (300.5,168.5) .. controls (302.99,168.5) and (305,170.51) .. (305,173) .. controls (305,175.49) and (302.99,177.5) .. (300.5,177.5) .. controls (298.01,177.5) and (296,175.49) .. (296,173) -- cycle ;
				\draw    (362.72,283.43) -- (397.37,278.48) ;
				\draw    (358.48,317.37) -- (362.72,283.43) ;
				\draw    (358.48,317.37) -- (397.37,278.48) ;
				\draw  [fill={rgb, 255:red, 255; green, 255; blue, 255 }  ,fill opacity=1 ] (400.55,281.66) .. controls (398.79,283.42) and (395.95,283.42) .. (394.19,281.66) .. controls (392.43,279.9) and (392.43,277.05) .. (394.19,275.3) .. controls (395.95,273.54) and (398.79,273.54) .. (400.55,275.3) .. controls (402.31,277.05) and (402.31,279.9) .. (400.55,281.66) -- cycle ;
				\draw  [fill={rgb, 255:red, 255; green, 255; blue, 255 }  ,fill opacity=1 ] (365.9,286.61) .. controls (364.15,288.37) and (361.3,288.37) .. (359.54,286.61) .. controls (357.78,284.85) and (357.78,282) .. (359.54,280.25) .. controls (361.3,278.49) and (364.15,278.49) .. (365.9,280.25) .. controls (367.66,282) and (367.66,284.85) .. (365.9,286.61) -- cycle ;
				\draw  [fill={rgb, 255:red, 255; green, 255; blue, 255 }  ,fill opacity=1 ] (361.66,320.55) .. controls (359.9,322.31) and (357.05,322.31) .. (355.3,320.55) .. controls (353.54,318.79) and (353.54,315.95) .. (355.3,314.19) .. controls (357.05,312.43) and (359.9,312.43) .. (361.66,314.19) .. controls (363.42,315.95) and (363.42,318.79) .. (361.66,320.55) -- cycle ;
				
				\draw  [dash pattern={on 5.63pt off 4.5pt}][line width=1.5]  (248.5,201) .. controls (250.5,173) and (223.5,157) .. (204.5,162) .. controls (185.5,167) and (185.5,233) .. (205.5,240) .. controls (225.5,247) and (246.5,229) .. (248.5,201) -- cycle ;
				\draw  [dash pattern={on 5.63pt off 4.5pt}][line width=1.5]  (411.5,273) .. controls (401.5,256) and (370.5,255) .. (353.5,268) .. controls (336.5,281) and (336.5,319) .. (351.5,331) .. controls (366.5,343) and (421.5,290) .. (411.5,273) -- cycle ;
				\draw  [dash pattern={on 5.63pt off 4.5pt}][line width=1.5]  (407.5,115) .. controls (401.5,92) and (372.5,68) .. (353.5,73) .. controls (334.5,78) and (333.5,117) .. (350.5,131) .. controls (367.5,145) and (413.5,138) .. (407.5,115) -- cycle ;
				\draw  [dash pattern={on 5.63pt off 4.5pt}][line width=1.5]  (337.5,201) .. controls (338.5,177) and (312.5,157) .. (293.5,162) .. controls (274.5,167) and (274.5,233) .. (294.5,240) .. controls (314.5,247) and (336.5,225) .. (337.5,201) -- cycle ;
				
				\draw (158,187.4) node [anchor=north west][inner sep=0.75pt]    {$T_{u}^{1}$};
				\draw (398,68.4) node [anchor=north west][inner sep=0.75pt]    {$T_{u}^{2}$};
				\draw (400,318.4) node [anchor=north west][inner sep=0.75pt]    {$T_{u}^{3}$};
				\draw (347,190.4) node [anchor=north west][inner sep=0.75pt]    {$T_{u}^{123}$};

			\end{tikzpicture}
			
			\caption{Some of the edges in the construction of $G_1$.}
			\label{fig:G1_triangles}
		\end{figure}

		Finally, for every vertex $u \in G$, we define $V_u$ and $C_u$ as being both equal to the set containing the vertices in $T^0_u$, $e_{uv}$, $T^i_u$,  $T^{ij}_u$, $T^{123}_u$, $T^i_{uv}$ for some $i,j \in \{1,2,3\}$ and $v \in N(u)$.
		
		Since $G$ has maximum degree~4, the size of $C_u$ is bounded by a constant for every $u \in V(G)$, so the local reduction has expansion~$O(1)$. Moreover, a simple verification shows that every vertex of $G'$ has degree~at most~40 (and the vertices that have degree~40 are the ones in the triangles $T^0_u$ for some $u \in V(G)$ which has degree~4). The only property of the definition of local reduction that is non-trivial to verify is Property~\cref{l2}: namely, we will show that $G$ is 3-colorable if and only if $G'$ admits a $2$-edge-coloring without monochromatic triangles.
		
		So, first, assume that $G'$ admits a $2$-edge-coloring
                without monochromatic triangles, and let us show that
                $G$ is $3$-colorable. Let us say that the edges of~$G'$ are colored \emph{red} and \emph{blue}. In a
                $2$-edge-coloring of a triangle $T$ which is not
                monochromatic, we will say that the color assigned to two out of the three edges of $T$ is the \emph{main color} of~$T$. We start by proving the following lemma:

		\begin{lemma}
			\label{lem:coloring K5}
			Let us consider a clique $K_5$ on $5$ vertices, denoted by $s,t,u,v,w$. In any $2$-edge-coloring of $K_5$ without any monochromatic triangle, the color of the edge $st$ is the same as the main color of the triangle $uvw$. Conversely, any $2$-edge-coloring of the edge $st$ and of the triangle $uvw$ which makes $uvw$ non-monochromatic and which satisfies the previous condition can be completed into a $2$-edge-coloring of $K_5$ without any monochromatic triangle.
		\end{lemma}
		
		\begin{proof}
			Consider a $2$-edge-coloring of $K_5$ without
                        any monochromatic triangle, with colors red
                        and blue. By symmetry, assume that the main
                        color of $uvw$ is blue, and that $uv$ is blue,
                        $uw$ is blue, and $vw$ is red. Note that the
                        blue-degree of every vertex is at most two
                        (since otherwise there would be a
                        monochromatic triangle), and similarly the
                        red-degree of every vertex is at most two. As
                        $K_5$ is 4-regular, the blue and red subgraphs
                        are both 5-cycles. It remains to observe that
                        any 5-cycle of $K_5$ containing the edges $uv$ and
                        $uw$ also contains the edge $st$.
			See Figure~\ref{fig:coloring K5} for an illustration.
		\end{proof}
		
		\begin{figure}[h!]
			\centering
				\scalebox{0.8}{
					\begin{tikzpicture}[x=0.75pt,y=0.75pt,yscale=-1,xscale=1]
						
						\draw [color={rgb, 255:red, 208; green, 2; blue, 27 }  ,draw opacity=1 ][line width=3]    (122.5,172.5) -- (163.5,65.5) ;
						\draw [color={rgb, 255:red, 208; green, 2; blue, 27 }  ,draw opacity=1 ][line width=3]    (204.5,172.5) -- (163.5,65.5) ;
						\draw [color={rgb, 255:red, 208; green, 2; blue, 27 }  ,draw opacity=1 ][line width=3]    (103.5,109.5) -- (223.5,109.5) ;
						\draw    (103.5,109.5) -- (204.5,172.5) ;
						\draw    (122.5,172.5) -- (223.5,109.5) ;
						\draw [color={rgb, 255:red, 74; green, 144; blue, 226 }  ,draw opacity=1 ][line width=3]    (122.5,172.5) -- (204.5,172.5) ;
						\draw    (103.5,109.5) -- (122.5,172.5) ;
						\draw [color={rgb, 255:red, 74; green, 144; blue, 226 }  ,draw opacity=1 ][line width=3]    (103.5,109.5) -- (163.5,65.5) ;
						\draw [color={rgb, 255:red, 74; green, 144; blue, 226 }  ,draw opacity=1 ][line width=3]    (163.5,65.5) -- (223.5,109.5) ;
						\draw    (204.5,172.5) -- (223.5,109.5) ;
						\draw  [fill={rgb, 255:red, 255; green, 255; blue, 255 }  ,fill opacity=1 ] (98,109.5) .. controls (98,106.46) and (100.46,104) .. (103.5,104) .. controls (106.54,104) and (109,106.46) .. (109,109.5) .. controls (109,112.54) and (106.54,115) .. (103.5,115) .. controls (100.46,115) and (98,112.54) .. (98,109.5) -- cycle ;
						\draw  [fill={rgb, 255:red, 255; green, 255; blue, 255 }  ,fill opacity=1 ] (218,109.5) .. controls (218,106.46) and (220.46,104) .. (223.5,104) .. controls (226.54,104) and (229,106.46) .. (229,109.5) .. controls (229,112.54) and (226.54,115) .. (223.5,115) .. controls (220.46,115) and (218,112.54) .. (218,109.5) -- cycle ;
						\draw  [fill={rgb, 255:red, 255; green, 255; blue, 255 }  ,fill opacity=1 ] (117,172.5) .. controls (117,169.46) and (119.46,167) .. (122.5,167) .. controls (125.54,167) and (128,169.46) .. (128,172.5) .. controls (128,175.54) and (125.54,178) .. (122.5,178) .. controls (119.46,178) and (117,175.54) .. (117,172.5) -- cycle ;
						\draw  [fill={rgb, 255:red, 255; green, 255; blue, 255 }  ,fill opacity=1 ] (199,172.5) .. controls (199,169.46) and (201.46,167) .. (204.5,167) .. controls (207.54,167) and (210,169.46) .. (210,172.5) .. controls (210,175.54) and (207.54,178) .. (204.5,178) .. controls (201.46,178) and (199,175.54) .. (199,172.5) -- cycle ;
						\draw  [fill={rgb, 255:red, 255; green, 255; blue, 255 }  ,fill opacity=1 ] (158,65.5) .. controls (158,62.46) and (160.46,60) .. (163.5,60) .. controls (166.54,60) and (169,62.46) .. (169,65.5) .. controls (169,68.54) and (166.54,71) .. (163.5,71) .. controls (160.46,71) and (158,68.54) .. (158,65.5) -- cycle ;
						
						\draw (166,50) node [anchor=north west][inner sep=0.75pt]    {$u$};
						\draw (84,112.4) node [anchor=north west][inner sep=0.75pt]    {$v$};
						\draw (230,112.9) node [anchor=north west][inner sep=0.75pt]    {$w$};
						\draw (124.5,180) node [anchor=north west][inner sep=0.75pt]    {$s$};
						\draw (206.5,178) node [anchor=north west][inner sep=0.75pt]    {$t$};
						
						\draw [draw opacity=0]   (29.5,109.5) -- (103.5,109.5) ;
						\draw [draw opacity=0]   (223.5,109.5) -- (297.5,109.5) ;
						
					\end{tikzpicture}
					}
				
			\hspace{3cm}
			
				\scalebox{0.8}{
					\begin{tikzpicture}[x=0.75pt,y=0.75pt,yscale=-1,xscale=1]
						
						\draw [color={rgb, 255:red, 208; green, 2; blue, 27 }  ,draw opacity=1 ][line width=3]    (122.5,172.5) -- (163.5,65.5) ;
						\draw [color={rgb, 255:red, 208; green, 2; blue, 27 }  ,draw opacity=1 ][line width=3]    (204.5,172.5) -- (163.5,65.5) ;
						\draw [color={rgb, 255:red, 208; green, 2; blue, 27 }  ,draw opacity=1 ][line width=3]    (103.5,109.5) -- (223.5,109.5) ;
						\draw [color={rgb, 255:red, 208; green, 2; blue, 27 }  ,draw opacity=1 ][line width=3]    (103.5,109.5) -- (204.5,172.5) ;
						\draw [color={rgb, 255:red, 208; green, 2; blue, 27 }  ,draw opacity=1 ][line width=3]    (122.5,172.5) -- (223.5,109.5) ;
						\draw [color={rgb, 255:red, 74; green, 144; blue, 226 }  ,draw opacity=1 ][line width=3]    (122.5,172.5) -- (204.5,172.5) ;
						\draw [color={rgb, 255:red, 74; green, 144; blue, 226 }  ,draw opacity=1 ][line width=3]    (103.5,109.5) -- (122.5,172.5) ;
						\draw [color={rgb, 255:red, 74; green, 144; blue, 226 }  ,draw opacity=1 ][line width=3]    (103.5,109.5) -- (163.5,65.5) ;
						\draw [color={rgb, 255:red, 74; green, 144; blue, 226 }  ,draw opacity=1 ][line width=3]    (163.5,65.5) -- (223.5,109.5) ;
						\draw [color={rgb, 255:red, 74; green, 144; blue, 226 }  ,draw opacity=1 ][line width=3]    (204.5,172.5) -- (223.5,109.5) ;
						\draw  [fill={rgb, 255:red, 255; green, 255; blue, 255 }  ,fill opacity=1 ] (98,109.5) .. controls (98,106.46) and (100.46,104) .. (103.5,104) .. controls (106.54,104) and (109,106.46) .. (109,109.5) .. controls (109,112.54) and (106.54,115) .. (103.5,115) .. controls (100.46,115) and (98,112.54) .. (98,109.5) -- cycle ;
						\draw  [fill={rgb, 255:red, 255; green, 255; blue, 255 }  ,fill opacity=1 ] (218,109.5) .. controls (218,106.46) and (220.46,104) .. (223.5,104) .. controls (226.54,104) and (229,106.46) .. (229,109.5) .. controls (229,112.54) and (226.54,115) .. (223.5,115) .. controls (220.46,115) and (218,112.54) .. (218,109.5) -- cycle ;
						\draw  [fill={rgb, 255:red, 255; green, 255; blue, 255 }  ,fill opacity=1 ] (117,172.5) .. controls (117,169.46) and (119.46,167) .. (122.5,167) .. controls (125.54,167) and (128,169.46) .. (128,172.5) .. controls (128,175.54) and (125.54,178) .. (122.5,178) .. controls (119.46,178) and (117,175.54) .. (117,172.5) -- cycle ;
						\draw  [fill={rgb, 255:red, 255; green, 255; blue, 255 }  ,fill opacity=1 ] (199,172.5) .. controls (199,169.46) and (201.46,167) .. (204.5,167) .. controls (207.54,167) and (210,169.46) .. (210,172.5) .. controls (210,175.54) and (207.54,178) .. (204.5,178) .. controls (201.46,178) and (199,175.54) .. (199,172.5) -- cycle ;
						\draw  [fill={rgb, 255:red, 255; green, 255; blue, 255 }  ,fill opacity=1 ] (158,65.5) .. controls (158,62.46) and (160.46,60) .. (163.5,60) .. controls (166.54,60) and (169,62.46) .. (169,65.5) .. controls (169,68.54) and (166.54,71) .. (163.5,71) .. controls (160.46,71) and (158,68.54) .. (158,65.5) -- cycle ;
						
						\draw (166,50) node [anchor=north west][inner sep=0.75pt]    {$u$};
						\draw (84,112.4) node [anchor=north west][inner sep=0.75pt]    {$v$};
						\draw (230,112.9) node [anchor=north west][inner sep=0.75pt]    {$w$};
						\draw (124.5,180) node [anchor=north west][inner sep=0.75pt]    {$s$};
						\draw (206.5,178) node [anchor=north west][inner sep=0.75pt]    {$t$};
						
						\draw [draw opacity=0]   (29.5,109.5) -- (103.5,109.5) ;
						\draw [draw opacity=0]   (223.5,109.5) -- (297.5,109.5) ;
						
					\end{tikzpicture}
					
				}
			
			\caption{(top) The main color of $uvw$ and
                          the color of $st$ are the same; (bottom) A coloring of $K_5$ without monochromatic triangle completing the previous one.}
			\label{fig:coloring K5}
		\end{figure}
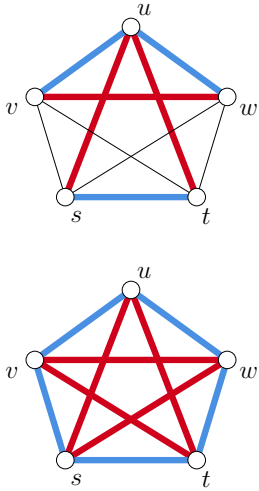

		By Lemma~\ref{lem:coloring K5}, in the 2-edge-coloring
                of $G'$ without monochromatic triangles, for every $uv
                \in E(G)$, since the union of $T^0_u$ and $e_{uv}$ induces a $K_5$, the main color of $T^0_u$ and the color of $e_{uv}$ are the same. Thus, the main colors of $T^0_u$ and $T^0_v$ are the same. Since $G$ is connected, the main color of $T^0_u$ is the same for all the vertices $u \in V(G)$. By symmetry, assume that this color is blue.
		Let us now prove the following claim to get the result.
		
		\begin{claim}
			\label{claim:main color}
			For every $u \in V(G)$, there is exactly one $i \in \{1,2,3\}$ such that the main color of $T^i_u$ is blue.
		\end{claim}
		
		\begin{proof}
			By Lemma~\ref{lem:coloring K5}, for every $1
                        \leqslant i < j \leqslant 3$, the colors of
                        the edges of $T^{ij}_u$ are respectively the
                        main colors of $T^i_u$,  $T^j_u$ and $T^0_u$
                        (which is blue). Since $T^{ij}_u$ is not
                        monochromatic,  the main color of either
                        $T^i_u$ or $T^j_u$ (or both) is red. Thus, at
                        least two of the main colors of $T^1_u$,
                        $T^2_u$, $T^3_u$ are red. But it cannot be red
                        for the 3 triangles, since othewise the triangle $T^{123}_u$ would be monochromatic (because its edges are colored with the main colors of $T^1_u$, $T^2_u$, $T^3_u$). So exactly one has blue as its main color.
		\end{proof}
		
		Let us define a proper 3-coloring of $G$ in the following way. By Claim~\ref{claim:main color}, for every $u \in V(G)$, there is a unique $i \in \{1,2,3\}$ such that the main color of $T^i_u$ is blue. Let $c(u)$ be this unique $i \in \{1,2,3\}$, and let us prove that $c$ is a proper 3-coloring of~$G$. By contradiction, assume that there exists $uv\in E(G)$ such that $c(u)=c(v)$. Then, the triangle $T^{c(u)}_{uv}$ is monochromatic (because its edges are colored with the main colors of $T^{c(u)}_u$, $T^{c(v)}_v$ and $T^0_u$, respectively, which are all blue), which is a contradiction. So $G$ is indeed 3-colorable.

		Conversely, assume that $G$ is 3-colorable. Then, we can define a $2$-edge-coloring of $G'$ without monochromatic triangle in the following way: for each $u \in V(G)$, we color two out of the three edges of $T^0_u$ by blue, and the other one by red. We also color two out of the three edges of $T^{c(u)}_u$ by blue, the other one by red, and we do the converse for $T^i_u$ if $i \neq c(u)$. Finally, using Lemma~\ref{lem:coloring K5}, we color the remaining edges to obtain a $2$-edge-coloring of $G'$ which does not have any monochromatic triangle.
	\end{proof}

        While every \textsf{coNP}-hard problem has a polynomial
        reduction from non-3-colorability, these reductions
        are often non-local, and sometimes reductions are much easier to perform from
        problems such as 3-SAT (which is not a purely graph theoretic problem). In the next section, we describe a
        mild modification of  Theorem~\ref{thm:main theorem local
          reduction} in the specific context of local reduction from
        3-SAT. Our main applications are non-Hamiltonicity and non-$\Delta$-edge-colorability.

	\section{Reductions from 3-SAT to $\p$}
	\label{sec:red3SAT}

	Let us recall the following reduction from 3-coloring to 3-SAT.
	Let $G$ be a graph. We construct a 3-CNF formula $\varphi_G$ which is
	satisfiable if and only if $G$ is 3-colorable.
	For every vertex $u \in V(G)$, we create three variables $u_1, u_2, u_3$, to
	express the fact that $u$ is colored with color 1, 2 or 3.
	
	For every $u \in V(G)$, we add the following clauses to $\varphi_G$:
	$$(u_1 \vee u_2 \vee u_3) \wedge (\neg u_1 \vee \neg u_2) \wedge (\neg u_1 \vee
	\neg u_3) \wedge (\neg u_2 \vee \neg u_3)$$
	
	For every edge $\{u,v\} \in E(G)$, we add the following clauses to $\varphi_G$:
	$$(\neg u_1 \vee \neg v_1) \wedge (\neg u_2 \vee \neg v_2) \wedge (\neg u_3 \vee
	\neg v_3)$$
	
	Note that $\varphi_G$ has $3|V(G)|$ variables and $4|V(G)| + 3|E(G)|$ clauses.
	
	\begin{claim}
		$G$ has a proper 3-coloring if and only if $\varphi_G$ is satisfiable.
	\end{claim}

	Let $\varphi$ be a 3-CNF formula, and $V(\varphi)$ be the set of variables of
	$\varphi$. 
	We say that two variables $x,y \in V(\varphi)$ \emph{have a clause in common} if
	there exists a clause of $\varphi$ which contains both $x$ or its negation, and
	$y$ or its negation (in particular, $x$ has a clause in common with itself).
	We define $\F_0$ as the set of 3-CNF formulae $\varphi$ such that each variable
	appears in at most $|V(\varphi)|+1$ clauses, and $\F_1$ as the set of 3-CNF
	formulae $\varphi$ such that each variable appears in at most 7~clauses. Note
	that for any graph $G$, we have $\varphi_G \in \F_0$, and for every $G$ of
	maximum degree~at most~4, we have $\varphi_G \in \F_1$.
	
	\subsection*{Local reductions from 3-SAT}
		Let $\F$ be a class of 3-CNF formulae, $\C$ be a graph class, and $\p$ be a
		property on $\C$. Let $\alpha, \beta : \mathbb{N} \to \mathbb{N}$. We say that
		there exists a \emph{local reduction from 3-SAT (in $\F$) to $\p$} with
		local expansion $\alpha$ and global expansion $\beta$ if there exists a function $f_\p$ which takes as input a 3-CNF
		formula $\varphi \in \F$ having the following
                properties:
                
		\begin{enumerate}[(P1)]
			\item the variables of $\varphi$ have identifiers in $\{1, \ldots,
			n\}$, where $n:=|V(\varphi)|$,\label[pr]{p1}
			\item for any variables $x,y \in V(\varphi)$, there exists a sequence of
			variables $x_1, \ldots, x_k$ such that $x=x_1$,$y=x_k$ and $x_i,x_{i+1}$ have a
			clause in common for all $i < k$,\label[pr]{p2}
                      \end{enumerate}
                      
		and outputs a graph $f_\p(\varphi) \in \C$ having at most $\beta(n)$ vertices, with identifiers, and such that the
		following properties are satisfied:
		
		\begin{enumerate}[(S1)]
			\item the vertices of $f_\p(\varphi)$ have unique identifiers
			written on $O(\log n)$ bits,\label[sat]{s1}
			
			\item $\varphi$ is satisfiable if and only if $f_\p(\varphi)$ satisfies $\p$, \label[sat]{s2}
			
			\item for each variable $x$ of $\varphi$, there exist two sets $C_x, V_x
			\subseteq V(f_\p(\varphi))$, such that, for every $x \in V(\varphi)$, the
			following properties are satisfied: \label[sat]{s3}
			\begin{enumerate}
				\item $V_x \subseteq C_x$, \label[sat]{s3a}
				\item $\bigcup_{y \in V(\varphi)} V_y = V(f_\p(\varphi))$, \label[sat]{s3b}
				\item $|C_x| \leqslant \alpha(n)$, \label[sat]{s3c}
				\item for every $t \in V(f_\p(\varphi))$, and for every variables $y$, $y'$
				such that $t \in C_y \cap C_{y'}$, there exists a sequence of variables $y_1 ,
				\ldots, y_k$ such that: $y_1 = y$, $y_k = y'$, $t \in \bigcap_{1 \leqslant i
					\leqslant k} C_{y_i}$, and for every $i \in \{1, \ldots, k-1\}$, $y_i$ and
				$y_{i+1}$ have a clause in common, \label[sat]{s3d}
				\item for every $t \in V_x$, $N[t]$ is included in $\bigcup_{y \in X} C_y$,
				where $X$ denotes the set of variables of $\varphi$ which have a clause in
				common with $x$, \label[sat]{s3e}
				\item $V_x$ and its neighborhood (resp.\ $C_x$) only depends on the set of
				clauses in which~$x$ (or its negation) appears. In other words: if $\varphi'$ is
				another formula on the same set of variables with identifiers $\{1, \ldots,
				n\}$, and if the set of clauses in which $x$ appears is the same in $\varphi$ and $\varphi'$, then the sets $V_x$ and the subgraphs with identifiers formed by
				by the vertices in $V_x$ and their adjacent edges are the same in $f_\p(\varphi)$ and $f_\p(\varphi')$ (resp.\ the sets $C_x$ in
				$f_\p(\varphi)$ and $f_\p(\varphi')$ are the same). \label[sat]{s3f}
			\end{enumerate}
		\end{enumerate}

	\begin{remark}
		We can make the same observations as in Remark~\ref{rem:def local reduction}. Namely: every local reduction from 3-SAT to $\p$ with local expansion $\alpha$ has global expansion at most $n \alpha(n)$, and the definition is symmetric between properties and their complement.
              \end{remark}

              We obtain the following variant of Theorem~\ref{thm:main theorem local
          reduction}, tailored to local reductions from 3-SAT.
	
	\begin{theorem}
		\label{thm:local reduction 3SAT}
		Let $\C, \C'$ be two graph classes, and $\F$ be a class of 3-CNF formulae. Let
		$\p$ be a graph property on $\C$. If there exists a local reduction from 3-SAT
		(in $\F$) to $\p$ with local expansion $\alpha$ and global expansion $\beta$, and if $\varphi_G \in \F$ for all $G
		\in \C'$, then there exists a local reduction from 3-colorability in $\C'$
		to $\p$ with local expansion~$\alpha'$ and global expansion~$\beta'$, where $\alpha'(n)=3\alpha(3n)$ and $\beta'(n)=\beta(3n)$.
	\end{theorem}
	
Again, the following immediate  corollary of Theorem~\ref{thm:local
          reduction 3SAT}  and Corollary~\ref{cor:reduction3col} will be the result we use in practice
to obtain lower bounds.

	\begin{corollary}
		\label{cor:reduction3SAT}
		Let $\p$ be a graph property (on some graph class~$\C$).
		\begin{enumerate}
			\item If there exists a local reduction from 3-SAT in $\F_0$ to $\p$ with
			local expansion $O(n^\delta)$ for some $0 \leqslant \delta < 2$, and global expansion $O(n^\gamma)$, then $\np$ requires certificates of size
			$\Omega\big(n^{(2-\delta)/\gamma}/\log n\big)$. In particular, since $\gamma \leqslant \delta+1$, $\np$ requires certificates of size $\Omega\big(n^{\frac{3}{\delta+1}-1}/\log n\big)$.
			
			\item If there exists a local reduction from 3-SAT in $\F_1$ to $\p$ with
			local expansion $O(n^\delta)$ for some $0 \leqslant \delta < 1$, and global expansion $O(n^\gamma)$, then $\np$ requires certificates of size
			$\Omega\big(n^{(1-\delta)/\gamma}/\log n\big)$. In particular, since $\gamma \leqslant \delta+1$, it requires certificates of size $\Omega\big(n^{\frac{2}{\delta+1}-1}/\log n\big)$.
		\end{enumerate}
	\end{corollary}

	
	We now proceed with the proof of Theorem~\ref{thm:local
          reduction 3SAT}, which is very similar to that of Theorem~\ref{thm:main theorem local
          reduction}.
	
	\begin{proof}[Proof of Theorem~\ref{thm:local reduction 3SAT}.]
		Let $f_{\p}$ be the function corresponding to the local reduction from 3-SAT
		(in $\F$) to $\p$ with local expansion~$\alpha$ and global expansion~$\beta$. We show that there exists a local
		reduction from 3-colorability in $\C'$ to $\p$ with local expansion $\alpha'(n):=3\alpha(3n)$ and global expansion $\beta'(n):=\beta(3n)$,
		with the function $f_{\text{3-col}, \p}(G):=f_\p(\varphi_G)$ for every $G$, where for
		each vertex $u \in V(G)$ having its identifier in $\{1, \ldots, n\}$, the
		variable $u_i$ in $\varphi_G$ received the identifier $3(\mathrm{id}(u)-1)+i$, so that the
		variables of $\varphi_G$ have indeed identifiers in $\{1, \ldots, 3n\}$ and thus $\varphi_G$ satisfies property~\cref{p1} of the definition of local reduction from 3-SAT (see the beginning of Section~\ref{sec:red3SAT}). Note that $\varphi_G$ also satisfies property~\cref{p2} because $G$ is connected.
		
		Now, let us prove that $f_{\text{3-col}, \p}(G)$ satisfies all the conditions of
		the definition of local reduction from 3-colorability
                to $\p$ (see the beginning of
                Section~\ref{sec:reduc}). First, it has at most
                $\beta(3n)=\beta'(n)$ vertices. Then, it is
                straightforward to see that conditions~\cref{l1}
                and~\cref{l2}, are satisfied. Let us prove
                that~\cref{l3} holds as well. For every vertex~$u$, let us define
		$C_u:=C_{u_1} \cup C_{u_2} \cup C_{u_3}$, and $V_u:=V_{u_1} \cup V_{u_2} \cup
		V_{u_3}$, where $u_1, u_2, u_3$ are the variables of $\varphi_G$ corresponding
		to $u$.
		It is again straightforward that~\cref{l3a} and~\cref{l3b} both hold.
		For~\cref{l3c}, observe that $\varphi_G$ has $3n$ variables, and that for every $u
		\in G$ and every $i \in \{1,2,3\}$, $|C_{u_i}| \leqslant \alpha(3n)$. Thus,
		$|C_u| \leqslant 3\alpha(3n) = \alpha'(n)$.
		For~\cref{l3d}, let $t \in V(f_{\text{3-col}, \p}(G))$ and let $u, v \in V(G)$ such that
		$t \in C_u \cap C_v$. Without loss of generality, assume that $t \in C_{u_1}
		\cap C_{v_1}$. By condition~\cref{s3d} of the definition of local reduction from 3-SAT, there exists a sequence of variables $y_1, \ldots, y_k$ such that $y_1=u_1$,
		$y_k=v_1$, for every $i \in \{1, \ldots, k-1\}$, $y_i$ and $y_{i+1}$ have a
		clause in common, and $t \in \bigcap_{1 \leqslant i \leqslant k} C_{y_i}$. Thus,
		by construction of $\varphi_G$, for every $i \in \{1, \ldots, k-1\}$, one of the
		two following cases holds:
		\begin{itemize}
			\item there exists $w \in V(G)$ and $j,j'\in \{1,2,3\}$ such that $y_i = w_j$
			and $y_{i+1}=w_{j'}$, or
			\item there exists $w,w' \in V(G)$ which are neighbors, and $j\in \{1,2,3\}$
			such that $y_i=w_j$ and $y_{i+1}=w'_j$
		\end{itemize}
		So there is a path from $u$ to $v$ which is included in $\{w \in V(G) \; | \; t
		\in C_w\}$, which proves~\cref{l3d}. Then, \cref{l3e}~simply follows from the definition
		of $C_u$ and from~\cref{s3e} of the definition of local reduction from 3-SAT. Finally,
		for~\cref{l3f}, if $G'$ is another $n$-vertex graph with unique identifiers in $\{1,
		\ldots, n\}$ and if the subgraph formed by $u$ and its neighbors is the same in $G$ and $G'$,
		then for every $i \in \{1,2,3\}$ the clauses in which the variable $u_i$ appears
		in $\varphi_G$ and $\varphi_{G'}$ are the same, and we can just apply the
		property~\cref{s3f} of the definition of local reduction from 3-SAT to conclude.
	\end{proof}

	\subsection{Non-Hamiltonicity}

	\begin{theorem}
		\label{thm:hamiltonian}
		The property of not having a Hamiltonian cycle has
                local complexity  $\Omega(\sqrt{n}/\log n)$, even in graphs of
		maximum degree~at most~4.
	\end{theorem}
	
	\begin{proof}
		Let $\p:=$ \ham{}. We prove that there exists a local reduction from 3-SAT (in
		$\F_0$) to $\p$ with local expansion $O(n)$, and then the result will follow from
		Corollary~\ref{cor:reduction3SAT}. Recall that $\F_0$
                was defined as the set of 3-CNF formulae $\varphi$ such that each variable
	appears in at most $|V(\varphi)|+1$ clauses.
		We use the reduction of~\cite{Sipser}. Let $\varphi$ be a 3-CNF formula with
		variables having identifiers in $\{1, \ldots, |V(\varphi)|\}$, and such that
		each variable appears in at most $|V(\varphi)|+1$ clauses.
		Let $n:=|V(\varphi)|$, and let $m$ be the number of clauses in $\varphi$. Let
		$x_1, \ldots, x_n$ be the variables of $\varphi$, and $C_1, \ldots, C_m$ be its
		clauses.
		
		The graph $f_\p(\varphi)$ is constructed in two steps. First, we construct a
		directed graph $f_\p^\ast(\varphi)$ which has a directed Hamiltonian cycle if
		and only if $\varphi$ if satisfiable. Then, we explain how to construct the
		undirected graph $f_\p(\varphi)$ from $f_\p^\ast(\varphi)$.
		
		The graph $f_\p^\ast(\varphi)$ is constructed as follows. For each variable
		$x_i$ such that $x_i$ or its negation appears in $d$ clauses,
		we create a horizontal row of $3d+3$ vertices, an entry node, and an exit node. The exit
		node of $x_i$ is the entry one of $x_{i+1}$, except the entry node of $x_1$
		which is called $s$, and the exit node of $x_n$ which is called $t$. There is
		also an arc from~$t$ to~$s$.
		Finally, we add $m$ single vertices, one for each clause. This is depicted on
		Figure~\ref{fig:ham1}. We now specify how to connect the vertices corresponding
		to the clauses to the rest of the graph.
		For each $i \in \{1, \ldots, n\}$, the $3d+3$ vertices of the row corresponding
		to $x_i$ are divided as follows. Among the $3d+1$ vertices which are not the
		first and the last one, they are grouped by adjacent pairs, with one separator
		between each pair.
		Let us denote $C_{i_1}, \ldots, C_{i_d}$ the clauses in which $x_i$ or its negation appears. For every $j \in \{1, \ldots, d\}$, we add two edges between the $j$-th
		pair of the row of $x_i$ and the vertex corresponding to $C_{i_j}$, and their
		direction depends only on whether $x_i$ appears positively or negatively in
		$C_j$. See~\cite{Sipser} for more details and Figure~\ref{fig:ham2} for an
		example.

		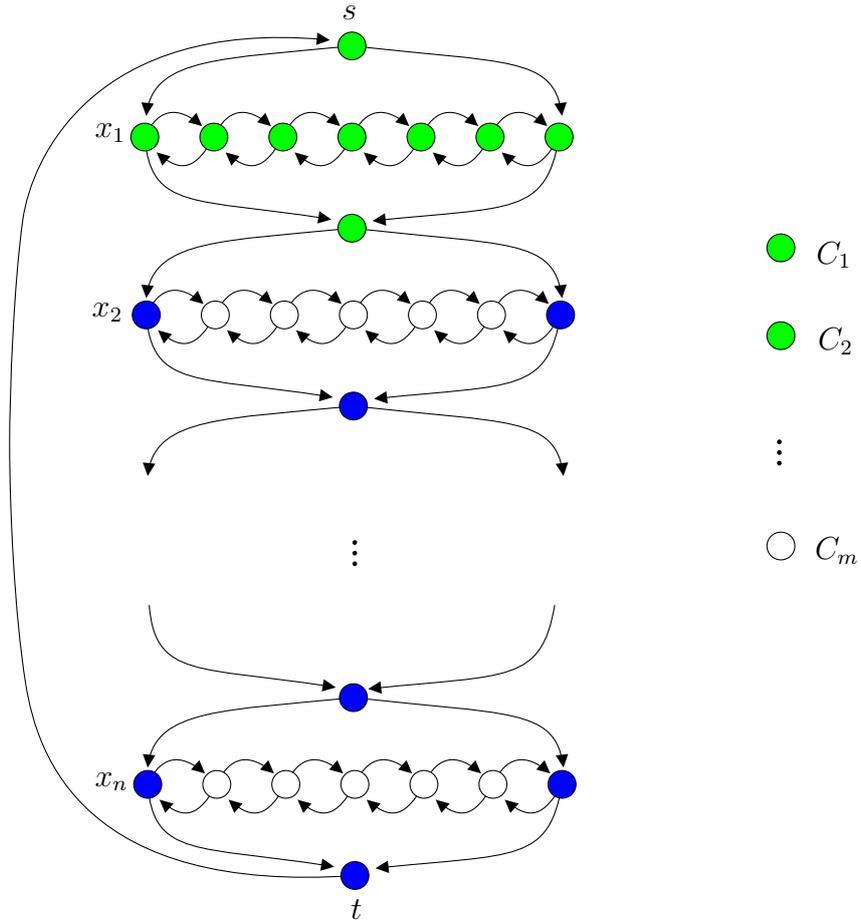
\begin{figure}
			\centering
			
			\begin{tikzpicture}[x=0.55pt,y=0.55pt,yscale=-1,xscale=1]
				
				\draw    (248.76,35.55) .. controls (93.62,21.91) and (50.44,116.43) ..
				(44.5,159) .. controls (38.5,202) and (35.5,274) .. (35.5,314) .. controls
				(35.5,354) and (38.5,428) .. (46.5,478) .. controls (54.5,528) and (86.5,622) ..
				(270.5,609) ;
				\draw [shift={(253.5,36)}, rotate = 185.75] [fill={rgb, 255:red, 0; green, 0;
					blue, 0 }  ][line width=0.08]  [draw opacity=0] (8.93,-4.29) -- (0,0) --
				(8.93,4.29) -- cycle    ;
				\draw    (129.97,331.51) .. controls (136.24,295.05) and (181.23,295.73) ..
				(269.5,287) ;
				\draw [shift={(129.5,335)}, rotate = 275.71] [fill={rgb, 255:red, 0; green,
					0; blue, 0 }  ][line width=0.08]  [draw opacity=0] (8.93,-4.29) -- (0,0) --
				(8.93,4.29) -- cycle    ;
				\draw    (412.08,332.01) .. controls (405.95,294.24) and (374.8,298.7) ..
				(269.5,287) ;
				\draw [shift={(412.5,335)}, rotate = 263.05] [fill={rgb, 255:red, 0; green,
					0; blue, 0 }  ][line width=0.08]  [draw opacity=0] (8.93,-4.29) -- (0,0) --
				(8.93,4.29) -- cycle    ;
				\draw    (287.68,281.56) .. controls (378.22,272.2) and (401.64,278.18) ..
				(410.5,224.5) ;
				\draw [shift={(283.5,282)}, rotate = 353.93] [fill={rgb, 255:red, 0; green,
					0; blue, 0 }  ][line width=0.08]  [draw opacity=0] (8.93,-4.29) -- (0,0) --
				(8.93,4.29) -- cycle    ;
				\draw    (251.4,280.34) .. controls (162.58,266.29) and (132.44,275.23) ..
				(128.5,224.5) ;
				\draw [shift={(255.5,281)}, rotate = 189.26] [fill={rgb, 255:red, 0; green,
					0; blue, 0 }  ][line width=0.08]  [draw opacity=0] (8.93,-4.29) -- (0,0) --
				(8.93,4.29) -- cycle    ;
				\draw    (128.97,209.51) .. controls (135.24,173.05) and (180.23,173.73) ..
				(268.5,165) ;
				\draw [shift={(128.5,213)}, rotate = 275.71] [fill={rgb, 255:red, 0; green,
					0; blue, 0 }  ][line width=0.08]  [draw opacity=0] (8.93,-4.29) -- (0,0) --
				(8.93,4.29) -- cycle    ;
				\draw    (411.08,210.01) .. controls (404.95,172.24) and (373.8,176.7) ..
				(268.5,165) ;
				\draw [shift={(411.5,213)}, rotate = 263.05] [fill={rgb, 255:red, 0; green,
					0; blue, 0 }  ][line width=0.08]  [draw opacity=0] (8.93,-4.29) -- (0,0) --
				(8.93,4.29) -- cycle    ;
				\draw    (175.5,224.5) .. controls (163.08,246.94) and (153.41,248.86) ..
				(138.62,236.8) ;
				\draw [shift={(136.5,235)}, rotate = 41.19] [fill={rgb, 255:red, 0; green, 0;
					blue, 0 }  ][line width=0.08]  [draw opacity=0] (8.93,-4.29) -- (0,0) --
				(8.93,4.29) -- cycle    ;
				\draw    (128.5,224.5) .. controls (135.26,209.54) and (146.67,200.17) ..
				(165.43,215.26) ;
				\draw [shift={(167.5,217)}, rotate = 221.19] [fill={rgb, 255:red, 0; green,
					0; blue, 0 }  ][line width=0.08]  [draw opacity=0] (8.93,-4.29) -- (0,0) --
				(8.93,4.29) -- cycle    ;
				\draw    (175.5,224.5) .. controls (182.26,209.54) and (193.67,200.17) ..
				(212.43,215.26) ;
				\draw [shift={(214.5,217)}, rotate = 221.19] [fill={rgb, 255:red, 0; green,
					0; blue, 0 }  ][line width=0.08]  [draw opacity=0] (8.93,-4.29) -- (0,0) --
				(8.93,4.29) -- cycle    ;
				\draw    (222.5,224.5) .. controls (229.26,209.54) and (240.67,200.17) ..
				(259.43,215.26) ;
				\draw [shift={(261.5,217)}, rotate = 221.19] [fill={rgb, 255:red, 0; green,
					0; blue, 0 }  ][line width=0.08]  [draw opacity=0] (8.93,-4.29) -- (0,0) --
				(8.93,4.29) -- cycle    ;
				\draw    (269.5,224.5) .. controls (276.26,209.54) and (287.67,200.17) ..
				(306.43,215.26) ;
				\draw [shift={(308.5,217)}, rotate = 221.19] [fill={rgb, 255:red, 0; green,
					0; blue, 0 }  ][line width=0.08]  [draw opacity=0] (8.93,-4.29) -- (0,0) --
				(8.93,4.29) -- cycle    ;
				\draw    (316.5,224.5) .. controls (323.26,209.54) and (334.67,200.17) ..
				(353.43,215.26) ;
				\draw [shift={(355.5,217)}, rotate = 221.19] [fill={rgb, 255:red, 0; green,
					0; blue, 0 }  ][line width=0.08]  [draw opacity=0] (8.93,-4.29) -- (0,0) --
				(8.93,4.29) -- cycle    ;
				\draw    (363.5,224.5) .. controls (370.26,209.54) and (381.67,200.17) ..
				(400.43,215.26) ;
				\draw [shift={(402.5,217)}, rotate = 221.19] [fill={rgb, 255:red, 0; green,
					0; blue, 0 }  ][line width=0.08]  [draw opacity=0] (8.93,-4.29) -- (0,0) --
				(8.93,4.29) -- cycle    ;
				\draw    (222.5,224.5) .. controls (210.08,246.94) and (200.41,248.86) ..
				(185.62,236.8) ;
				\draw [shift={(183.5,235)}, rotate = 41.19] [fill={rgb, 255:red, 0; green, 0;
					blue, 0 }  ][line width=0.08]  [draw opacity=0] (8.93,-4.29) -- (0,0) --
				(8.93,4.29) -- cycle    ;
				\draw    (269.5,224.5) .. controls (257.08,246.94) and (247.41,248.86) ..
				(232.62,236.8) ;
				\draw [shift={(230.5,235)}, rotate = 41.19] [fill={rgb, 255:red, 0; green, 0;
					blue, 0 }  ][line width=0.08]  [draw opacity=0] (8.93,-4.29) -- (0,0) --
				(8.93,4.29) -- cycle    ;
				\draw    (316.5,224.5) .. controls (304.08,246.94) and (294.41,248.86) ..
				(279.62,236.8) ;
				\draw [shift={(277.5,235)}, rotate = 41.19] [fill={rgb, 255:red, 0; green, 0;
					blue, 0 }  ][line width=0.08]  [draw opacity=0] (8.93,-4.29) -- (0,0) --
				(8.93,4.29) -- cycle    ;
				\draw    (363.5,224.5) .. controls (351.09,246.94) and (341.41,248.86) ..
				(326.62,236.8) ;
				\draw [shift={(324.5,235)}, rotate = 41.19] [fill={rgb, 255:red, 0; green, 0;
					blue, 0 }  ][line width=0.08]  [draw opacity=0] (8.93,-4.29) -- (0,0) --
				(8.93,4.29) -- cycle    ;
				\draw    (410.5,224.5) .. controls (398.09,246.94) and (388.41,248.86) ..
				(373.62,236.8) ;
				\draw [shift={(371.5,235)}, rotate = 41.19] [fill={rgb, 255:red, 0; green, 0;
					blue, 0 }  ][line width=0.08]  [draw opacity=0] (8.93,-4.29) -- (0,0) --
				(8.93,4.29) -- cycle    ;
				\draw  [fill={rgb, 255:red, 0; green, 0; blue, 255 }  ,fill opacity=1 ]
				(119,224.5) .. controls (119,219.25) and (123.25,215) .. (128.5,215) .. controls
				(133.75,215) and (138,219.25) .. (138,224.5) .. controls (138,229.75) and
				(133.75,234) .. (128.5,234) .. controls (123.25,234) and (119,229.75) ..
				(119,224.5) -- cycle ;
				\draw  [fill={rgb, 255:red, 255; green, 255; blue, 255 }  ,fill opacity=1 ]
				(166,224.5) .. controls (166,219.25) and (170.25,215) .. (175.5,215) .. controls
				(180.75,215) and (185,219.25) .. (185,224.5) .. controls (185,229.75) and
				(180.75,234) .. (175.5,234) .. controls (170.25,234) and (166,229.75) ..
				(166,224.5) -- cycle ;
				\draw  [fill={rgb, 255:red, 255; green, 255; blue, 255 }  ,fill opacity=1 ]
				(213,224.5) .. controls (213,219.25) and (217.25,215) .. (222.5,215) .. controls
				(227.75,215) and (232,219.25) .. (232,224.5) .. controls (232,229.75) and
				(227.75,234) .. (222.5,234) .. controls (217.25,234) and (213,229.75) ..
				(213,224.5) -- cycle ;
				\draw  [fill={rgb, 255:red, 255; green, 255; blue, 255 }  ,fill opacity=1 ]
				(260,224.5) .. controls (260,219.25) and (264.25,215) .. (269.5,215) .. controls
				(274.75,215) and (279,219.25) .. (279,224.5) .. controls (279,229.75) and
				(274.75,234) .. (269.5,234) .. controls (264.25,234) and (260,229.75) ..
				(260,224.5) -- cycle ;
				\draw  [fill={rgb, 255:red, 255; green, 255; blue, 255 }  ,fill opacity=1 ]
				(307,224.5) .. controls (307,219.25) and (311.25,215) .. (316.5,215) .. controls
				(321.75,215) and (326,219.25) .. (326,224.5) .. controls (326,229.75) and
				(321.75,234) .. (316.5,234) .. controls (311.25,234) and (307,229.75) ..
				(307,224.5) -- cycle ;
				\draw  [fill={rgb, 255:red, 255; green, 255; blue, 255 }  ,fill opacity=1 ]
				(354,224.5) .. controls (354,219.25) and (358.25,215) .. (363.5,215) .. controls
				(368.75,215) and (373,219.25) .. (373,224.5) .. controls (373,229.75) and
				(368.75,234) .. (363.5,234) .. controls (358.25,234) and (354,229.75) ..
				(354,224.5) -- cycle ;
				\draw  [fill={rgb, 255:red, 0; green, 0; blue, 255 }  ,fill opacity=1 ]
				(401,224.5) .. controls (401,219.25) and (405.25,215) .. (410.5,215) .. controls
				(415.75,215) and (420,219.25) .. (420,224.5) .. controls (420,229.75) and
				(415.75,234) .. (410.5,234) .. controls (405.25,234) and (401,229.75) ..
				(401,224.5) -- cycle ;
				\draw  [fill={rgb, 255:red, 0; green, 0; blue, 255 }  ,fill opacity=1 ]
				(260,287) .. controls (260,281.75) and (264.25,277.5) .. (269.5,277.5) ..
				controls (274.75,277.5) and (279,281.75) .. (279,287) .. controls (279,292.25)
				and (274.75,296.5) .. (269.5,296.5) .. controls (264.25,296.5) and (260,292.25)
				.. (260,287) -- cycle ;
				\draw    (286.68,159.56) .. controls (377.22,150.2) and (400.64,156.18) ..
				(409.5,102.5) ;
				\draw [shift={(282.5,160)}, rotate = 353.93] [fill={rgb, 255:red, 0; green,
					0; blue, 0 }  ][line width=0.08]  [draw opacity=0] (8.93,-4.29) -- (0,0) --
				(8.93,4.29) -- cycle    ;
				\draw    (128.97,84.51) .. controls (135.24,48.05) and (180.23,48.73) ..
				(268.5,40) ;
				\draw [shift={(128.5,88)}, rotate = 275.71] [fill={rgb, 255:red, 0; green, 0;
					blue, 0 }  ][line width=0.08]  [draw opacity=0] (8.93,-4.29) -- (0,0) --
				(8.93,4.29) -- cycle    ;
				\draw    (411.08,85.01) .. controls (404.95,47.24) and (373.8,51.7) ..
				(268.5,40) ;
				\draw [shift={(411.5,88)}, rotate = 263.05] [fill={rgb, 255:red, 0; green, 0;
					blue, 0 }  ][line width=0.08]  [draw opacity=0] (8.93,-4.29) -- (0,0) --
				(8.93,4.29) -- cycle    ;
				\draw    (250.4,158.34) .. controls (161.58,144.29) and (131.44,153.23) ..
				(127.5,102.5) ;
				\draw [shift={(254.5,159)}, rotate = 189.26] [fill={rgb, 255:red, 0; green,
					0; blue, 0 }  ][line width=0.08]  [draw opacity=0] (8.93,-4.29) -- (0,0) --
				(8.93,4.29) -- cycle    ;
				\draw    (174.5,102.5) .. controls (162.08,124.94) and (152.41,126.86) ..
				(137.62,114.8) ;
				\draw [shift={(135.5,113)}, rotate = 41.19] [fill={rgb, 255:red, 0; green, 0;
					blue, 0 }  ][line width=0.08]  [draw opacity=0] (8.93,-4.29) -- (0,0) --
				(8.93,4.29) -- cycle    ;
				\draw    (127.5,102.5) .. controls (134.26,87.54) and (145.67,78.17) ..
				(164.43,93.26) ;
				\draw [shift={(166.5,95)}, rotate = 221.19] [fill={rgb, 255:red, 0; green, 0;
					blue, 0 }  ][line width=0.08]  [draw opacity=0] (8.93,-4.29) -- (0,0) --
				(8.93,4.29) -- cycle    ;
				\draw    (174.5,102.5) .. controls (181.26,87.54) and (192.67,78.17) ..
				(211.43,93.26) ;
				\draw [shift={(213.5,95)}, rotate = 221.19] [fill={rgb, 255:red, 0; green, 0;
					blue, 0 }  ][line width=0.08]  [draw opacity=0] (8.93,-4.29) -- (0,0) --
				(8.93,4.29) -- cycle    ;
				\draw    (221.5,102.5) .. controls (228.26,87.54) and (239.67,78.17) ..
				(258.43,93.26) ;
				\draw [shift={(260.5,95)}, rotate = 221.19] [fill={rgb, 255:red, 0; green, 0;
					blue, 0 }  ][line width=0.08]  [draw opacity=0] (8.93,-4.29) -- (0,0) --
				(8.93,4.29) -- cycle    ;
				\draw    (268.5,102.5) .. controls (275.26,87.54) and (286.67,78.17) ..
				(305.43,93.26) ;
				\draw [shift={(307.5,95)}, rotate = 221.19] [fill={rgb, 255:red, 0; green, 0;
					blue, 0 }  ][line width=0.08]  [draw opacity=0] (8.93,-4.29) -- (0,0) --
				(8.93,4.29) -- cycle    ;
				\draw    (315.5,102.5) .. controls (322.26,87.54) and (333.67,78.17) ..
				(352.43,93.26) ;
				\draw [shift={(354.5,95)}, rotate = 221.19] [fill={rgb, 255:red, 0; green, 0;
					blue, 0 }  ][line width=0.08]  [draw opacity=0] (8.93,-4.29) -- (0,0) --
				(8.93,4.29) -- cycle    ;
				\draw    (362.5,102.5) .. controls (369.26,87.54) and (380.67,78.17) ..
				(399.43,93.26) ;
				\draw [shift={(401.5,95)}, rotate = 221.19] [fill={rgb, 255:red, 0; green, 0;
					blue, 0 }  ][line width=0.08]  [draw opacity=0] (8.93,-4.29) -- (0,0) --
				(8.93,4.29) -- cycle    ;
				\draw    (221.5,102.5) .. controls (209.08,124.94) and (199.41,126.86) ..
				(184.62,114.8) ;
				\draw [shift={(182.5,113)}, rotate = 41.19] [fill={rgb, 255:red, 0; green, 0;
					blue, 0 }  ][line width=0.08]  [draw opacity=0] (8.93,-4.29) -- (0,0) --
				(8.93,4.29) -- cycle    ;
				\draw    (268.5,102.5) .. controls (256.08,124.94) and (246.41,126.86) ..
				(231.62,114.8) ;
				\draw [shift={(229.5,113)}, rotate = 41.19] [fill={rgb, 255:red, 0; green, 0;
					blue, 0 }  ][line width=0.08]  [draw opacity=0] (8.93,-4.29) -- (0,0) --
				(8.93,4.29) -- cycle    ;
				\draw    (315.5,102.5) .. controls (303.08,124.94) and (293.41,126.86) ..
				(278.62,114.8) ;
				\draw [shift={(276.5,113)}, rotate = 41.19] [fill={rgb, 255:red, 0; green, 0;
					blue, 0 }  ][line width=0.08]  [draw opacity=0] (8.93,-4.29) -- (0,0) --
				(8.93,4.29) -- cycle    ;
				\draw    (362.5,102.5) .. controls (350.09,124.94) and (340.41,126.86) ..
				(325.62,114.8) ;
				\draw [shift={(323.5,113)}, rotate = 41.19] [fill={rgb, 255:red, 0; green, 0;
					blue, 0 }  ][line width=0.08]  [draw opacity=0] (8.93,-4.29) -- (0,0) --
				(8.93,4.29) -- cycle    ;
				\draw    (409.5,102.5) .. controls (397.09,124.94) and (387.41,126.86) ..
				(372.62,114.8) ;
				\draw [shift={(370.5,113)}, rotate = 41.19] [fill={rgb, 255:red, 0; green, 0;
					blue, 0 }  ][line width=0.08]  [draw opacity=0] (8.93,-4.29) -- (0,0) --
				(8.93,4.29) -- cycle    ;
				\draw  [fill={rgb, 255:red, 0; green, 255; blue, 0 }  ,fill opacity=1 ]
				(259,40) .. controls (259,34.75) and (263.25,30.5) .. (268.5,30.5) .. controls
				(273.75,30.5) and (278,34.75) .. (278,40) .. controls (278,45.25) and
				(273.75,49.5) .. (268.5,49.5) .. controls (263.25,49.5) and (259,45.25) ..
				(259,40) -- cycle ;
				\draw  [fill={rgb, 255:red, 0; green, 255; blue, 0 }  ,fill opacity=1 ]
				(118,102.5) .. controls (118,97.25) and (122.25,93) .. (127.5,93) .. controls
				(132.75,93) and (137,97.25) .. (137,102.5) .. controls (137,107.75) and
				(132.75,112) .. (127.5,112) .. controls (122.25,112) and (118,107.75) ..
				(118,102.5) -- cycle ;
				\draw  [fill={rgb, 255:red, 0; green, 255; blue, 0 }  ,fill opacity=1 ]
				(165,102.5) .. controls (165,97.25) and (169.25,93) .. (174.5,93) .. controls
				(179.75,93) and (184,97.25) .. (184,102.5) .. controls (184,107.75) and
				(179.75,112) .. (174.5,112) .. controls (169.25,112) and (165,107.75) ..
				(165,102.5) -- cycle ;
				\draw  [fill={rgb, 255:red, 0; green, 255; blue, 0 }  ,fill opacity=1 ]
				(212,102.5) .. controls (212,97.25) and (216.25,93) .. (221.5,93) .. controls
				(226.75,93) and (231,97.25) .. (231,102.5) .. controls (231,107.75) and
				(226.75,112) .. (221.5,112) .. controls (216.25,112) and (212,107.75) ..
				(212,102.5) -- cycle ;
				\draw  [fill={rgb, 255:red, 0; green, 255; blue, 0 }  ,fill opacity=1 ]
				(259,102.5) .. controls (259,97.25) and (263.25,93) .. (268.5,93) .. controls
				(273.75,93) and (278,97.25) .. (278,102.5) .. controls (278,107.75) and
				(273.75,112) .. (268.5,112) .. controls (263.25,112) and (259,107.75) ..
				(259,102.5) -- cycle ;
				\draw  [fill={rgb, 255:red, 0; green, 255; blue, 0 }  ,fill opacity=1 ]
				(306,102.5) .. controls (306,97.25) and (310.25,93) .. (315.5,93) .. controls
				(320.75,93) and (325,97.25) .. (325,102.5) .. controls (325,107.75) and
				(320.75,112) .. (315.5,112) .. controls (310.25,112) and (306,107.75) ..
				(306,102.5) -- cycle ;
				\draw  [fill={rgb, 255:red, 0; green, 255; blue, 0 }  ,fill opacity=1 ]
				(353,102.5) .. controls (353,97.25) and (357.25,93) .. (362.5,93) .. controls
				(367.75,93) and (372,97.25) .. (372,102.5) .. controls (372,107.75) and
				(367.75,112) .. (362.5,112) .. controls (357.25,112) and (353,107.75) ..
				(353,102.5) -- cycle ;
				\draw  [fill={rgb, 255:red, 0; green, 255; blue, 0 }  ,fill opacity=1 ]
				(400,102.5) .. controls (400,97.25) and (404.25,93) .. (409.5,93) .. controls
				(414.75,93) and (419,97.25) .. (419,102.5) .. controls (419,107.75) and
				(414.75,112) .. (409.5,112) .. controls (404.25,112) and (400,107.75) ..
				(400,102.5) -- cycle ;
				\draw  [fill={rgb, 255:red, 0; green, 255; blue, 0 }  ,fill opacity=1 ]
				(259,165) .. controls (259,159.75) and (263.25,155.5) .. (268.5,155.5) ..
				controls (273.75,155.5) and (278,159.75) .. (278,165) .. controls (278,170.25)
				and (273.75,174.5) .. (268.5,174.5) .. controls (263.25,174.5) and (259,170.25)
				.. (259,165) -- cycle ;
				\draw    (288.68,603.56) .. controls (379.22,594.2) and (402.64,600.18) ..
				(411.5,546.5) ;
				\draw [shift={(284.5,604)}, rotate = 353.93] [fill={rgb, 255:red, 0; green,
					0; blue, 0 }  ][line width=0.08]  [draw opacity=0] (8.93,-4.29) -- (0,0) --
				(8.93,4.29) -- cycle    ;
				\draw    (252.4,602.34) .. controls (163.58,588.29) and (133.44,597.23) ..
				(129.5,546.5) ;
				\draw [shift={(256.5,603)}, rotate = 189.26] [fill={rgb, 255:red, 0; green,
					0; blue, 0 }  ][line width=0.08]  [draw opacity=0] (8.93,-4.29) -- (0,0) --
				(8.93,4.29) -- cycle    ;
				\draw    (129.97,531.51) .. controls (136.24,495.05) and (181.23,495.73) ..
				(269.5,487) ;
				\draw [shift={(129.5,535)}, rotate = 275.71] [fill={rgb, 255:red, 0; green,
					0; blue, 0 }  ][line width=0.08]  [draw opacity=0] (8.93,-4.29) -- (0,0) --
				(8.93,4.29) -- cycle    ;
				\draw    (412.08,532.01) .. controls (405.95,494.24) and (374.8,498.7) ..
				(269.5,487) ;
				\draw [shift={(412.5,535)}, rotate = 263.05] [fill={rgb, 255:red, 0; green,
					0; blue, 0 }  ][line width=0.08]  [draw opacity=0] (8.93,-4.29) -- (0,0) --
				(8.93,4.29) -- cycle    ;
				\draw    (176.5,546.5) .. controls (164.08,568.94) and (154.41,570.86) ..
				(139.62,558.8) ;
				\draw [shift={(137.5,557)}, rotate = 41.19] [fill={rgb, 255:red, 0; green, 0;
					blue, 0 }  ][line width=0.08]  [draw opacity=0] (8.93,-4.29) -- (0,0) --
				(8.93,4.29) -- cycle    ;
				\draw    (129.5,546.5) .. controls (136.26,531.54) and (147.67,522.17) ..
				(166.43,537.26) ;
				\draw [shift={(168.5,539)}, rotate = 221.19] [fill={rgb, 255:red, 0; green,
					0; blue, 0 }  ][line width=0.08]  [draw opacity=0] (8.93,-4.29) -- (0,0) --
				(8.93,4.29) -- cycle    ;
				\draw    (176.5,546.5) .. controls (183.26,531.54) and (194.67,522.17) ..
				(213.43,537.26) ;
				\draw [shift={(215.5,539)}, rotate = 221.19] [fill={rgb, 255:red, 0; green,
					0; blue, 0 }  ][line width=0.08]  [draw opacity=0] (8.93,-4.29) -- (0,0) --
				(8.93,4.29) -- cycle    ;
				\draw    (223.5,546.5) .. controls (230.26,531.54) and (241.67,522.17) ..
				(260.43,537.26) ;
				\draw [shift={(262.5,539)}, rotate = 221.19] [fill={rgb, 255:red, 0; green,
					0; blue, 0 }  ][line width=0.08]  [draw opacity=0] (8.93,-4.29) -- (0,0) --
				(8.93,4.29) -- cycle    ;
				\draw    (270.5,546.5) .. controls (277.26,531.54) and (288.67,522.17) ..
				(307.43,537.26) ;
				\draw [shift={(309.5,539)}, rotate = 221.19] [fill={rgb, 255:red, 0; green,
					0; blue, 0 }  ][line width=0.08]  [draw opacity=0] (8.93,-4.29) -- (0,0) --
				(8.93,4.29) -- cycle    ;
				\draw    (317.5,546.5) .. controls (324.26,531.54) and (335.67,522.17) ..
				(354.43,537.26) ;
				\draw [shift={(356.5,539)}, rotate = 221.19] [fill={rgb, 255:red, 0; green,
					0; blue, 0 }  ][line width=0.08]  [draw opacity=0] (8.93,-4.29) -- (0,0) --
				(8.93,4.29) -- cycle    ;
				\draw    (364.5,546.5) .. controls (371.26,531.54) and (382.67,522.17) ..
				(401.43,537.26) ;
				\draw [shift={(403.5,539)}, rotate = 221.19] [fill={rgb, 255:red, 0; green,
					0; blue, 0 }  ][line width=0.08]  [draw opacity=0] (8.93,-4.29) -- (0,0) --
				(8.93,4.29) -- cycle    ;
				\draw    (223.5,546.5) .. controls (211.08,568.94) and (201.41,570.86) ..
				(186.62,558.8) ;
				\draw [shift={(184.5,557)}, rotate = 41.19] [fill={rgb, 255:red, 0; green, 0;
					blue, 0 }  ][line width=0.08]  [draw opacity=0] (8.93,-4.29) -- (0,0) --
				(8.93,4.29) -- cycle    ;
				\draw    (270.5,546.5) .. controls (258.08,568.94) and (248.41,570.86) ..
				(233.62,558.8) ;
				\draw [shift={(231.5,557)}, rotate = 41.19] [fill={rgb, 255:red, 0; green, 0;
					blue, 0 }  ][line width=0.08]  [draw opacity=0] (8.93,-4.29) -- (0,0) --
				(8.93,4.29) -- cycle    ;
				\draw    (317.5,546.5) .. controls (305.08,568.94) and (295.41,570.86) ..
				(280.62,558.8) ;
				\draw [shift={(278.5,557)}, rotate = 41.19] [fill={rgb, 255:red, 0; green, 0;
					blue, 0 }  ][line width=0.08]  [draw opacity=0] (8.93,-4.29) -- (0,0) --
				(8.93,4.29) -- cycle    ;
				\draw    (364.5,546.5) .. controls (352.09,568.94) and (342.41,570.86) ..
				(327.62,558.8) ;
				\draw [shift={(325.5,557)}, rotate = 41.19] [fill={rgb, 255:red, 0; green, 0;
					blue, 0 }  ][line width=0.08]  [draw opacity=0] (8.93,-4.29) -- (0,0) --
				(8.93,4.29) -- cycle    ;
				\draw    (411.5,546.5) .. controls (399.09,568.94) and (389.41,570.86) ..
				(374.62,558.8) ;
				\draw [shift={(372.5,557)}, rotate = 41.19] [fill={rgb, 255:red, 0; green, 0;
					blue, 0 }  ][line width=0.08]  [draw opacity=0] (8.93,-4.29) -- (0,0) --
				(8.93,4.29) -- cycle    ;
				\draw  [fill={rgb, 255:red, 0; green, 0; blue, 255 }  ,fill opacity=1 ]
				(120,546.5) .. controls (120,541.25) and (124.25,537) .. (129.5,537) .. controls
				(134.75,537) and (139,541.25) .. (139,546.5) .. controls (139,551.75) and
				(134.75,556) .. (129.5,556) .. controls (124.25,556) and (120,551.75) ..
				(120,546.5) -- cycle ;
				\draw  [fill={rgb, 255:red, 255; green, 255; blue, 255 }  ,fill opacity=1 ]
				(167,546.5) .. controls (167,541.25) and (171.25,537) .. (176.5,537) .. controls
				(181.75,537) and (186,541.25) .. (186,546.5) .. controls (186,551.75) and
				(181.75,556) .. (176.5,556) .. controls (171.25,556) and (167,551.75) ..
				(167,546.5) -- cycle ;
				\draw  [fill={rgb, 255:red, 255; green, 255; blue, 255 }  ,fill opacity=1 ]
				(214,546.5) .. controls (214,541.25) and (218.25,537) .. (223.5,537) .. controls
				(228.75,537) and (233,541.25) .. (233,546.5) .. controls (233,551.75) and
				(228.75,556) .. (223.5,556) .. controls (218.25,556) and (214,551.75) ..
				(214,546.5) -- cycle ;
				\draw  [fill={rgb, 255:red, 255; green, 255; blue, 255 }  ,fill opacity=1 ]
				(261,546.5) .. controls (261,541.25) and (265.25,537) .. (270.5,537) .. controls
				(275.75,537) and (280,541.25) .. (280,546.5) .. controls (280,551.75) and
				(275.75,556) .. (270.5,556) .. controls (265.25,556) and (261,551.75) ..
				(261,546.5) -- cycle ;
				\draw  [fill={rgb, 255:red, 255; green, 255; blue, 255 }  ,fill opacity=1 ]
				(308,546.5) .. controls (308,541.25) and (312.25,537) .. (317.5,537) .. controls
				(322.75,537) and (327,541.25) .. (327,546.5) .. controls (327,551.75) and
				(322.75,556) .. (317.5,556) .. controls (312.25,556) and (308,551.75) ..
				(308,546.5) -- cycle ;
				\draw  [fill={rgb, 255:red, 255; green, 255; blue, 255 }  ,fill opacity=1 ]
				(355,546.5) .. controls (355,541.25) and (359.25,537) .. (364.5,537) .. controls
				(369.75,537) and (374,541.25) .. (374,546.5) .. controls (374,551.75) and
				(369.75,556) .. (364.5,556) .. controls (359.25,556) and (355,551.75) ..
				(355,546.5) -- cycle ;
				\draw  [fill={rgb, 255:red, 0; green, 0; blue, 255 }  ,fill opacity=1 ]
				(402,546.5) .. controls (402,541.25) and (406.25,537) .. (411.5,537) .. controls
				(416.75,537) and (421,541.25) .. (421,546.5) .. controls (421,551.75) and
				(416.75,556) .. (411.5,556) .. controls (406.25,556) and (402,551.75) ..
				(402,546.5) -- cycle ;
				\draw  [fill={rgb, 255:red, 0; green, 0; blue, 255 }  ,fill opacity=1 ]
				(261,609) .. controls (261,603.75) and (265.25,599.5) .. (270.5,599.5) ..
				controls (275.75,599.5) and (280,603.75) .. (280,609) .. controls (280,614.25)
				and (275.75,618.5) .. (270.5,618.5) .. controls (265.25,618.5) and (261,614.25)
				.. (261,609) -- cycle ;
				\draw  [fill={rgb, 255:red, 0; green, 0; blue, 255 }  ,fill opacity=1 ]
				(260,487) .. controls (260,481.75) and (264.25,477.5) .. (269.5,477.5) ..
				controls (274.75,477.5) and (279,481.75) .. (279,487) .. controls (279,492.25)
				and (274.75,496.5) .. (269.5,496.5) .. controls (264.25,496.5) and (260,492.25)
				.. (260,487) -- cycle ;
				\draw    (283.68,480.56) .. controls (374.22,471.2) and (397.64,477.18) ..
				(406.5,423.5) ;
				\draw [shift={(279.5,481)}, rotate = 353.93] [fill={rgb, 255:red, 0; green,
					0; blue, 0 }  ][line width=0.08]  [draw opacity=0] (8.93,-4.29) -- (0,0) --
				(8.93,4.29) -- cycle    ;
				\draw    (253.4,479.34) .. controls (164.58,465.29) and (134.44,474.23) ..
				(130.5,423.5) ;
				\draw [shift={(257.5,480)}, rotate = 189.26] [fill={rgb, 255:red, 0; green,
					0; blue, 0 }  ][line width=0.08]  [draw opacity=0] (8.93,-4.29) -- (0,0) --
				(8.93,4.29) -- cycle    ;
				\draw  [fill={rgb, 255:red, 0; green, 255; blue, 0 }  ,fill opacity=1 ]
				(551,178.5) .. controls (551,173.25) and (555.25,169) .. (560.5,169) .. controls
				(565.75,169) and (570,173.25) .. (570,178.5) .. controls (570,183.75) and
				(565.75,188) .. (560.5,188) .. controls (555.25,188) and (551,183.75) ..
				(551,178.5) -- cycle ;
				\draw  [fill={rgb, 255:red, 0; green, 255; blue, 0 }  ,fill opacity=1 ]
				(551,238.75) .. controls (551,233.5) and (555.25,229.25) .. (560.5,229.25) ..
				controls (565.75,229.25) and (570,233.5) .. (570,238.75) .. controls (570,244)
				and (565.75,248.25) .. (560.5,248.25) .. controls (555.25,248.25) and (551,244)
				.. (551,238.75) -- cycle ;
				\draw  [fill={rgb, 255:red, 255; green, 255; blue, 255 }  ,fill opacity=1 ]
				(551,383) .. controls (551,377.75) and (555.25,373.5) .. (560.5,373.5) ..
				controls (565.75,373.5) and (570,377.75) .. (570,383) .. controls (570,388.25)
				and (565.75,392.5) .. (560.5,392.5) .. controls (555.25,392.5) and (551,388.25)
				.. (551,383) -- cycle ;
				
				\draw (263,366.4) node [anchor=north west][inner sep=0.75pt]    {\Huge
					$\vdots $};
				\draw (92,91.4) node [anchor=north west][inner sep=0.75pt]    {$x_{1}$};
				\draw (90,214.4) node [anchor=north west][inner sep=0.75pt]    {$x_{2}$};
				\draw (92,536.4) node [anchor=north west][inner sep=0.75pt]    {$x_{n}$};
				\draw (260,10.4) node [anchor=north west][inner sep=0.75pt]    {$s$};
				\draw (266,623.4) node [anchor=north west][inner sep=0.75pt]    {$t$};
				\draw (552,297.4) node [anchor=north west][inner sep=0.75pt]    {\Huge
					$\vdots $};
				\draw (583,172.4) node [anchor=north west][inner sep=0.75pt]    {$C_{1}$};
				\draw (584,231.4) node [anchor=north west][inner sep=0.75pt]    {$C_{2}$};
				\draw (582,375.4) node [anchor=north west][inner sep=0.75pt]    {$C_{m}$};

			\end{tikzpicture}

			\caption{The graph $f_\p^\ast(\varphi)$. There are $n$ rows (one per
				variable). The edges between the vertices corresponding to the clauses and the
				rest of the graph are not represented here. The green vertices are those in
				$V_{x_1}$ (assuming that $x_1$ or its negation belongs to $C_1$ and $C_2$, but
				not to $C_m$). The blue vertices are those in $C_{x_1} \setminus V_{x_1}$.}
			\label{fig:ham1}
		\end{figure}

		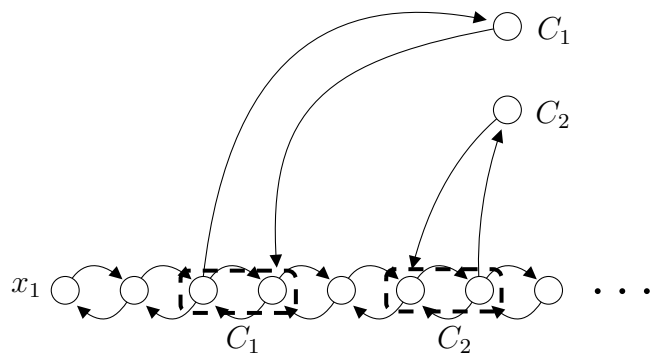
\begin{figure}
			\centering
			
			\begin{tikzpicture}[x=0.55pt,y=0.55pt,yscale=-1,xscale=1]
				
				\draw    (410.5,224.5) .. controls (407.59,186.18) and (414.09,147.4) ..
				(424.52,116.81) ;
				\draw [shift={(425.5,114)}, rotate = 109.54] [fill={rgb, 255:red, 0; green,
					0; blue, 0 }  ][line width=0.08]  [draw opacity=0] (8.93,-4.29) -- (0,0) --
				(8.93,4.29) -- cycle    ;
				\draw    (365.35,206.67) .. controls (379.5,152.69) and (404.02,120.6) ..
				(429.5,100.75) ;
				\draw [shift={(364.5,210)}, rotate = 284.04] [fill={rgb, 255:red, 0; green,
					0; blue, 0 }  ][line width=0.08]  [draw opacity=0] (8.93,-4.29) -- (0,0) --
				(8.93,4.29) -- cycle    ;
				\draw    (222.5,215) .. controls (241.31,49.67) and (321.87,20.57) ..
				(412.74,38.44) ;
				\draw [shift={(415.5,39)}, rotate = 191.67] [fill={rgb, 255:red, 0; green, 0;
					blue, 0 }  ][line width=0.08]  [draw opacity=0] (8.93,-4.29) -- (0,0) --
				(8.93,4.29) -- cycle    ;
				\draw    (272.27,205.72) .. controls (265.29,98.34) and (304.76,65.78) ..
				(429.5,43.5) ;
				\draw [shift={(272.5,209)}, rotate = 265.84] [fill={rgb, 255:red, 0; green,
					0; blue, 0 }  ][line width=0.08]  [draw opacity=0] (8.93,-4.29) -- (0,0) --
				(8.93,4.29) -- cycle    ;
				\draw    (410.5,224.5) .. controls (417.26,209.54) and (428.67,200.17) ..
				(447.43,215.26) ;
				\draw [shift={(449.5,217)}, rotate = 221.19] [fill={rgb, 255:red, 0; green,
					0; blue, 0 }  ][line width=0.08]  [draw opacity=0] (8.93,-4.29) -- (0,0) --
				(8.93,4.29) -- cycle    ;
				\draw    (175.5,224.5) .. controls (163.08,246.94) and (153.41,248.86) ..
				(138.62,236.8) ;
				\draw [shift={(136.5,235)}, rotate = 41.19] [fill={rgb, 255:red, 0; green, 0;
					blue, 0 }  ][line width=0.08]  [draw opacity=0] (8.93,-4.29) -- (0,0) --
				(8.93,4.29) -- cycle    ;
				\draw    (128.5,224.5) .. controls (135.26,209.54) and (146.67,200.17) ..
				(165.43,215.26) ;
				\draw [shift={(167.5,217)}, rotate = 221.19] [fill={rgb, 255:red, 0; green,
					0; blue, 0 }  ][line width=0.08]  [draw opacity=0] (8.93,-4.29) -- (0,0) --
				(8.93,4.29) -- cycle    ;
				\draw    (175.5,224.5) .. controls (182.26,209.54) and (193.67,200.17) ..
				(212.43,215.26) ;
				\draw [shift={(214.5,217)}, rotate = 221.19] [fill={rgb, 255:red, 0; green,
					0; blue, 0 }  ][line width=0.08]  [draw opacity=0] (8.93,-4.29) -- (0,0) --
				(8.93,4.29) -- cycle    ;
				\draw    (222.5,224.5) .. controls (229.26,209.54) and (240.67,200.17) ..
				(259.43,215.26) ;
				\draw [shift={(261.5,217)}, rotate = 221.19] [fill={rgb, 255:red, 0; green,
					0; blue, 0 }  ][line width=0.08]  [draw opacity=0] (8.93,-4.29) -- (0,0) --
				(8.93,4.29) -- cycle    ;
				\draw    (269.5,224.5) .. controls (276.26,209.54) and (287.67,200.17) ..
				(306.43,215.26) ;
				\draw [shift={(308.5,217)}, rotate = 221.19] [fill={rgb, 255:red, 0; green,
					0; blue, 0 }  ][line width=0.08]  [draw opacity=0] (8.93,-4.29) -- (0,0) --
				(8.93,4.29) -- cycle    ;
				\draw    (316.5,224.5) .. controls (323.26,209.54) and (334.67,200.17) ..
				(353.43,215.26) ;
				\draw [shift={(355.5,217)}, rotate = 221.19] [fill={rgb, 255:red, 0; green,
					0; blue, 0 }  ][line width=0.08]  [draw opacity=0] (8.93,-4.29) -- (0,0) --
				(8.93,4.29) -- cycle    ;
				\draw    (363.5,224.5) .. controls (370.26,209.54) and (381.67,200.17) ..
				(400.43,215.26) ;
				\draw [shift={(402.5,217)}, rotate = 221.19] [fill={rgb, 255:red, 0; green,
					0; blue, 0 }  ][line width=0.08]  [draw opacity=0] (8.93,-4.29) -- (0,0) --
				(8.93,4.29) -- cycle    ;
				\draw    (222.5,224.5) .. controls (210.08,246.94) and (200.41,248.86) ..
				(185.62,236.8) ;
				\draw [shift={(183.5,235)}, rotate = 41.19] [fill={rgb, 255:red, 0; green, 0;
					blue, 0 }  ][line width=0.08]  [draw opacity=0] (8.93,-4.29) -- (0,0) --
				(8.93,4.29) -- cycle    ;
				\draw    (269.5,224.5) .. controls (257.08,246.94) and (247.41,248.86) ..
				(232.62,236.8) ;
				\draw [shift={(230.5,235)}, rotate = 41.19] [fill={rgb, 255:red, 0; green, 0;
					blue, 0 }  ][line width=0.08]  [draw opacity=0] (8.93,-4.29) -- (0,0) --
				(8.93,4.29) -- cycle    ;
				\draw    (316.5,224.5) .. controls (304.08,246.94) and (294.41,248.86) ..
				(279.62,236.8) ;
				\draw [shift={(277.5,235)}, rotate = 41.19] [fill={rgb, 255:red, 0; green, 0;
					blue, 0 }  ][line width=0.08]  [draw opacity=0] (8.93,-4.29) -- (0,0) --
				(8.93,4.29) -- cycle    ;
				\draw    (363.5,224.5) .. controls (351.09,246.94) and (341.41,248.86) ..
				(326.62,236.8) ;
				\draw [shift={(324.5,235)}, rotate = 41.19] [fill={rgb, 255:red, 0; green, 0;
					blue, 0 }  ][line width=0.08]  [draw opacity=0] (8.93,-4.29) -- (0,0) --
				(8.93,4.29) -- cycle    ;
				\draw    (410.5,224.5) .. controls (398.09,246.94) and (388.41,248.86) ..
				(373.62,236.8) ;
				\draw [shift={(371.5,235)}, rotate = 41.19] [fill={rgb, 255:red, 0; green, 0;
					blue, 0 }  ][line width=0.08]  [draw opacity=0] (8.93,-4.29) -- (0,0) --
				(8.93,4.29) -- cycle    ;
				\draw  [fill={rgb, 255:red, 255; green, 255; blue, 255 }  ,fill opacity=1 ]
				(119,224.5) .. controls (119,219.25) and (123.25,215) .. (128.5,215) .. controls
				(133.75,215) and (138,219.25) .. (138,224.5) .. controls (138,229.75) and
				(133.75,234) .. (128.5,234) .. controls (123.25,234) and (119,229.75) ..
				(119,224.5) -- cycle ;
				\draw  [fill={rgb, 255:red, 255; green, 255; blue, 255 }  ,fill opacity=1 ]
				(166,224.5) .. controls (166,219.25) and (170.25,215) .. (175.5,215) .. controls
				(180.75,215) and (185,219.25) .. (185,224.5) .. controls (185,229.75) and
				(180.75,234) .. (175.5,234) .. controls (170.25,234) and (166,229.75) ..
				(166,224.5) -- cycle ;
				\draw  [fill={rgb, 255:red, 255; green, 255; blue, 255 }  ,fill opacity=1 ]
				(213,224.5) .. controls (213,219.25) and (217.25,215) .. (222.5,215) .. controls
				(227.75,215) and (232,219.25) .. (232,224.5) .. controls (232,229.75) and
				(227.75,234) .. (222.5,234) .. controls (217.25,234) and (213,229.75) ..
				(213,224.5) -- cycle ;
				\draw  [fill={rgb, 255:red, 255; green, 255; blue, 255 }  ,fill opacity=1 ]
				(260,224.5) .. controls (260,219.25) and (264.25,215) .. (269.5,215) .. controls
				(274.75,215) and (279,219.25) .. (279,224.5) .. controls (279,229.75) and
				(274.75,234) .. (269.5,234) .. controls (264.25,234) and (260,229.75) ..
				(260,224.5) -- cycle ;
				\draw  [fill={rgb, 255:red, 255; green, 255; blue, 255 }  ,fill opacity=1 ]
				(307,224.5) .. controls (307,219.25) and (311.25,215) .. (316.5,215) .. controls
				(321.75,215) and (326,219.25) .. (326,224.5) .. controls (326,229.75) and
				(321.75,234) .. (316.5,234) .. controls (311.25,234) and (307,229.75) ..
				(307,224.5) -- cycle ;
				\draw  [fill={rgb, 255:red, 255; green, 255; blue, 255 }  ,fill opacity=1 ]
				(354,224.5) .. controls (354,219.25) and (358.25,215) .. (363.5,215) .. controls
				(368.75,215) and (373,219.25) .. (373,224.5) .. controls (373,229.75) and
				(368.75,234) .. (363.5,234) .. controls (358.25,234) and (354,229.75) ..
				(354,224.5) -- cycle ;
				\draw  [fill={rgb, 255:red, 255; green, 255; blue, 255 }  ,fill opacity=1 ]
				(401,224.5) .. controls (401,219.25) and (405.25,215) .. (410.5,215) .. controls
				(415.75,215) and (420,219.25) .. (420,224.5) .. controls (420,229.75) and
				(415.75,234) .. (410.5,234) .. controls (405.25,234) and (401,229.75) ..
				(401,224.5) -- cycle ;
				\draw  [fill={rgb, 255:red, 255; green, 255; blue, 255 }  ,fill opacity=1 ]
				(420,43.5) .. controls (420,38.25) and (424.25,34) .. (429.5,34) .. controls
				(434.75,34) and (439,38.25) .. (439,43.5) .. controls (439,48.75) and
				(434.75,53) .. (429.5,53) .. controls (424.25,53) and (420,48.75) .. (420,43.5)
				-- cycle ;
				\draw  [fill={rgb, 255:red, 255; green, 255; blue, 255 }  ,fill opacity=1 ]
				(420,100.75) .. controls (420,95.5) and (424.25,91.25) .. (429.5,91.25) ..
				controls (434.75,91.25) and (439,95.5) .. (439,100.75) .. controls (439,106) and
				(434.75,110.25) .. (429.5,110.25) .. controls (424.25,110.25) and (420,106) ..
				(420,100.75) -- cycle ;
				\draw  [dash pattern={on 5.63pt off 4.5pt}][line width=1.5]  (207,216.8) ..
				controls (207,213.6) and (209.6,211) .. (212.8,211) -- (279.7,211) .. controls
				(282.9,211) and (285.5,213.6) .. (285.5,216.8) -- (285.5,234.2) .. controls
				(285.5,237.4) and (282.9,240) .. (279.7,240) -- (212.8,240) .. controls
				(209.6,240) and (207,237.4) .. (207,234.2) -- cycle ;
				\draw    (457.5,224.5) .. controls (445.09,246.94) and (435.41,248.86) ..
				(420.62,236.8) ;
				\draw [shift={(418.5,235)}, rotate = 41.19] [fill={rgb, 255:red, 0; green, 0;
					blue, 0 }  ][line width=0.08]  [draw opacity=0] (8.93,-4.29) -- (0,0) --
				(8.93,4.29) -- cycle    ;
				\draw  [fill={rgb, 255:red, 255; green, 255; blue, 255 }  ,fill opacity=1 ]
				(448,224.5) .. controls (448,219.25) and (452.25,215) .. (457.5,215) .. controls
				(462.75,215) and (467,219.25) .. (467,224.5) .. controls (467,229.75) and
				(462.75,234) .. (457.5,234) .. controls (452.25,234) and (448,229.75) ..
				(448,224.5) -- cycle ;
				\draw  [dash pattern={on 5.63pt off 4.5pt}][line width=1.5]  (347,215.8) ..
				controls (347,212.6) and (349.6,210) .. (352.8,210) -- (419.7,210) .. controls
				(422.9,210) and (425.5,212.6) .. (425.5,215.8) -- (425.5,233.2) .. controls
				(425.5,236.4) and (422.9,239) .. (419.7,239) -- (352.8,239) .. controls
				(349.6,239) and (347,236.4) .. (347,233.2) -- cycle ;
				
				\draw (90,214.4) node [anchor=north west][inner sep=0.75pt]    {$x_{1}$};
				\draw (448,35.4) node [anchor=north west][inner sep=0.75pt]    {$C_{1}$};
				\draw (447,93.4) node [anchor=north west][inner sep=0.75pt]    {$C_{2}$};
				\draw (484,216) node [anchor=north west][inner sep=0.75pt]    {\LARGE $\cdots
					$};
				\draw (236,247.4) node [anchor=north west][inner sep=0.75pt]    {$C_{1}$};
				\draw (380,247.4) node [anchor=north west][inner sep=0.75pt]    {$C_{2}$};

			\end{tikzpicture}

			\caption{The connection between the nodes corresponding to clauses and the
				rest of the graph. On this example, $x_1$ appears in $C_1$, $\overline{x_1}$
				appears in $C_2$.}
			\label{fig:ham2}
		\end{figure}

		Now, let us explain how to construct the graph $f_\p(\varphi)$ from
		$f_\p^\ast(\varphi)$. We proceed again as in~\cite{Sipser}, by replacing each
		vertex $u$ of $f_\p^\ast(\varphi)$ by three vertices $u_{\text{in}}$,
		$u_{\text{mid}}$, $u_{\text{out}}$. These three vertices form a path, in this
		order. Each (directed) edge from $u$ to $v$ is replaced an (undirected) edge
		from $u_{\text{out}}$ to $v_{\text{in}}$.
		
		The identifiers of the vertices of $f_\p(\varphi)$ consist in a pair: the
		first entry indicates the identifier of the corresponding vertex $u$ in
		$f_\p^\ast(\varphi)$ (see thereafter the definition), and the second entry
		indicates if it is $u_{\text{in}}$, $u_{\text{mid}}$ or $u_{\text{out}}$. In
		$f_\p^\ast(\varphi)$, the identifier of each vertex $u$ indicates if it is a
		vertex corresponding to a clause (by encoding the identifiers of the three
		variables in it), or to a variable, and in this latter case if it is linked to a
		clause vertex (and which one) or if it is a separator between two pairs (and the
		corresponding pairs it separates). All this information can be encoded with
		$O(\log n)$ bits.
		
		For every variable $x_i$ of $\varphi$, let us describe the sets $C_{x_i}$ and
		$V_{x_i}$ (and see Figure~\ref{fig:ham1} for an example).
		\begin{itemize}
			\item The set $V_{x_i}$ consists in all the vertices $u_{\text{in}}$,
			$u_{\text{mid}}$, $u_{\text{out}}$ for each vertex $u$ which is either the entry
			node corresponding to $x_i$, or the exit node, or a node in the row of $x_i$, or
			a node corresponding to a clause to which $x_i$ or its negation belongs. These
			vertices are colored in green on Figure~\ref{fig:ham1}.
			
			\item The set $C_{x_i}$ consists in $V_{x_i}$, plus all the vertices
			$u_{\text{in}}$, $u_{\text{mid}}$, $u_{\text{out}}$ for $u$ which is either an
			entry/exit node of any variable, or a first/last vertex in the row of any
			variable. These vertices are colored in blue on Figure~\ref{fig:ham1}.
		\end{itemize}
		
		Let us verify that the properties of the definition of local reduction from 3-SAT (see the beginning of section~\ref{sec:red3SAT}) are satisfied. For \cref{s1}, the identifiers are indeed on $O(\log |V(\varphi)|)$
		bits. For \cref{s2}, it is proved
                in~\cite{Sipser}. Properties \cref{s3a}, \cref{s3b}, \cref{s3d}, \cref{s3e}, and \cref{s3f} follow from the definition of $C_{x_i}$ and $V_{x_i}$. For \cref{s3c}, we indeed
		have $|C_{x_i}| = O(n)$ for every variable $x_i$: this holds because, by
		definition of $\F_0$, $x_i$ appears in at most $n$ clauses.
		
		Thus, we can apply Corollary~\ref{cor:reduction3SAT} which gives us the
		$\Omega\left(\frac{\sqrt{n}}{\log n}\right)$ lower bound. Since $f_\p(G)$ has
		maximum degree~at most~4, the lower bound holds even for this class of graphs.
	\end{proof}

	\subsection{Chromatic index}

	The \emph{chromatic index} of a graph $G$ is the least $k$
        such that the edges of $G$ can be colored with $k$ colors, so
        that any two edges sharing a endpoint have distinct colors.
   By Vizing's theorem, if~$\Delta$ is the maximum degree of~$G$, its chromatic index is either~$\Delta$ or $\Delta+1$.
	
	\begin{theorem}
		\label{thm:coloration aretes}
		The property of  having chromatic
                index $\Delta+1$ (where $\Delta$ denotes the maximum degree) has local complexity $\Omega(n/\log
                n)$, even in cubic graphs.
	\end{theorem}
	
	
	
	\begin{proof}[Proof of Theorem~\ref{thm:coloration aretes}]
		Let $\p$ be the property of having chromatic index $\Delta$. We will show that
		there exists a local reduction from 3-SAT in $\F_1$ to $\p$ with local expansion
		$O(1)$, and then the result will follow from
                Corollary~\ref{cor:reduction3SAT}. Recall that  $\F_1$
                was defined as the set of 3-CNF
	formulae $\varphi$ such that each variable appears in at most 7~clauses.
		Let $\varphi$ be a 3-CNF formula with variables having identifiers in
		$\{1, \ldots, |V(\varphi)|\}$, and such that each variable appears in at most
		7~clauses. Let $n:=|V(\varphi)|$, and let $m$ be the number of clauses in
		$\varphi$. Let $x_1, \ldots, x_n$ be the variables of $\varphi$, and $C_1,
		\ldots, C_m$ be its clauses.
		
		We use the reduction presented in~\cite{Holyer81}. For completeness, we recall
		this construction. First, we define an \emph{inverting component}, which is
		shown on Figure~\ref{fig:inverting}.
		
		\begin{figure}[h]
			\centering
			\includegraphics[scale=1]{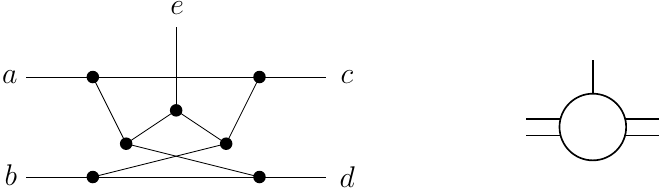}
			\caption{An inverting component and its symbolic representation.}
			\label{fig:inverting}
		\end{figure}
		
		In an edge-coloring, a pair of edges is said to be \emph{true} (resp.\
		\emph{false}) if it is colored with the same colors (resp.\ with different
		colors). In any proper 3-edge-coloring of an inverting component,
		exactly one of the two pairs $(a,b)$ or $(c,d)$ is true, and the three other
		edges have different colors. Thus, regarding $(a,b)$ as the input and $(c,d)$ as
		the output, an inverting component changes true to false and vice-versa.
		
		Then, we define a \emph{variable-setting} component with $k$ outputs. It consists in a cycle of $k$ pairs of inverting components, as shown on
			Figure~\ref{fig:variable setting} with $k=4$.
		In any proper 3-coloring of the edges, all the output pairs must have the same
		value, either all true or all false.
		
		\begin{figure}[h]
			\centering
			\includegraphics[scale=0.6]{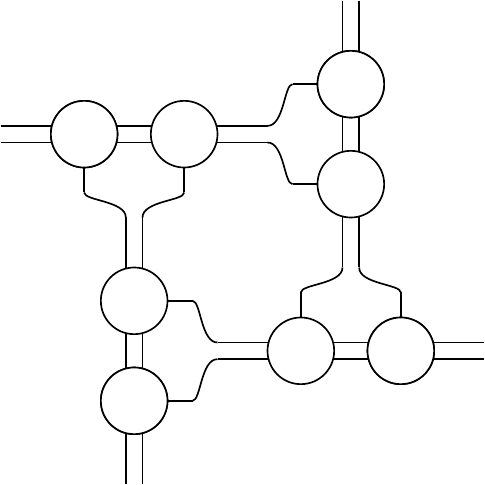}
			\caption{A variable-setting component with four outputs.}
			\label{fig:variable setting}
		\end{figure}
		
		Finally, for each clause, we will create a \emph{satisfaction-testing
			component}, represented on Figure~\ref{fig:satisfaction testing}. This component
		has a proper 3-coloring of the edges if and only if at least one of its three
		input pairs is true.
		
		\begin{figure}[h]
			\centering
			\includegraphics[scale=0.6]{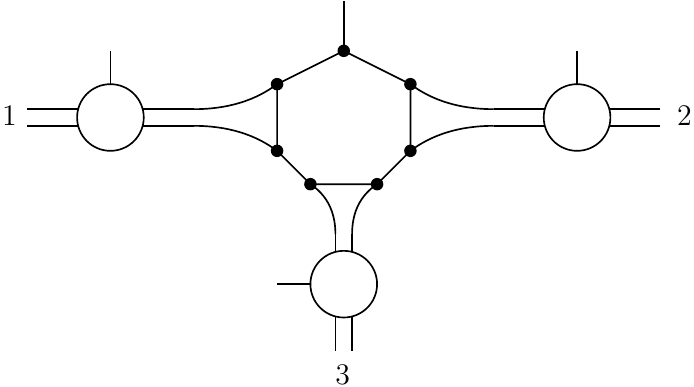}
			\caption{A satisfaction-testing component.}
			\label{fig:satisfaction testing}
		\end{figure}
		
		The graph $f_\p(\varphi)$ is constructed as follows.
		For each variable $x_i$ such that $x_i$ or its negation belongs to $d$ clauses,
		we create a variable-setting component with $d$ outputs, numbered from~$1$
		to~$d$.
		For each clause $C_j$, we create a satisfaction-testing component. Let
		$x_{j_1}, x_{j_2}, x_{j_3}$ be the variables which appear (positively or
		negatively) in $C_j$, by increasing order of their identifiers. Let $\ell$ be
		the index such that, by sorting the clauses containing $x_{j_1}$ or its negation
		by lexicographic order of the variables appearing inside, $C_j$ is the
		$\ell$-th. We connect the $\ell$-th output of the variable-setting component of
		$x_{j_1}$ to the input~1 of the satisfaction-testing component of $C_j$, and if
		$x_{j_1}$ appears negatively in $C_j$, we insert an inverting component in this
		branching. We do the same for $x_{j_2}$ and $x_{j_3}$ with inputs~2 and~3 of the
		satisfaction-testing component.
		This results in a graph $H$, which has some edges having an endpoint of
		degree~1 (in the satisfaction-testing components).
		The graph $f_\p(\varphi)$ is obtained by taking two copies of this graph $H$
		and identifying these edges. It results in a cubic graph.
		
		The identifier of each vertex $u$ of $f_\p(G)$ is the following. First, it
		indicates to which of the two copies of the graph $H$ does $u$ belong to.
		Then, if $u$ belongs to a variable-setting component, it indicates to which
		variable it corresponds, its position in the variable-setting component, and if it is adjacent to an output, it also indicates the clause corresponding to the satisfaction-testing component to which it is adjacent.
		If $u$ belongs to a satisfaction-testing component, it indicates to which
		clause it corresponds, and its position in the satisfaction-testing component.
		All this information can be encoded on $O(\log n)$ bits.
		
		For every variable $x_i$ of $\varphi$, we define $C_{x_i} = V_{x_i}$ as being
		the set which contains all the vertices in the two copies of $H$ which are in
		the variable-setting component of $x_i$, in the satisfaction-testing of any
		clause to which $x_i$ belongs, and in an inverting component branched between
		its variable-setting component and the satisfaction-testing component of a
		clause in which it appears negatively.
		
		Let us verify that the properties of the definition of local reduction from 3-SAT (see the beginning of section~\ref{sec:red3SAT}) are satisfied. Property~\cref{s1} is true by definition of the identifiers of the
		vertices of $f_\p(\varphi)$, and~\cref{s2} comes from the correctness of the reduction
		in~\cite{Holyer81}. Properties \cref{s3a}, \cref{s3b}, \cref{s3e} and~\cref{s3f} follow from the
		definition of $V_{x_i}$ and $C_{x_i}$. For \cref{s3f}, we have $C_{x_i} = O(1)$ for
		every variable $x_i$, because by definition of $\F_1$, $x_i$ appears in a
		constant number of clauses (at most~7). Finally, \cref{s3d} is true because if
		$C_{x_i} \cap C_{x_j} \neq \emptyset$ then either $x_i=x_j$, or $x_i$ and $x_j$
		have a clause in common.
		
		Thus, Corollary~\ref{cor:reduction3SAT} gives us a
                lower bound of order
		$\Omega(n/\log n)$ on the local complexity of
                of the problem, and the lower bound holds even for cubic
		graphs because $f_\p(\varphi)$ is a cubic graph.
	\end{proof}

	\section{Variant of our theorems in a bounded-degree class}
	\label{sec:variants}
	
	
	In this subsection, we give an example of a lower bound for
        which the proof does not follow from a direct application of our
        general results on local reductions, because the definition of
        local reduction requires a small tweak to be applied
        efficiently here. This lower bound is stated in the following
        theorem. We then show that the bound on the maximum degree~is
        close to 
        optimal: if the maximum degree~is decreased by two, then the
        local complexity of the problem drops from $\Omega(n/\log n)$ to $O(\log n)$.

	\begin{theorem}
		\label{thm:nonkcol_lineaire}
		Let $k \geqslant 3$. In graphs of maximum degree~$k + \lceil\sqrt{k}\rceil - 1$, non-$k$-colorability has local
                complexity $\Omega(n/\log n)$.
	\end{theorem}
	
	Note that Theorem~\ref{thm:nonkcol_lineaire} generalizes Theorem~\ref{thm:non3coldegmax4_lineaire} to non-$k$-colorability.
	To prove it, we will rely on the reduction presented in~\cite{Emden-WeinertHK98}, that the authors use to prove the \mbox{\textsf{NP}-completeness} of the $k$-colorability problem in graphs of maximum degree~$k + \lceil\sqrt{k}\rceil - 1$. First, let us recall this \textsf{NP}-completeness reduction. Given a graph $G$ of maximum degree~larger than $k + \lceil\sqrt{k}\rceil - 1$, we perform iteratively the following operation: we pick a vertex $u \in V(G)$ of degree~$d(u) > k + \lceil\sqrt{k}\rceil - 1$, add vertices $r_1, \ldots, r_{\lceil\sqrt{k}\rceil}$ to $G$, detach $k$ edges incident to $u$ and attach them evenly to $r_1, \ldots, r_{\lceil\sqrt{k}\rceil}$. Finally, we add a $(k-1)$-clique $K_{k-1}$ and connect~$u$ and all the vertices $r_1, \ldots, r_{\lceil\sqrt{k}\rceil}$ to each vertex of $K_{k-1}$. It is straightforward that this operation preserves the $k$-colorability property, and decreases the degree~of $u$ by one, without creating vertices of degree~larger than $k + \lceil\sqrt{k}\rceil - 1$. Thus, at the end of the process, all the vertices have degree~at most $k + \lceil\sqrt{k}\rceil - 1$ and the $k$-colorability property of the graph has been preserved. See Figure~\ref{fig:kcol_degree_decreasing} for an example of this operation.
	In the following, we will use it to prove Theorem~\ref{thm:nonkcol_lineaire}.

	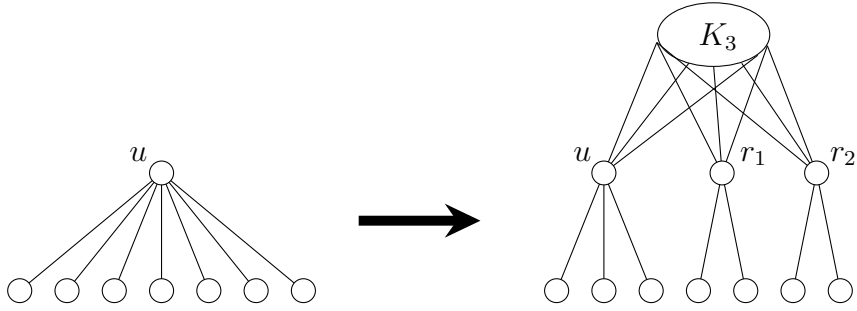
\begin{figure}[h!]
	\centering
	
	\begin{tikzpicture}[x=0.6pt,y=0.6pt,yscale=-1,xscale=1]
		
		\draw    (577.25,114) -- (530.5,45) ;
		\draw    (518.25,114) -- (513,47) ;
		\draw    (444.5,114) -- (497.5,45) ;
		\draw    (477.5,32) -- (444.5,114) ;
		\draw    (540.5,40) -- (444.5,114) ;
		\draw    (477.5,32) -- (518.25,114) ;
		\draw    (546.5,34) -- (518.25,114) ;
		\draw    (477.5,32) -- (577.25,114) ;
		\draw    (546.5,34) -- (577.25,114) ;
		\draw    (444.5,114) -- (444.5,188) ;
		\draw    (518.25,114) -- (533,188) ;
		\draw    (577.25,114) -- (592,188) ;
		\draw    (444.5,114) -- (415,188) ;
		\draw    (444.5,114) -- (474,188) ;
		\draw    (518.25,114) -- (503.5,188) ;
		\draw    (577.25,114) -- (562.5,188) ;
		\draw    (168.5,114) -- (168.5,188) ;
		\draw    (168.5,114) -- (80,188) ;
		\draw    (168.5,114) -- (109.5,188) ;
		\draw    (168.5,114) -- (139,188) ;
		\draw    (168.5,114) -- (198,188) ;
		\draw    (168.5,114) -- (227.5,188) ;
		\draw    (168.5,114) -- (257,188) ;
		\draw  [fill={rgb, 255:red, 255; green, 255; blue, 255 }  ,fill opacity=1 ] (72.5,188) .. controls (72.5,183.86) and (75.86,180.5) .. (80,180.5) .. controls (84.14,180.5) and (87.5,183.86) .. (87.5,188) .. controls (87.5,192.14) and (84.14,195.5) .. (80,195.5) .. controls (75.86,195.5) and (72.5,192.14) .. (72.5,188) -- cycle ;
		\draw  [fill={rgb, 255:red, 255; green, 255; blue, 255 }  ,fill opacity=1 ] (102,188) .. controls (102,183.86) and (105.36,180.5) .. (109.5,180.5) .. controls (113.64,180.5) and (117,183.86) .. (117,188) .. controls (117,192.14) and (113.64,195.5) .. (109.5,195.5) .. controls (105.36,195.5) and (102,192.14) .. (102,188) -- cycle ;
		\draw  [fill={rgb, 255:red, 255; green, 255; blue, 255 }  ,fill opacity=1 ] (131.5,188) .. controls (131.5,183.86) and (134.86,180.5) .. (139,180.5) .. controls (143.14,180.5) and (146.5,183.86) .. (146.5,188) .. controls (146.5,192.14) and (143.14,195.5) .. (139,195.5) .. controls (134.86,195.5) and (131.5,192.14) .. (131.5,188) -- cycle ;
		\draw  [fill={rgb, 255:red, 255; green, 255; blue, 255 }  ,fill opacity=1 ] (161,188) .. controls (161,183.86) and (164.36,180.5) .. (168.5,180.5) .. controls (172.64,180.5) and (176,183.86) .. (176,188) .. controls (176,192.14) and (172.64,195.5) .. (168.5,195.5) .. controls (164.36,195.5) and (161,192.14) .. (161,188) -- cycle ;
		\draw  [fill={rgb, 255:red, 255; green, 255; blue, 255 }  ,fill opacity=1 ] (190.5,188) .. controls (190.5,183.86) and (193.86,180.5) .. (198,180.5) .. controls (202.14,180.5) and (205.5,183.86) .. (205.5,188) .. controls (205.5,192.14) and (202.14,195.5) .. (198,195.5) .. controls (193.86,195.5) and (190.5,192.14) .. (190.5,188) -- cycle ;
		\draw  [fill={rgb, 255:red, 255; green, 255; blue, 255 }  ,fill opacity=1 ] (249.5,188) .. controls (249.5,183.86) and (252.86,180.5) .. (257,180.5) .. controls (261.14,180.5) and (264.5,183.86) .. (264.5,188) .. controls (264.5,192.14) and (261.14,195.5) .. (257,195.5) .. controls (252.86,195.5) and (249.5,192.14) .. (249.5,188) -- cycle ;
		\draw  [fill={rgb, 255:red, 255; green, 255; blue, 255 }  ,fill opacity=1 ] (220,188) .. controls (220,183.86) and (223.36,180.5) .. (227.5,180.5) .. controls (231.64,180.5) and (235,183.86) .. (235,188) .. controls (235,192.14) and (231.64,195.5) .. (227.5,195.5) .. controls (223.36,195.5) and (220,192.14) .. (220,188) -- cycle ;
		\draw  [fill={rgb, 255:red, 255; green, 255; blue, 255 }  ,fill opacity=1 ] (161,114) .. controls (161,109.86) and (164.36,106.5) .. (168.5,106.5) .. controls (172.64,106.5) and (176,109.86) .. (176,114) .. controls (176,118.14) and (172.64,121.5) .. (168.5,121.5) .. controls (164.36,121.5) and (161,118.14) .. (161,114) -- cycle ;
		\draw  [fill={rgb, 255:red, 255; green, 255; blue, 255 }  ,fill opacity=1 ] (437,114) .. controls (437,109.86) and (440.36,106.5) .. (444.5,106.5) .. controls (448.64,106.5) and (452,109.86) .. (452,114) .. controls (452,118.14) and (448.64,121.5) .. (444.5,121.5) .. controls (440.36,121.5) and (437,118.14) .. (437,114) -- cycle ;
		\draw  [fill={rgb, 255:red, 255; green, 255; blue, 255 }  ,fill opacity=1 ] (437,188) .. controls (437,183.86) and (440.36,180.5) .. (444.5,180.5) .. controls (448.64,180.5) and (452,183.86) .. (452,188) .. controls (452,192.14) and (448.64,195.5) .. (444.5,195.5) .. controls (440.36,195.5) and (437,192.14) .. (437,188) -- cycle ;
		\draw  [fill={rgb, 255:red, 255; green, 255; blue, 255 }  ,fill opacity=1 ] (407.5,188) .. controls (407.5,183.86) and (410.86,180.5) .. (415,180.5) .. controls (419.14,180.5) and (422.5,183.86) .. (422.5,188) .. controls (422.5,192.14) and (419.14,195.5) .. (415,195.5) .. controls (410.86,195.5) and (407.5,192.14) .. (407.5,188) -- cycle ;
		\draw  [fill={rgb, 255:red, 255; green, 255; blue, 255 }  ,fill opacity=1 ] (466.5,188) .. controls (466.5,183.86) and (469.86,180.5) .. (474,180.5) .. controls (478.14,180.5) and (481.5,183.86) .. (481.5,188) .. controls (481.5,192.14) and (478.14,195.5) .. (474,195.5) .. controls (469.86,195.5) and (466.5,192.14) .. (466.5,188) -- cycle ;
		\draw  [fill={rgb, 255:red, 255; green, 255; blue, 255 }  ,fill opacity=1 ] (496,188) .. controls (496,183.86) and (499.36,180.5) .. (503.5,180.5) .. controls (507.64,180.5) and (511,183.86) .. (511,188) .. controls (511,192.14) and (507.64,195.5) .. (503.5,195.5) .. controls (499.36,195.5) and (496,192.14) .. (496,188) -- cycle ;
		\draw  [fill={rgb, 255:red, 255; green, 255; blue, 255 }  ,fill opacity=1 ] (525.5,188) .. controls (525.5,183.86) and (528.86,180.5) .. (533,180.5) .. controls (537.14,180.5) and (540.5,183.86) .. (540.5,188) .. controls (540.5,192.14) and (537.14,195.5) .. (533,195.5) .. controls (528.86,195.5) and (525.5,192.14) .. (525.5,188) -- cycle ;
		\draw  [fill={rgb, 255:red, 255; green, 255; blue, 255 }  ,fill opacity=1 ] (555,188) .. controls (555,183.86) and (558.36,180.5) .. (562.5,180.5) .. controls (566.64,180.5) and (570,183.86) .. (570,188) .. controls (570,192.14) and (566.64,195.5) .. (562.5,195.5) .. controls (558.36,195.5) and (555,192.14) .. (555,188) -- cycle ;
		\draw  [fill={rgb, 255:red, 255; green, 255; blue, 255 }  ,fill opacity=1 ] (584.5,188) .. controls (584.5,183.86) and (587.86,180.5) .. (592,180.5) .. controls (596.14,180.5) and (599.5,183.86) .. (599.5,188) .. controls (599.5,192.14) and (596.14,195.5) .. (592,195.5) .. controls (587.86,195.5) and (584.5,192.14) .. (584.5,188) -- cycle ;
		\draw  [fill={rgb, 255:red, 255; green, 255; blue, 255 }  ,fill opacity=1 ] (569.75,114) .. controls (569.75,109.86) and (573.11,106.5) .. (577.25,106.5) .. controls (581.39,106.5) and (584.75,109.86) .. (584.75,114) .. controls (584.75,118.14) and (581.39,121.5) .. (577.25,121.5) .. controls (573.11,121.5) and (569.75,118.14) .. (569.75,114) -- cycle ;
		\draw  [fill={rgb, 255:red, 255; green, 255; blue, 255 }  ,fill opacity=1 ] (510.75,114) .. controls (510.75,109.86) and (514.11,106.5) .. (518.25,106.5) .. controls (522.39,106.5) and (525.75,109.86) .. (525.75,114) .. controls (525.75,118.14) and (522.39,121.5) .. (518.25,121.5) .. controls (514.11,121.5) and (510.75,118.14) .. (510.75,114) -- cycle ;
		\draw   (478,27) .. controls (478,15.95) and (493.67,7) .. (513,7) .. controls (532.33,7) and (548,15.95) .. (548,27) .. controls (548,38.05) and (532.33,47) .. (513,47) .. controls (493.67,47) and (478,38.05) .. (478,27) -- cycle ;
		\draw [line width=3.75]    (291.5,144) -- (360.5,144) ;
		\draw [shift={(367.5,144)}, rotate = 180] [fill={rgb, 255:red, 0; green, 0; blue, 0 }  ][line width=0.08]  [draw opacity=0] (22.33,-10.72) -- (0,0) -- (22.33,10.73) -- (14.83,0) -- cycle    ;
		
		\draw (502,18.4) node [anchor=north west][inner sep=0.75pt]    {$K_{3}$};
		
		\draw (147,96) node [anchor=north west][inner sep=0.75pt]    {$u$};
		\draw (424,96) node [anchor=north west][inner sep=0.75pt]    {$u$};
		\draw (528,96) node [anchor=north west][inner sep=0.75pt]    {$r_1$};
		\draw (584,96) node [anchor=north west][inner sep=0.75pt]    {$r_2$};

	\end{tikzpicture}
	
	\caption{An example of the operation performed in the reduction of~\cite{Emden-WeinertHK98}, with $k=4$. The degree~of $u$ decreased by one after this operation.}
	\label{fig:kcol_degree_decreasing}
	
	\end{figure}

	
	First, recall that by Theorem~\ref{thm:nonkcol_quadratique}, at least $\Omega(n/\log n)$ bits are necessary to certify non-$k$-colorability in graphs of maximum degree~$5k-4$. We will prove that, for every $\delta \in \{k + \lceil\sqrt{k}\rceil, \ldots, 5k-4\}$, there exists a local reduction of local expansion~$O(1)$ from non-$k$-colorability in graphs of maximum degree~$\delta$ to non-$k$-colorability in graphs of maximum degree~$\delta-1$, with a slightly different definition of local reduction (and we will prove that this small change will not affect Theorem~\ref{thm:main theorem local reduction}). By doing so, since $k$ is a constant, by applying iteratively Theorem~\ref{thm:main theorem local reduction}, the $\Omega(n/\log n)$ lower bound for non-$k$-colorability will hold in graphs of maximum degree~$k + \lceil\sqrt{k}\rceil-1$.
	
	Let $\delta \in \{k + \lceil\sqrt{k}\rceil, \ldots, 5k-4\}$,
        and let $\p$ (resp.\ $\p'$) be the $k$-colorability property in graphs of maximum degree~$\delta$ (resp.\ $\delta-1$). Let $G$ be an $n$-vertex graph with vertices having unique identifiers in $\{1, \ldots, n\}$ and maximum degree~$\delta$.
	We define the graph $G' = f_{\p,\p'}(G)$ in the following way. We perform simultaneously, for all the vertices $u$ of~$G$, the operation represented on Figure~\ref{fig:kcol_degree_decreasing}: namely, we add $\lceil \sqrt{k} \rceil$ additional vertices $r_1(u), \ldots, r_{\lceil \sqrt{k} \rceil}(u)$ and a clique $K_{k-1}(u)$ of size $k-1$ complete to $u$ and to all the vertices $r_1(u), \ldots, r_{\lceil \sqrt{k} \rceil}(u)$. Then, for each edge $uv$ of $G$, let $i_u(v)$ (resp.\ $i_v(u)$) denote the index of the identifier of $v$ (resp.\ of $u$) in the list of identifiers of the neighbors of~$u$ (resp.\ of $v$) sorted in increasing order. We put an edge in $f_{\p',\p}(G)$ which has the following endpoints:
	\begin{itemize}
		\item if $i_u(v) \leqslant k$, the first endpoint is $r_{(i_u(v) \pmod {\lceil \sqrt{k} \rceil)} + 1}(u)$, otherwise it is $u$;
		\item if $i_v(u) \leqslant k$, the second endpoint is $r_{(i_v(u) \pmod {\lceil \sqrt{k} \rceil)} + 1}(v)$, otherwise it is $v$.
	\end{itemize}
	
	In other words: for each vertex $u$, we detach the edges of $u$ incident to its $\lceil \sqrt{k} \rceil$ neighbors having smallest identifiers, and attach them evenly to $r_1(u), \ldots, r_{\lceil \sqrt{k} \rceil}(u)$.
	Let us also define $V_u=C_u$ as the set containing $u$, $r_1(u), \ldots, r_{\lceil \sqrt{k} \rceil}(u)$, $K_{k-1}(u)$.
	Finally, we define an identifier for every vertex in $G'$ that indicates to which vertex it corresponds in $G$, if it belongs to a $(k-1)$-clique, or if it is a vertex $r_i(u)$ for some~$i$ and~$u$. All this information can be encoded on $O(\log n)$ bits.
	
	Now, we would like to prove that it defines a local reduction
        of local expansion $O(1)$ and to apply Theorem~\ref{thm:main
          theorem local reduction}. In fact, it is nearly the case:
        all the properties of the definition of local reduction (see
        Section~\ref{sec:reduc}) are satisfied (by a straightforward
        verification) except property~\cref{l3f}. Indeed, the
        neighborhood of $V_u$ does not only depend on $N[u]$, because
        for each $v \in N(u)$, $i_v(u)$ depends on the neighborhood of
        $v$. In fact, it is straightforward to see that the
        neighborhood of $V_u$ depends only on $N^2[u]$. We will now
        show that this is not a significant issue, as soon as the property $\p$ is a property on a class of bounded-degree graphs (which is in particular the case here, since we look at graphs of maximum degree$~\delta$, which is at most $5k-4$, which is itself a constant).
	
	\begin{proposition}
		\label{prop:bounded_degree}
		Let $\p$, $\p'$ be two graph properties, on two graph classes $\C$, $\C'$ respectively, and let $d \geqslant 1$. If $\C$ is a class of bounded-degree graphs, then Theorem~\ref{thm:main theorem local reduction} (and thus also Corollary~\ref{cor:reduction3col}) remains true even if we replace condition \cref{l3f} of the definition of local reduction (see Section~\ref{sec:reduc}) by the following one:
		\begin{enumerate}
				\item[(R3f')] $V_u$ and its neighborhood (resp.\ $C_u$) only depends on $N^d[u]$. In
				other words: if $G'$ is another $n$-vertex graph with unique identifiers in
				$\{1, \ldots, n\}$, and if the subgraph with identifiers formed by the vertices in $N^{d-1}[u]$ and their adjacent edges is the same in $G$ and
				$G'$, then the sets $V_u$ and the subgraphs with identifiers formed by the vertices in $V_u$ and their adjacent edges in
				$f_{\p',\p}(G)$ and $f_{\p',\p}(G')$ are the same (resp.\ the sets $C_u$ in $f_{\p',\p}(G)$
				and $f_{\p',\p}(G')$ are the same).
		\end{enumerate}
	\end{proposition}
	
	\begin{proof}
		Since $\C$ is a class of bounded-degree graphs, the set $N^{d-1}[u]$ has constant size for every vertex $u$. In the proof of Theorem~\ref{thm:main theorem local reduction}, the prover adds the following information in the certificate of each vertex~$u$: for every $v \in N^{d-1}[u]$, it writes the identifier of $v$, its distance to~$u$, and all the edges adjacent to $v$. This information has size $O(\log n)$ and its correctness can be checked by the vertices at the beginning of the verification procedure. Thus, we can assume that each vertex $u$ knows all the edges adjacent to the vertices in $N^{d-1}[u]$, and the rest of the proof of Theorem~\ref{thm:main theorem local reduction} is the same.
	\end{proof}
	
	Finally, the result of Theorem~\ref{thm:nonkcol_lineaire}
        follows from Proposition~\ref{prop:bounded_degree}: indeed,
        for every $\delta \in \{k + \lceil \sqrt{k} \rceil, \ldots,
        5k-4\}$, we obtain a local reduction (in the sense of
        Proposition~\ref{prop:bounded_degree}) of constant local
        expansion, from non-$k$-colorability in graphs of maximum
        degree~$\delta$ to non-$k$-colorability in graphs of maximum
        degree~$\delta-1$. Thus, the $\Omega(n/\log
            n)$ lower bound holds for graphs of maximum degree
            $k + \lceil \sqrt{k} \rceil - 1$. \qed

            \medskip

            As a consequence of Theorem~\ref{thm:nonkcol_lineaire}, we
            obtain the following improved version of the second item
            of Corollary~\ref{cor:reduction3col},
            which implies in particular that the maximum degree in the
            statement of Theorem \ref{thm:domatic} can be decreased to
            $k^2 + k \lceil \sqrt{k+1} \rceil$.

            	\begin{corollary}
		\label{cor:reduction3coldeg}
		
		Let $\p$ be a graph property (on some graph class $\C$) and let $k \geqslant 3$. If there exists a local reduction from $k$-colorability in graphs of
			maximum degree~$k + \lceil \sqrt{k} \rceil - 1$ to $\p$ with local expansion
                        $O(n^\delta)$ for some $0 \leqslant \delta <
                        1$, and global expansion $O(n^\gamma)$, then
                        $\np$ has local complexity 
			$\Omega\big(n^{(1-\delta)/\gamma}/\log
                          n\big)$. In particular, since $\gamma
                        \leqslant \delta+1$, $\np$ has local complexity $\Omega\big(n^{\frac{2}{\delta+1}-1}/\log n\big)$.
	\end{corollary}

        We now show that the bound        $k + \lceil \sqrt{k} \rceil
        - 1$ on the maximum degree~in
        Theorem~\ref{thm:nonkcol_lineaire} is close to best possible for
        sufficiently large (but constant) $k$.

        \begin{theorem}\label{thm:ubnonkcol}
          Let $k$ be a sufficiently large integer. In graphs of maximum degree \mbox{$k + \lceil\sqrt{k}\rceil - 3$}, non-$k$-colorability has local
                complexity $O(\log n)$.
        \end{theorem}

        \begin{proof}
It was proved by Molloy and Reed \cite[Theorem 5]{MR14} that for sufficiently large~$k$,  a graph $G$ of maximum degree
$k + \lceil\sqrt{k}\rceil - 3$ is not $k$-colorable if and only if there exists
a vertex $v\in V(G)$ such that the subgraph of $G$ induced by $N[v]$
is not $k$-colorable. We use this to produce a proof labeling scheme
for non-$k$-colorability as follows. The prover identifies a vertex
$v$ such that the subgraph of $G$ induced by $N[v]$
is not $k$-colorable, and each vertex $u$ in
the graph receives as a certificate  the list $L(u)$ of all neighbors of $u$ (as the graph has
bounded maximum degree, this list of identifiers takes $O(\log n)$
bits). Using the lists $\{L(u):uv\in E(G) \}$, the vertex $v$ knows the subgraph
of $G$ induced by $N[v]$, and can thus check that this subgraph is indeed not
$k$-colorable.  
        \end{proof}

        \section{Final remarks}
	\label{sec:remark}

Since the polynomial lower bounds obtained in this paper are all based on \textsf{NP}-hardness
        reductions, all the applications presented so far relate to
        \textsf{coNP}-hard problems. It might be tempting to
        conjecture that all such problems have at least polynomial local
        complexity. We now observe that this is not the case, by
        giving a \textsf{coNP}-hard problem with logarithmic complexity.

        \smallskip
        
	Indeed, consider the property $\p$ defined as the set of
        connected graphs $G$ for which there exists a vertex $u \in
        V(G)$ and an integer $m \geqslant 2$ such that $u$ has degree
        $3m$ and the neighbors $v_1, \ldots, v_{3m}$ of~$u$ can be
        partitioned into $m$~triples $T_1, \ldots, T_m$ such that,
        for all $i \in \{1, \ldots, m\}$, $\Sigma_{v \in T_i} \deg(v)
        = \frac1m \Sigma_{j=1}^{3m}\deg(v_j)$.
	The problem of determining whether a graph satisfies~$\p$ is \textsf{NP}-complete~: indeed, there is a simple polynomial reduction from 3-partition where the inputs are encoded in unary (recall that 3-partition is strongly \textsf{NP}-complete), which is the following: on an instance $\mathcal{I}:=\{a_1, \ldots, a_{3m}\}$ of 3-partition (where $m \geqslant 2$), we create the graph $G_{\mathcal{I}}$ consisting of a central vertex $u$ that has $3m$ neighbors $v_1, \ldots, v_{3m}$, and for every $i \in \{1, \ldots, m\}$, $v_i$ has $a_i-1$ additional neighbors of degree~1. For every $i \in \{1, \ldots, n\}$, we have $\deg(v_i)=a_i$. The graph $G_{\mathcal{I}}$ satisfies $\p$ if and only if we can partition the set of neighbors of~$u$ into triples of vertices that have the same sum of degrees, so if and only if $\mathcal{I}$ is a valid instance of 3-partition.
	However, the local complexity of $\np$ is $O(\log n)$: indeed, for each vertex~$u$, the prover can just write the degree of~$u$ in its certificate. Then, every vertex can simply check that the degree written in its certificate is correct and that its set of neighbors cannot be partitioned into such triples.
	
        \medskip

        We note that a priori, nothing prevents our reduction
        framework to be applied to problems that are not \textsf{coNP}-hard. Instead of starting from non-3-colorability, it
         suffices to start with a problem that has high local
        complexity, but which is easy globally.
        
        Let us give  examples of reductions that apply to problems
        that are not \textsf{coNP}-hard. For some graphs~$H$, it is
        known that certifying that a graph~$G$ does not contain $H$ as
        an induced subgraph (we will call this property
        \emph{$H$-freeness}) requires certificates of
        size~$\Omega(n)$. For instance, it is the case when $H$ is a
        clique of size at
        least~4~\cite[Theorem~15]{DruckerKO13}\footnote{The model
          considered in~\cite{DruckerKO13} is a bit different from
          local certification but their proof can be easily adapted to
          work in this model.}, or when $H$ is a path of length at
        least~7~\cite[Theorem~7]{BCFPZ}\footnote{Authors
          in~\cite{BCFPZ} actually prove the~$\Omega(n)$ lower bound
          for $P_7$-free graphs but their proof can be easily adapted to
          paths of length $\ell$ for every $\ell \geqslant 7$, by
          adding a pending path of length $\ell - 7$ to the vertex
          $v_0$ in their construction.}. However, for every fixed graph~$H$, determining if an $n$-vertex graph $G$ contains~$H$ as an induced subgraph can be done globally in time~$O(n^{|V(H)|})$, so in polynomial time.
          Here, we present how,
        starting from a lower bound on the local complexity of  $H$-freeness, we can prove a
        lower bound for $f(H)$-freeness, for some functions~$f$ on
        graphs. For instance, let $k \geqslant 2$ and let $f$ be the
        following function: for every graph~$G$, $f(G)$ is the graph
        obtained by replacing each vertex of~$G$ by a clique of
        size~$k$ and each edge of~$G$ by a complete bipartite graph
        $K_{k,k}$ (this operation is sometimes called the
        lexicographic product of $G$ with $K_k$). It can be easily checked that, 
	for every graphs~$G$ and~$H$, $G$ is $H$-free if and only if
        $f(G)$ is $f(H)$-free. Moreover, this transformation defines a
        local reduction of constant local expansion, where for each
        vertex $u \in V(G)$, the sets $V_u$ and $C_u$ are both equal
        to the clique of size $k$ replacing~$u$ in~$V(f(G))$. Thus, if
        $H$-freeness requires certificates of size~$\Omega(n)$,
        $f(H)$-freeness requires certificates of size~$\Omega(n)$ by
        Theorem~\ref{thm:main theorem local reduction}. For instance,
        if $H$ is a path of length~7 and if $k=2$, we obtain
        an~$\Omega(n)$ lower bound for the local complexity of
        $H'$-freeness where $H'$ is the graph depicted in Figure
        \ref{fig:sppk}.

        \medskip
	\begin{figure}[h!]
		\centering
		
		\begin{tikzpicture}[x=0.75pt,y=0.75pt,yscale=-1,xscale=1]
			
			\draw  [fill={rgb, 255:red, 0; green, 0; blue, 0 }  ,fill opacity=1 ] (141,107.5) .. controls (141,105.01) and (143.01,103) .. (145.5,103) .. controls (147.99,103) and (150,105.01) .. (150,107.5) .. controls (150,109.99) and (147.99,112) .. (145.5,112) .. controls (143.01,112) and (141,109.99) .. (141,107.5) -- cycle ;
			\draw    (145.5,107.5) -- (189.5,107.5) ;
			\draw    (145.5,107.5) -- (145.5,142) ;
			\draw  [fill={rgb, 255:red, 0; green, 0; blue, 0 }  ,fill opacity=1 ] (141,142) .. controls (141,139.51) and (143.01,137.5) .. (145.5,137.5) .. controls (147.99,137.5) and (150,139.51) .. (150,142) .. controls (150,144.49) and (147.99,146.5) .. (145.5,146.5) .. controls (143.01,146.5) and (141,144.49) .. (141,142) -- cycle ;
			\draw    (189.5,107.5) -- (189.5,142) ;
			\draw    (145.5,107.5) -- (189.5,142) ;
			\draw    (189.5,107.5) -- (233.5,142) ;
			\draw    (233.5,107.5) -- (233.5,142) ;
			\draw    (145.5,142) -- (189.5,142) ;
			\draw    (189.5,107.5) -- (233.5,107.5) ;
			\draw    (189.5,142) -- (233.5,142) ;
			\draw    (145.5,142) -- (189.5,107.5) ;
			\draw    (189.5,142) -- (233.5,107.5) ;
			\draw  [fill={rgb, 255:red, 0; green, 0; blue, 0 }  ,fill opacity=1 ] (185,107.5) .. controls (185,105.01) and (187.01,103) .. (189.5,103) .. controls (191.99,103) and (194,105.01) .. (194,107.5) .. controls (194,109.99) and (191.99,112) .. (189.5,112) .. controls (187.01,112) and (185,109.99) .. (185,107.5) -- cycle ;
			\draw  [fill={rgb, 255:red, 0; green, 0; blue, 0 }  ,fill opacity=1 ] (185,142) .. controls (185,139.51) and (187.01,137.5) .. (189.5,137.5) .. controls (191.99,137.5) and (194,139.51) .. (194,142) .. controls (194,144.49) and (191.99,146.5) .. (189.5,146.5) .. controls (187.01,146.5) and (185,144.49) .. (185,142) -- cycle ;
			\draw  [fill={rgb, 255:red, 0; green, 0; blue, 0 }  ,fill opacity=1 ] (229,107.5) .. controls (229,105.01) and (231.01,103) .. (233.5,103) .. controls (235.99,103) and (238,105.01) .. (238,107.5) .. controls (238,109.99) and (235.99,112) .. (233.5,112) .. controls (231.01,112) and (229,109.99) .. (229,107.5) -- cycle ;
			\draw  [fill={rgb, 255:red, 0; green, 0; blue, 0 }  ,fill opacity=1 ] (229,142) .. controls (229,139.51) and (231.01,137.5) .. (233.5,137.5) .. controls (235.99,137.5) and (238,139.51) .. (238,142) .. controls (238,144.49) and (235.99,146.5) .. (233.5,146.5) .. controls (231.01,146.5) and (229,144.49) .. (229,142) -- cycle ;
			\draw    (233.5,107.5) -- (277.5,107.5) ;
			\draw    (233.5,142) -- (277.5,142) ;
			\draw  [fill={rgb, 255:red, 0; green, 0; blue, 0 }  ,fill opacity=1 ] (273,107.5) .. controls (273,105.01) and (275.01,103) .. (277.5,103) .. controls (279.99,103) and (282,105.01) .. (282,107.5) .. controls (282,109.99) and (279.99,112) .. (277.5,112) .. controls (275.01,112) and (273,109.99) .. (273,107.5) -- cycle ;
			\draw  [fill={rgb, 255:red, 0; green, 0; blue, 0 }  ,fill opacity=1 ] (273,142) .. controls (273,139.51) and (275.01,137.5) .. (277.5,137.5) .. controls (279.99,137.5) and (282,139.51) .. (282,142) .. controls (282,144.49) and (279.99,146.5) .. (277.5,146.5) .. controls (275.01,146.5) and (273,144.49) .. (273,142) -- cycle ;
			\draw    (277.5,107.5) -- (277.5,142) ;
			\draw    (233.5,107.5) -- (277.5,142) ;
			\draw    (277.5,107.5) -- (233.5,142) ;
			\draw    (277.5,107.5) -- (321.5,107.5) ;
			\draw    (277.5,142) -- (321.5,142) ;
			\draw    (321.5,107.5) -- (321.5,142) ;
			\draw  [fill={rgb, 255:red, 0; green, 0; blue, 0 }  ,fill opacity=1 ] (317,107.5) .. controls (317,105.01) and (319.01,103) .. (321.5,103) .. controls (323.99,103) and (326,105.01) .. (326,107.5) .. controls (326,109.99) and (323.99,112) .. (321.5,112) .. controls (319.01,112) and (317,109.99) .. (317,107.5) -- cycle ;
			\draw  [fill={rgb, 255:red, 0; green, 0; blue, 0 }  ,fill opacity=1 ] (317,142) .. controls (317,139.51) and (319.01,137.5) .. (321.5,137.5) .. controls (323.99,137.5) and (326,139.51) .. (326,142) .. controls (326,144.49) and (323.99,146.5) .. (321.5,146.5) .. controls (319.01,146.5) and (317,144.49) .. (317,142) -- cycle ;
			\draw    (321.5,107.5) -- (277.5,142) ;
			\draw    (365.5,107.5) -- (321.5,142) ;
			\draw    (277.5,107.5) -- (321.5,142) ;
			\draw    (321.5,107.5) -- (365.5,142) ;
			\draw    (321.5,107.5) -- (365.5,107.5) ;
			\draw    (321.5,142) -- (365.5,142) ;
			\draw    (365.5,107.5) -- (365.5,142) ;
			\draw    (365.5,107.5) -- (409.5,142) ;
			\draw    (409.5,107.5) -- (365.5,142) ;
			\draw    (409.5,107.5) -- (409.5,142) ;
			\draw    (365.5,107.5) -- (409.5,107.5) ;
			\draw    (365.5,142) -- (409.5,142) ;
			\draw  [fill={rgb, 255:red, 0; green, 0; blue, 0 }  ,fill opacity=1 ] (361,107.5) .. controls (361,105.01) and (363.01,103) .. (365.5,103) .. controls (367.99,103) and (370,105.01) .. (370,107.5) .. controls (370,109.99) and (367.99,112) .. (365.5,112) .. controls (363.01,112) and (361,109.99) .. (361,107.5) -- cycle ;
			\draw  [fill={rgb, 255:red, 0; green, 0; blue, 0 }  ,fill opacity=1 ] (361,142) .. controls (361,139.51) and (363.01,137.5) .. (365.5,137.5) .. controls (367.99,137.5) and (370,139.51) .. (370,142) .. controls (370,144.49) and (367.99,146.5) .. (365.5,146.5) .. controls (363.01,146.5) and (361,144.49) .. (361,142) -- cycle ;
			\draw  [fill={rgb, 255:red, 0; green, 0; blue, 0 }  ,fill opacity=1 ] (405,107.5) .. controls (405,105.01) and (407.01,103) .. (409.5,103) .. controls (411.99,103) and (414,105.01) .. (414,107.5) .. controls (414,109.99) and (411.99,112) .. (409.5,112) .. controls (407.01,112) and (405,109.99) .. (405,107.5) -- cycle ;
			\draw  [fill={rgb, 255:red, 0; green, 0; blue, 0 }  ,fill opacity=1 ] (405,142) .. controls (405,139.51) and (407.01,137.5) .. (409.5,137.5) .. controls (411.99,137.5) and (414,139.51) .. (414,142) .. controls (414,144.49) and (411.99,146.5) .. (409.5,146.5) .. controls (407.01,146.5) and (405,144.49) .. (405,142) -- cycle ;
			\end{tikzpicture}
		\caption{The strong product of a path and a $K_2$.}\label{fig:sppk}
	\end{figure}
	
	If $f$ consists in adding a pending vertex to each vertex, it can easily be checked that $G$ is $H$-free if and only $f(G)$ is $f(H)$-free, 
	and this defines again a local reduction of constant local
        expansion. So for instance, if $H$ is a path of length~7, we
        obtain an $\Omega(n)$ lower bound for the local complexity of
        $H'$-freeness where $H'$ is the graph depicted in Figure \ref{fig:comb}.
	
	\begin{figure}[h!]
		\centering
		
		\begin{tikzpicture}[x=0.75pt,y=0.75pt,yscale=-1,xscale=1]
			
			\draw  [fill={rgb, 255:red, 0; green, 0; blue, 0 }  ,fill opacity=1 ] (141,107.5) .. controls (141,105.01) and (143.01,103) .. (145.5,103) .. controls (147.99,103) and (150,105.01) .. (150,107.5) .. controls (150,109.99) and (147.99,112) .. (145.5,112) .. controls (143.01,112) and (141,109.99) .. (141,107.5) -- cycle ;
			\draw    (145.5,107.5) -- (145.5,142) ;
			\draw  [fill={rgb, 255:red, 0; green, 0; blue, 0 }  ,fill opacity=1 ] (141,142) .. controls (141,139.51) and (143.01,137.5) .. (145.5,137.5) .. controls (147.99,137.5) and (150,139.51) .. (150,142) .. controls (150,144.49) and (147.99,146.5) .. (145.5,146.5) .. controls (143.01,146.5) and (141,144.49) .. (141,142) -- cycle ;
			\draw    (189.5,107.5) -- (189.5,142) ;
			\draw    (233.5,107.5) -- (233.5,142) ;
			\draw    (145.5,142) -- (189.5,142) ;
			\draw    (189.5,142) -- (233.5,142) ;
			\draw  [fill={rgb, 255:red, 0; green, 0; blue, 0 }  ,fill opacity=1 ] (185,107.5) .. controls (185,105.01) and (187.01,103) .. (189.5,103) .. controls (191.99,103) and (194,105.01) .. (194,107.5) .. controls (194,109.99) and (191.99,112) .. (189.5,112) .. controls (187.01,112) and (185,109.99) .. (185,107.5) -- cycle ;
			\draw  [fill={rgb, 255:red, 0; green, 0; blue, 0 }  ,fill opacity=1 ] (185,142) .. controls (185,139.51) and (187.01,137.5) .. (189.5,137.5) .. controls (191.99,137.5) and (194,139.51) .. (194,142) .. controls (194,144.49) and (191.99,146.5) .. (189.5,146.5) .. controls (187.01,146.5) and (185,144.49) .. (185,142) -- cycle ;
			\draw  [fill={rgb, 255:red, 0; green, 0; blue, 0 }  ,fill opacity=1 ] (229,107.5) .. controls (229,105.01) and (231.01,103) .. (233.5,103) .. controls (235.99,103) and (238,105.01) .. (238,107.5) .. controls (238,109.99) and (235.99,112) .. (233.5,112) .. controls (231.01,112) and (229,109.99) .. (229,107.5) -- cycle ;
			\draw  [fill={rgb, 255:red, 0; green, 0; blue, 0 }  ,fill opacity=1 ] (229,142) .. controls (229,139.51) and (231.01,137.5) .. (233.5,137.5) .. controls (235.99,137.5) and (238,139.51) .. (238,142) .. controls (238,144.49) and (235.99,146.5) .. (233.5,146.5) .. controls (231.01,146.5) and (229,144.49) .. (229,142) -- cycle ;
			\draw    (233.5,142) -- (277.5,142) ;
			\draw  [fill={rgb, 255:red, 0; green, 0; blue, 0 }  ,fill opacity=1 ] (273,107.5) .. controls (273,105.01) and (275.01,103) .. (277.5,103) .. controls (279.99,103) and (282,105.01) .. (282,107.5) .. controls (282,109.99) and (279.99,112) .. (277.5,112) .. controls (275.01,112) and (273,109.99) .. (273,107.5) -- cycle ;
			\draw  [fill={rgb, 255:red, 0; green, 0; blue, 0 }  ,fill opacity=1 ] (273,142) .. controls (273,139.51) and (275.01,137.5) .. (277.5,137.5) .. controls (279.99,137.5) and (282,139.51) .. (282,142) .. controls (282,144.49) and (279.99,146.5) .. (277.5,146.5) .. controls (275.01,146.5) and (273,144.49) .. (273,142) -- cycle ;
			\draw    (277.5,107.5) -- (277.5,142) ;
			\draw    (277.5,142) -- (321.5,142) ;
			\draw    (321.5,107.5) -- (321.5,142) ;
			\draw  [fill={rgb, 255:red, 0; green, 0; blue, 0 }  ,fill opacity=1 ] (317,107.5) .. controls (317,105.01) and (319.01,103) .. (321.5,103) .. controls (323.99,103) and (326,105.01) .. (326,107.5) .. controls (326,109.99) and (323.99,112) .. (321.5,112) .. controls (319.01,112) and (317,109.99) .. (317,107.5) -- cycle ;
			\draw  [fill={rgb, 255:red, 0; green, 0; blue, 0 }  ,fill opacity=1 ] (317,142) .. controls (317,139.51) and (319.01,137.5) .. (321.5,137.5) .. controls (323.99,137.5) and (326,139.51) .. (326,142) .. controls (326,144.49) and (323.99,146.5) .. (321.5,146.5) .. controls (319.01,146.5) and (317,144.49) .. (317,142) -- cycle ;
			\draw    (321.5,142) -- (365.5,142) ;
			\draw    (365.5,107.5) -- (365.5,142) ;
			\draw    (409.5,107.5) -- (409.5,142) ;
			\draw    (365.5,142) -- (409.5,142) ;
			\draw  [fill={rgb, 255:red, 0; green, 0; blue, 0 }  ,fill opacity=1 ] (361,107.5) .. controls (361,105.01) and (363.01,103) .. (365.5,103) .. controls (367.99,103) and (370,105.01) .. (370,107.5) .. controls (370,109.99) and (367.99,112) .. (365.5,112) .. controls (363.01,112) and (361,109.99) .. (361,107.5) -- cycle ;
			\draw  [fill={rgb, 255:red, 0; green, 0; blue, 0 }  ,fill opacity=1 ] (361,142) .. controls (361,139.51) and (363.01,137.5) .. (365.5,137.5) .. controls (367.99,137.5) and (370,139.51) .. (370,142) .. controls (370,144.49) and (367.99,146.5) .. (365.5,146.5) .. controls (363.01,146.5) and (361,144.49) .. (361,142) -- cycle ;
			\draw  [fill={rgb, 255:red, 0; green, 0; blue, 0 }  ,fill opacity=1 ] (405,107.5) .. controls (405,105.01) and (407.01,103) .. (409.5,103) .. controls (411.99,103) and (414,105.01) .. (414,107.5) .. controls (414,109.99) and (411.99,112) .. (409.5,112) .. controls (407.01,112) and (405,109.99) .. (405,107.5) -- cycle ;
			\draw  [fill={rgb, 255:red, 0; green, 0; blue, 0 }  ,fill opacity=1 ] (405,142) .. controls (405,139.51) and (407.01,137.5) .. (409.5,137.5) .. controls (411.99,137.5) and (414,139.51) .. (414,142) .. controls (414,144.49) and (411.99,146.5) .. (409.5,146.5) .. controls (407.01,146.5) and (405,144.49) .. (405,142) -- cycle ;

                      \end{tikzpicture}
                      \caption{A comb of length 7.}\label{fig:comb}
	\end{figure}
	
	If~$f$ consists in adding a pending path of length~$t$ to
        every vertex, if $i, j$ are two integers such that $i + j
        \geqslant 6$, it can easily be checked that~$G$ contains no
        path of length $i+j+1$ if and only if $f(G)$ contains no claw
        with a central vertex and three branches of length $(t, t+i,
        t+j)$. This transformation defines again a local reduction of
        constant local expansion. For instance, with $t=3$, $i=2$,
        $j=4$, we obtain a $\Omega(n)$ lower bound on the local
        complexity of $H'$-freeness where~$H'$ is the claw with
        branches of size $(3, 5, 7)$, depicted in
        Figure \ref{fig:claw}.
	
	\begin{figure}[h!]
		\centering
		
		\begin{tikzpicture}[x=0.75pt,y=0.75pt,yscale=-1,xscale=1]
			
			\draw  [fill={rgb, 255:red, 0; green, 0; blue, 0 }  ,fill opacity=1 ] (137,185) .. controls (137,182.51) and (139.01,180.5) .. (141.5,180.5) .. controls (143.99,180.5) and (146,182.51) .. (146,185) .. controls (146,187.49) and (143.99,189.5) .. (141.5,189.5) .. controls (139.01,189.5) and (137,187.49) .. (137,185) -- cycle ;
			\draw  [fill={rgb, 255:red, 0; green, 0; blue, 0 }  ,fill opacity=1 ] (165,185) .. controls (165,182.51) and (167.01,180.5) .. (169.5,180.5) .. controls (171.99,180.5) and (174,182.51) .. (174,185) .. controls (174,187.49) and (171.99,189.5) .. (169.5,189.5) .. controls (167.01,189.5) and (165,187.49) .. (165,185) -- cycle ;
			\draw  [fill={rgb, 255:red, 0; green, 0; blue, 0 }  ,fill opacity=1 ] (193,185) .. controls (193,182.51) and (195.01,180.5) .. (197.5,180.5) .. controls (199.99,180.5) and (202,182.51) .. (202,185) .. controls (202,187.49) and (199.99,189.5) .. (197.5,189.5) .. controls (195.01,189.5) and (193,187.49) .. (193,185) -- cycle ;
			\draw  [fill={rgb, 255:red, 0; green, 0; blue, 0 }  ,fill opacity=1 ] (221,185) .. controls (221,182.51) and (223.01,180.5) .. (225.5,180.5) .. controls (227.99,180.5) and (230,182.51) .. (230,185) .. controls (230,187.49) and (227.99,189.5) .. (225.5,189.5) .. controls (223.01,189.5) and (221,187.49) .. (221,185) -- cycle ;
			\draw  [fill={rgb, 255:red, 0; green, 0; blue, 0 }  ,fill opacity=1 ] (249,185) .. controls (249,182.51) and (251.01,180.5) .. (253.5,180.5) .. controls (255.99,180.5) and (258,182.51) .. (258,185) .. controls (258,187.49) and (255.99,189.5) .. (253.5,189.5) .. controls (251.01,189.5) and (249,187.49) .. (249,185) -- cycle ;
			\draw  [fill={rgb, 255:red, 0; green, 0; blue, 0 }  ,fill opacity=1 ] (277,185) .. controls (277,182.51) and (279.01,180.5) .. (281.5,180.5) .. controls (283.99,180.5) and (286,182.51) .. (286,185) .. controls (286,187.49) and (283.99,189.5) .. (281.5,189.5) .. controls (279.01,189.5) and (277,187.49) .. (277,185) -- cycle ;
			\draw  [fill={rgb, 255:red, 0; green, 0; blue, 0 }  ,fill opacity=1 ] (305,185) .. controls (305,182.51) and (307.01,180.5) .. (309.5,180.5) .. controls (311.99,180.5) and (314,182.51) .. (314,185) .. controls (314,187.49) and (311.99,189.5) .. (309.5,189.5) .. controls (307.01,189.5) and (305,187.49) .. (305,185) -- cycle ;
			\draw  [fill={rgb, 255:red, 0; green, 0; blue, 0 }  ,fill opacity=1 ] (389,185) .. controls (389,182.51) and (391.01,180.5) .. (393.5,180.5) .. controls (395.99,180.5) and (398,182.51) .. (398,185) .. controls (398,187.49) and (395.99,189.5) .. (393.5,189.5) .. controls (391.01,189.5) and (389,187.49) .. (389,185) -- cycle ;
			\draw  [fill={rgb, 255:red, 0; green, 0; blue, 0 }  ,fill opacity=1 ] (361,185) .. controls (361,182.51) and (363.01,180.5) .. (365.5,180.5) .. controls (367.99,180.5) and (370,182.51) .. (370,185) .. controls (370,187.49) and (367.99,189.5) .. (365.5,189.5) .. controls (363.01,189.5) and (361,187.49) .. (361,185) -- cycle ;
			\draw  [fill={rgb, 255:red, 0; green, 0; blue, 0 }  ,fill opacity=1 ] (333,185) .. controls (333,182.51) and (335.01,180.5) .. (337.5,180.5) .. controls (339.99,180.5) and (342,182.51) .. (342,185) .. controls (342,187.49) and (339.99,189.5) .. (337.5,189.5) .. controls (335.01,189.5) and (333,187.49) .. (333,185) -- cycle ;
			\draw    (85.5,185) -- (113.5,185) ;
			\draw  [fill={rgb, 255:red, 0; green, 0; blue, 0 }  ,fill opacity=1 ] (417,185) .. controls (417,182.51) and (419.01,180.5) .. (421.5,180.5) .. controls (423.99,180.5) and (426,182.51) .. (426,185) .. controls (426,187.49) and (423.99,189.5) .. (421.5,189.5) .. controls (419.01,189.5) and (417,187.49) .. (417,185) -- cycle ;
			\draw  [fill={rgb, 255:red, 0; green, 0; blue, 0 }  ,fill opacity=1 ] (109,185) .. controls (109,182.51) and (111.01,180.5) .. (113.5,180.5) .. controls (115.99,180.5) and (118,182.51) .. (118,185) .. controls (118,187.49) and (115.99,189.5) .. (113.5,189.5) .. controls (111.01,189.5) and (109,187.49) .. (109,185) -- cycle ;
			\draw  [fill={rgb, 255:red, 0; green, 0; blue, 0 }  ,fill opacity=1 ] (81,185) .. controls (81,182.51) and (83.01,180.5) .. (85.5,180.5) .. controls (87.99,180.5) and (90,182.51) .. (90,185) .. controls (90,187.49) and (87.99,189.5) .. (85.5,189.5) .. controls (83.01,189.5) and (81,187.49) .. (81,185) -- cycle ;
			\draw  [fill={rgb, 255:red, 0; green, 0; blue, 0 }  ,fill opacity=1 ] (221,161) .. controls (221,158.51) and (223.01,156.5) .. (225.5,156.5) .. controls (227.99,156.5) and (230,158.51) .. (230,161) .. controls (230,163.49) and (227.99,165.5) .. (225.5,165.5) .. controls (223.01,165.5) and (221,163.49) .. (221,161) -- cycle ;
			\draw  [fill={rgb, 255:red, 0; green, 0; blue, 0 }  ,fill opacity=1 ] (221,137) .. controls (221,134.51) and (223.01,132.5) .. (225.5,132.5) .. controls (227.99,132.5) and (230,134.51) .. (230,137) .. controls (230,139.49) and (227.99,141.5) .. (225.5,141.5) .. controls (223.01,141.5) and (221,139.49) .. (221,137) -- cycle ;
			\draw  [fill={rgb, 255:red, 0; green, 0; blue, 0 }  ,fill opacity=1 ] (221,113) .. controls (221,110.51) and (223.01,108.5) .. (225.5,108.5) .. controls (227.99,108.5) and (230,110.51) .. (230,113) .. controls (230,115.49) and (227.99,117.5) .. (225.5,117.5) .. controls (223.01,117.5) and (221,115.49) .. (221,113) -- cycle ;
			\draw    (113.5,185) -- (141.5,185) ;
			\draw    (141.5,185) -- (169.5,185) ;
			\draw    (169.5,185) -- (197.5,185) ;
			\draw    (197.5,185) -- (225.5,185) ;
			\draw    (225.5,185) -- (253.5,185) ;
			\draw    (253.5,185) -- (281.5,185) ;
			\draw    (281.5,185) -- (309.5,185) ;
			\draw    (309.5,185) -- (337.5,185) ;
			\draw    (337.5,185) -- (365.5,185) ;
			\draw    (365.5,185) -- (393.5,185) ;
			\draw    (393.5,185) -- (421.5,185) ;
			\draw    (225.5,185) -- (225.5,161) ;
			\draw    (225.5,161) -- (225.5,137) ;
			\draw    (225.5,137) -- (225.5,113) ;

		\end{tikzpicture}
		\caption{The claw with branches of size $(3, 5, 7)$.}\label{fig:claw}
              \end{figure}

	Finally, let us mention some problems for which we were not
        able to apply our reduction framework. Let $G$ be an
        $n$-vertex graph, and let $S_n$ denote the set of all
        bijections $V(G) \to \{1, \ldots, n\}$. The \emph{bandwidth}
        of $G$, denoted by $\beta(G)$, is defined as follows:
        $\beta(G) := \min_{\pi \in S_n} \max_{\{u,v\} \in E(G)}
        |\pi(u)-\pi(v)|$. The \textsc{Bandwidth} problem takes as
        input a graph~$G$ and an integer~$k$, and outputs whether~$G$
        has bandwidth at most~$k$. We say that $G$ is a \emph{disk
          graph} (resp.\ a \emph{unit-disk graph}) if there exists a
        family of  disks (resp. disks of diameter 1) in the plane such that $G$ is isomorphic to the intersection graph of these disks. The \textsc{Disk Graph} (resp.\ \textsc{Unit-Disk Graph}) problem takes as input a graph~$G$, and outputs whether~$G$ is a disk graph (resp.\ a unit-disk graph).
	
	It is known that the \textsc{Bandwidth}, \textsc{Disk Graph}
        and~\textsc{Unit-Disk Graph} problems are all
        \textsf{NP}-complete~\cite{Papadimitriou76, HK01, BK98}. From
        the perspective of the framework
        introduced in this paper, it is natural to ask, for each of
        these problems, if there is a polynomial lower bound for the
        complementary property (where, for the \textsc{Bandwidth}
        problem, the integer~$k$ is given as an input to the
        vertices). Unfortunately, we were not able to apply our local
        reduction framework and obtain polynomial lower bounds, either
        because existing \textsf{NP}-completeness reductions are not
        sufficiently local, or because they create too many vertices.
	Note also that, for each property to which we applied our reduction framework (summarized in Table~\ref{fig:sum}), the complementary property, which is \textsf{NP}-complete, is also easy to certify (in the sense that certificates of constant or logarithmic size are sufficient). Surprisingly, here there is no straightforward logarithmic certification scheme for the \textsc{Bandwidth}, \textsc{Disk Graph} and~\textsc{Unit Disk Graph} properties, which suggests that these properties may have a different behavior in the perspective of local certification. In fact, it is shown in~\cite{DELMR24} that the local complexity of~\textsc{Unit Disk Graph} is $\Omega(n^{1-\delta})$ for any~$\delta > 0$. For the \textsc{Bandwidth} problem, a natural way to certify it would be to give to every vertex~$u$ the integer~$\pi(u)$ in its certificate. However, to certify that the function $\pi$ written in the certificates is injective, no upper bound better than $O(n)$ is known~\cite{BousquetEFZ24}.
	The previous observations thus raise the following problems.
	
	\begin{problem}
		What is the local complexity of the complementary properties of \textsc{Disk Graph} and \textsc{Unit-Disk Graph} ?
	\end{problem}
	
	\begin{problem}
		What is the local complexity of the \textsc{Bandwidth} property and of its complement ?
	\end{problem}

	\vspace{0.5cm}

              \subsection*{Acknowledgements}
We thank the reviewers of the conference version of the paper for
their helpful comments and suggestions.
              
	\bibliographystyle{plain}
	\bibliography{biblio}

\begin{thebibliography}{10}

\bibitem{BCFPZ}
Nicolas Bousquet, Linda Cook, Laurent Feuilloley, Th{\'{e}}o Pierron, and
  S{\'{e}}bastien Zeitoun.
\newblock Local certification of forbidden subgraphs.
\newblock {\em CoRR}, abs/2402.12148, 2024.

\bibitem{BousquetEFZ24}
Nicolas Bousquet, Louis Esperet, Laurent Feuilloley, and S{\'{e}}bastien
  Zeitoun.
\newblock Renaming in distributed certification.
\newblock {\em CoRR}, abs/2409.15404, 2024.

\bibitem{BK98}
Heinz Breu and David~G. Kirkpatrick.
\newblock Unit disk graph recognition is np-hard.
\newblock {\em Comput. Geom.}, 9(1-2):3--24, 1998.

\bibitem{DELMR24}
Oscar Defrain, Louis Esperet, Aur{\'{e}}lie Lagoutte, Pat Morin, and
  Jean{-}Florent Raymond.
\newblock Local certification of geometric graph classes.
\newblock In Rastislav Kr{\'{a}}lovic and Anton{\'{\i}}n Kucera, editors, {\em
  49th International Symposium on Mathematical Foundations of Computer Science,
  {MFCS} 2024, August 26-30, 2024, Bratislava, Slovakia}, volume 306 of {\em
  LIPIcs}, pages 48:1--48:14. Schloss Dagstuhl - Leibniz-Zentrum f{\"{u}}r
  Informatik, 2024.

\bibitem{DruckerKO13}
Andrew Drucker, Fabian Kuhn, and Rotem Oshman.
\newblock On the power of the congested clique model.
\newblock In Magn{\'{u}}s~M. Halld{\'{o}}rsson and Shlomi Dolev, editors, {\em
  {ACM} Symposium on Principles of Distributed Computing, {PODC} '14, Paris,
  France, July 15-18, 2014}, pages 367--376. {ACM}, 2014.

\bibitem{Emden-WeinertHK98}
Thomas Emden{-}Weinert, Stefan Hougardy, and Bernd Kreuter.
\newblock Uniquely colourable graphs and the hardness of colouring graphs of
  large girth.
\newblock {\em Comb. Probab. Comput.}, 7(4):375--386, 1998.

\bibitem{EL}
Louis Esperet and Benjamin L{\'{e}}v{\^{e}}que.
\newblock Local certification of graphs on surfaces.
\newblock {\em Theoretical Computer Science}, 909:68--75, 2022.

\bibitem{Feuilloley21}
Laurent Feuilloley.
\newblock Introduction to local certification.
\newblock {\em Discret. Math. Theor. Comput. Sci.}, 23(3), 2021.

\bibitem{BFT2}
Laurent Feuilloley, Nicolas Bousquet, and Th{\'{e}}o Pierron.
\newblock What can be certified compactly? {C}ompact local certification of
  {MSO} properties in tree-like graphs.
\newblock In Alessia Milani and Philipp Woelfel, editors, {\em {PODC} '22:
  {ACM} Symposium on Principles of Distributed Computing, Salerno, Italy, July
  25 - 29, 2022}, pages 131--140. {ACM}, 2022.

\bibitem{planar}
Laurent Feuilloley, Pierre Fraigniaud, Pedro Montealegre, Ivan Rapaport,
  {\'{E}}ric R{\'{e}}mila, and Ioan Todinca.
\newblock Compact distributed certification of planar graphs.
\newblock {\em Algorithmica}, 83(7):2215--2244, 2021.

\bibitem{genus}
Laurent Feuilloley, Pierre Fraigniaud, Pedro Montealegre, Ivan Rapaport,
  {\'{E}}ric R{\'{e}}mila, and Ioan Todinca.
\newblock Local certification of graphs with bounded genus.
\newblock {\em Discrete Applied Mathematics}, 325:9--36, 2023.

\bibitem{FM0RT23}
Pierre Fraigniaud, Fr{\'{e}}d{\'{e}}ric Mazoit, Pedro Montealegre, Ivan
  Rapaport, and Ioan Todinca.
\newblock Distributed certification for classes of dense graphs.
\newblock In Rotem Oshman, editor, {\em 37th International Symposium on
  Distributed Computing, {DISC} 2023, October 10-12, 2023, L'Aquila, Italy},
  volume 281 of {\em LIPIcs}, pages 20:1--20:17. Schloss Dagstuhl -
  Leibniz-Zentrum f{\"{u}}r Informatik, 2023.

\bibitem{tw}
Pierre Fraigniaud, Pedro Montealegre, Ivan Rapaport, and Ioan Todinca.
\newblock A meta-theorem for distributed certification.
\newblock In Merav Parter, editor, {\em Structural Information and
  Communication Complexity - 29th International Colloquium, {SIROCCO} 2022,
  Paderborn, Germany, June 27-29, 2022, Proceedings}, volume 13298 of {\em
  Lecture Notes in Computer Science}, pages 116--134. Springer, 2022.

\bibitem{GareyJS76}
Michael~R. Garey, David~S. Johnson, and Larry~J. Stockmeyer.
\newblock Some simplified {NP}-complete graph problems.
\newblock {\em Theor. Comput. Sci.}, 1(3):237--267, 1976.

\bibitem{GoosS16}
Mika G{\"{o}}{\"{o}}s and Jukka Suomela.
\newblock Locally checkable proofs in distributed computing.
\newblock {\em Theory Comput.}, 12(1):1--33, 2016.

\bibitem{HK01}
Petr Hlinen{\'{y}} and Jan Kratochv{\'{\i}}l.
\newblock Representing graphs by disks and balls (a survey of
  recognition-complexity results).
\newblock {\em Discret. Math.}, 229(1-3):101--124, 2001.

\bibitem{Holyer81}
Ian Holyer.
\newblock The {NP}-completeness of edge-coloring.
\newblock {\em {SIAM} J. Comput.}, 10(4):718--720, 1981.

\bibitem{MR14}
Michael Molloy and Bruce Reed.
\newblock Colouring graphs when the number of colours is almost the maximum
  degree.
\newblock {\em J. Comb. Theory, Ser. B}, 109:134--195, 2014.

\bibitem{NPY20}
Moni Naor, Merav Parter, and Eylon Yogev.
\newblock The power of distributed verifiers in interactive proofs.
\newblock In Shuchi Chawla, editor, {\em Proceedings of the 2020 {ACM-SIAM}
  Symposium on Discrete Algorithms, {SODA} 2020, Salt Lake City, UT, USA,
  January 5-8, 2020}, pages 1096--115. {SIAM}, 2020.

\bibitem{Papadimitriou76}
Christos~H. Papadimitriou.
\newblock The np-completeness of the bandwidth minimization problem.
\newblock {\em Computing}, 16(3):263--270, 1976.

\bibitem{Sipser}
Michael Sipser.
\newblock {\em Introduction to the theory of computation}.
\newblock {PWS} Publishing Company, 1997.

\bibitem{Stewart94}
Iain~A. Stewart.
\newblock Deciding whether a planar graph has a cubic subgraph is
  {NP}-complete.
\newblock {\em Discret. Math.}, 126(1-3):349--357, 1994.

\bibitem{Stewart97}
Iain~A. Stewart.
\newblock On locating cubic subgraphs in bounded-degree connected bipartite
  graphs.
\newblock {\em Discret. Math.}, 163(1-3):319--324, 1997.

\end{thebibliography}

\end{document}